\documentclass[11pt]{article}

\usepackage[latin1]{inputenc}
\usepackage[T1]{fontenc}
\usepackage{color}
\usepackage{eurosym}
\usepackage[top=4cm,bottom=4cm,left=3cm,right=3cm]{geometry}
\usepackage{fancyhdr}
\usepackage{float}
\usepackage{amssymb,amsmath, amsthm}
\usepackage{url,graphicx}

\usepackage{booktabs}
\usepackage{stmaryrd} 
\usepackage{mathenv}
\usepackage{dsfont}
\usepackage{amssymb}
\usepackage{latexsym}
\usepackage{amsmath}
\usepackage{color}
\usepackage{comment}
\usepackage{amsthm}
\usepackage{bm}
\usepackage{bbm} 
\usepackage[bbgreekl]{mathbbol} 
\usepackage{authblk}
\usepackage{multirow}
\usepackage{enumerate}   
\newif\ifrs
\rstrue
\ifrs \usepackage{mathrsfs} \fi  
\newif\ifcol
\coltrue 

\newtheorem{theorem*}{Theorem}[section]

\newtheorem{note*}[theorem*]{Note}
\newtheorem{lemma*}[theorem*]{Lemma}
\newtheorem{definition*}[theorem*]{Definition}
\newtheorem{proposition*}[theorem*]{Proposition}
\newtheorem{corollary*}[theorem*]{Corollary}
\newtheorem{remark*}[theorem*]{Remark}
\newtheorem{example*}[theorem*]{Example}
\newtheorem*{conditionE'3}{Condition (E3')}

\numberwithin{equation}{section}
\newif\ifcol
\coltrue
\ifcol
\newcommand{\colorr}{\color[rgb]{0.8,0,0}}

\newcommand{\colorn}{\color[rgb]{1,1,1}}

\else

\newcommand{\colorr}{\color{black}}

\newcommand{\colorn}{\color{black}}

\fi
\excludecomment{en-text}
\includecomment{jp-text}
\includecomment{comment}
\newtheorem{remark}{Remark}

\def\bd{\begin{description}}
\def\ed{\end{description}}

\def\D2{\bbD_{2,\infty-}}

\def\D{{\bf D}}

\def\F{{\bf F}}
\def\G{{\bf G}}

\def\I{{\bf I}}
\def\J{{\bf J}}
\def\K{{\bf K}}

\def\calb{{\cal B}}

\def\cale{{\cal E}}
\def\calf{{\cal F}}
\def\calg{{\cal G}}
\def\calh{{\cal H}}

\def\calk{{\cal K}}
\def\call{{\cal L}}
\def\calm{{\cal M}}
\def\caln{{\cal N}}

\def\calq{{\cal Q}}

\def\calu{{\cal U}}

\def\half{\frac{1}{2}}

\def\be{\begin{equation}}
\def\ee{\end{equation}}
\def\bea{\begin{eqnarray}}
\def\eea{\end{eqnarray}}
\def\beas{\begin{eqnarray*}}
\def\eeas{\end{eqnarray*}}
\def\bi{\begin{itemize}}
\def\ei{\end{itemize}}

\def\bd{\begin{description}}
\def\ed{\end{description}}
\def\l{\left}
\def\r{\right}

\newcommand{\bbD}{{\mathbb D}}

\newcommand{\reels}{\mathbb{R}}
\newcommand{\naturels}{\mathbb{N}}

\newcommand{\esp}{\mathbb{E}}
\newcommand{\proba}{\mathbb{P}}

\newcommand{\inv}[1]{\frac{1}{#1}}

\newcommand{\var}{\operatorname{Var}}

\newcommand{\fru}{\frac{u_0}{\sqrt{N_n} }}
\newcommand{\fruu}{\frac{u_0^2}{2N_n}}

\newcommand{\bo}[1]{O\l(#1\r)} 
 
\newcommand{\bop}[1]{O_\proba\l(#1\r)} 
\newcommand{\lop}[1]{o_\proba\l(#1\r)} 

\newcommand{\ti}[1]{t_{#1}^n}

\newcommand{\espu}{\esp_\calu}

\newcommand{\qterm}{T^{-1}\int_0^T{\alpha_s^{-1}ds}\l\{\int_0^T{\sigma_s^4 \alpha_sds} + \sum_{0 < s \leq T} \Delta J_s^2 (\sigma_s^2\alpha_s + \sigma_{s-}^2\alpha_{s-}) \r\}}

\theoremstyle{plain}

\theoremstyle{remark}

\begin{document}
 \title{Testing if the market microstructure noise is fully explained by the informational content of some variables from the limit order book\footnote{We would like to thank Yingying Li, Xinghua Zheng, Dacheng Xiu, Viktor Todorov, Torben Andersen, Rasmus Varneskov, Rui Da, Yacine A\"{i}t-Sahalia (the Editor), two anonymous referees and an anonymous Associate Editor, the participants of ECOSTA 2018, 2018 Asian meeting of the Econometric Society, SoFiE Financial Econometrics Summer School 2017 at Kellogg School of Management and 2017 Asian meeting of the Econometric Society at The Chinese University of Hong Kong for helpful discussions and advice. The research of Yoann Potiron is supported by Japanese Society for the Promotion of Science Grant-in-Aid for Young Scientists (B) No. 60781119 and a special grant from Keio University. The research of Simon Clinet is supported by CREST Japan
Science and Technology Agency and a special grant from Keio University.}}
\author{Simon Clinet\footnote{Faculty of Economics, Keio University. 2-15-45 Mita, Minato-ku, Tokyo, 108-8345, Japan. Phone:  +81-3-5427-1506. E-mail: clinet@keio.jp website: http://user.keio.ac.jp/\char`\~ clinet} \footnote{CREST, Japan Science and Technology Agency, Japan.}  and Yoann Potiron\footnote{Faculty of Business and Commerce, Keio University. 2-15-45 Mita, Minato-ku, Tokyo, 108-8345, Japan. Phone:  +81-3-5418-6571. E-mail: potiron@fbc.keio.ac.jp website: http://www.fbc.keio.ac.jp/\char`\~ potiron}}
\date{This version: \today}

\maketitle

\begin{abstract}
In this paper, we build tests for the presence of residual noise in a model where the market microstructure noise is a known parametric function of some variables from the limit order book. The tests compare two distinct quasi-maximum likelihood estimators of volatility, where the related model includes a residual noise in the market microstructure noise or not. The limit theory is investigated in a general nonparametric framework. In the presence of residual noise, we examine the central limit theory of the related quasi-maximum likelihood estimation approach. 

\end{abstract}

\section{Introduction}
If one can sample directly from the efficient price, the estimation of volatility is a well-studied matter. The realized volatility (RV) estimator, i.e. summing the square of the log-returns, is both consistent and efficient. However, in practice the observed price does not behave as expected. When sampling at high frequency, it can be quite different from the efficient price due to bid-ask bounce mechanism, the spread, the fact that transactions lie on a tick grid, etc. Market microstructure noise (MMN) typically degrades RV to the extent that it is highly biased when performed on tick by tick data. 

\smallskip
One approach to overcome the problem consists in sub-sampling, say every 5 minutes, as in the pioneer work from \cite{andersen2001distribution} and \cite{barndorff2002estimating}. On the contrary, \cite{ait2005often} recognize the MMN as an inherent part of the data and advice for the use of the Quasi-Maximum Likelihood Estimator (QMLE) which was later shown to be robust to time-varying volatility in \cite{xiu2010quasi}. Concurrent methods include and are not limited to: the Two-Scale Realized Volatility (TSRV) in \cite{zhang2005tale}, the Multi-Scale Realized Volatility in \cite{zhang2006efficient}, the Pre-averaging approach (PAE) in \cite{jacod2009microstructure}, realized kernels (RK) in \cite{barndorff2008designing} and the spectral approach considered in \cite{altmeyer2015functional}. 

\smallskip
Those two approaches have cons in that the sub-sampling technique discards a large proportion of the data and the noise-robust estimators have a slower rate of convergence than RV in the absence of noise. The only exception is the QMLE proposed by \cite{da2017moving}, which is automatically adaptive to noise magnitude and always enjoys optimal rate. In this paper, we consider the MMN as a salient part of the financial market, but using the growing limit order book (LOB) big data available to the econometrician we ask a different question. Can we test whether the MMN is fully explained by the informational content of some variables from the limit order book and can we estimate the parameter of the model? If the MMN can be expressed as an observable function, then we can estimate the efficient price and use RV including all the data points. This idea is not new, and actually our work is heavily based on the two very nice papers \cite{li2016efficient} and \cite{chaker2017high}. We will explain the differences later in the introduction.

\smallskip
In fact it is rather natural to see the MMN as a function of some variables of the LOB as in the pioneer work from \cite{roll1984simple} where the trade type, i.e. whether the trade was buyer or seller initiated, is used to correct for the bid-ask bounce effect in the observed price. The related observed price $Z_{t_i}$ is then defined as
\begin{eqnarray}
\label{rollmodel}
\underbrace{Z_{t_i}}_{\text{observed price}} & = & \underbrace{X_{t_i}}_{\text{efficient price}} + \underbrace{I_{i} \theta_0}_{\text{MMN}},
\end{eqnarray}
where $\theta_0$ can be interpreted as one-half of the effective bid-ask spread, $I_{i}$ is equal to 1 if the trade at time $t_i$ is buyer-initiated and -1 if seller-initiated. A simple extension where the spread $S_i$ is time-varying is given by 
\begin{eqnarray}
\label{spreadmodel}
Z_{t_i} & = & X_{t_i} + \frac{1}{2} I_{i} S_i \theta_0.
\end{eqnarray}
Discussion and related leading models can be found in: \cite{black1986noise}, \cite{hasbrouck1993assessing}, \cite{o1995market}, \cite{madhavan1997security}, \cite{madhavan2000market}, \cite{stoll2000presidential} and \cite{hasbrouck2007empirical} among other prominent work. 

\smallskip
The question we address is: can we trust models such as (\ref{rollmodel}) or (\ref{spreadmodel})? To investigate it, we introduce the general set-up as
\begin{eqnarray}
\label{genmodel}
\underbrace{Z_{t_i}}_{\text{observed price}} & = & \underbrace{X_{t_i}}_{\text{efficient price}} + \underbrace{\underbrace{\phi(Q_i, \theta_0)}_{\text{explicative part}} + \underbrace{\epsilon_{t_i}}_{\text{ residual noise}}}_{\text{MMN}},
\end{eqnarray}
where $Q_i$ are observable variables included in the LOB while $\phi$ is known to the econometrician, and we develop tests for the presence of the  residual noise $\epsilon_{t_i}$ at any given sampling frequency. The associated null hypothesis is such that $\{ \epsilon_{t_i} = 0, \phi = \Phi \}$ and the alternative $\{\var [\epsilon_{t_i}] > 0, \phi = \Phi \text{ or } \phi = 0\}$, where $\Phi := \Phi (Q_i, \theta_0 ) \neq 0$ is known up to the parameter $\theta_0$.

\smallskip
 Our tests are based on \cite{hausman1978specification} tests\footnote{As far as we know, the use of Hausman tests in high-frequency data can be traced back to the TSRV and \cite{huang2005relative}.} developed in  \cite{ait2016hausman}, which is restricted to the case $\Phi=0$. The authors consider the difference $\widehat{\sigma}_{RV}^2 -\widehat{\sigma}_{QMLE}^2$, where $\widehat{\sigma}_{RV}^2 = T^{-1} \sum_i (Z_{t_i} - Z_{t_{i-1}})^2$ and $\widehat{\sigma}_{QMLE}^2$ corresponds to the QMLE in a model where $\phi=0$. The Hausman test statistic is of the form $H = n(\widehat{\sigma}_{RV}^2 -\widehat{\sigma}_{QMLE}^2)^2/\widehat{V}$, where $\widehat{V}$ is an estimator of $AVAR(\widehat{\sigma}_{RV}^2 -\widehat{\sigma}_{QMLE}^2)$ under the null hypothesis. Under the alternative, RV is not consistent whereas the QMLE stays consistent so that the authors show that $H$ explodes in that case.

\smallskip
To test the presence of  residual noise in (\ref{genmodel}) in the case where $\phi \neq 0$, we consider the Hausman tests comparing two distinct QMLE related to a model including the explicative part, i.e. $\widehat{\sigma}_{exp}^{2}$ which is restricted to a null residual noise and $\widehat{\sigma}_{err}^{2}$ including the  residual noise. As far as the authors know, $\widehat{\sigma}_{err}^{2}$ is novel in the particular context of high frequency data. They respectively play the role of $\widehat{\sigma}_{RV}^2$ and $\widehat{\sigma}_{QMLE}^2$, but they are both distinct from the latter related to a model with $\phi=0$. Note that \cite{ait2016hausman} consider other candidates for testing, including the PAE, that we set aside in this paper.

\smallskip
The estimator $\widehat{\sigma}_{exp}^{2}$ corresponds exactly to estimators considered in  \cite{li2016efficient} (the so-called estimated-price RV in the latter paper) and \cite{chaker2017high} (although the latter work is restricted to a linear $\phi$). Moreover, the E-QMLE discussed in \cite{li2016efficient} (Section 2.2.1), i.e. first estimating the price and then applying the usual QMLE to the estimates,  is asymptotically equivalent to $\widehat{\sigma}_{err}^{2}$. In addition, \cite{chaker2017high} actually provides tests from a different nature for the presence of residual noise when $\phi$ is linear. Finally, some extensions are considered in Section 4.4 of \cite{potiron2016local}.

\smallskip
Our first main theoretical contribution includes the investigation of the joint limit theory of $(\widehat{\sigma}_{exp}^{2}, \widehat{\sigma}_{err}^{2}, \widehat{a}_{err}^2, \widehat{\theta}_{exp}, \widehat{\theta}_{err})$, where $\widehat{a}_{err}^2$ is an estimator of the residual noise and $(\widehat{\theta}_{exp}, \widehat{\theta}_{err})$ are estimates of the parameter both obtained via QMLE under small residual noise, i.e. $\var [\epsilon_{t_i}] = O (1/n)$. The marginal limit theory for $\widehat{\sigma}_{exp}^{2}$ boils down to Theorem 2 in \cite{li2016efficient} as $\widehat{\sigma}_{exp}^{2}$ is equal to the estimated-price RV estimator, and to Theorem 4 (i) in \cite{chaker2017high} when $\phi$ is linear. In addition, used in conjunction with the toolkit in \cite{ait2016hausman}, one could easily obtain an asymptotically equivalent joint limit theory of $(\widehat{\sigma}_{exp}^{2}, \widehat{\sigma}_{err}^{2}, \widehat{a}_{err}^2)$ as the E-QMLE is asymptotically equivalent to $\widehat{\sigma}_{err}^{2}$. However the main difference between our setting and that of \cite{ait2016hausman} is that we have stochastic observation times whereas the cited authors only consider regular sampling times. In particular, in case of regular observation times, our contribution boils down to the marginal and joint limits for $(\widehat{\theta}_{exp}, \widehat{\theta}_{err})$. We further demonstrate that only $\widehat{\sigma}_{err}^{2}$ is residual noise robust so that we can consider the corresponding Hausman statistic to test the presence of residual noise. 

\smallskip
When there is no  residual noise in the model, i.e. $\epsilon_{t_i} = 0$, a byproduct of our contribution is that the parameter estimators are asymptotically equivalent. Subsequently, following the procedure considered in \cite{chaker2017high} and \cite{li2016efficient}, we can consistently estimate the efficient price directly from the data as
\begin{eqnarray}
\label{efficientpriceestimator}
\widehat{X}_{t_i} = Z_{t_i} -  \phi(Q_i, \widehat{\theta}_{exp}).
\end{eqnarray}
This procedure seems to be traced back to the model with uncertainty zones, which was introduced in \cite{robert2010new} and \cite{robert2012volatility}. See also  the pioneer work from \cite{hansen2006realized} and the more recent work from \cite{andersen2017volatility} for efficient price estimation, although in a slightly different context. 

\smallskip
When we assume a residual noise in the model, we examine the measure of goodness of fit introduced in \cite{li2016efficient}, which corresponds to the proportion of MMN variance explained by the explicative part. Such measure can be estimated using the parameter and the  residual noise variance estimates obtained with the QMLE related to the model including the  residual noise. Our second main contribution establishes the corresponding central limit theorem in case when the variance of the  residual noise stays constant. This goes one step further than Theorem 3 from \cite{li2016efficient} in that the noise variance does not shrink to 0 asymptotically, and that we can actually provide a more reliable  residual noise variance estimator, along with the asymptotic theory. Also, the convergence rate is smaller than the pre-estimation (\ref{efficientpriceestimator})-TSRV approach considered in \cite{chaker2017high} (see Theorem 4 (ii)). In particular, volatility estimation is naturally not as fast as when we assume small noise in the model.

\smallskip
We implement the tests over a one month period with tick by tick data, and find out that the linear signed spread model (\ref{spreadmodel}) consistently stands out from many other alternatives including Roll model (\ref{rollmodel}). The tests further reveal that the large majority of stocks can be reasonably considered as free from  residual noise with such model. Moreover, we implement the tests from \cite{ait2016hausman} regarding the estimated efficient price (\ref{efficientpriceestimator}) as the given observed price. They largely corroborate the findings.

\smallskip
As far as we know, there are at least another paper in volatility estimation closely related to our work. The impact of $\phi$ on RV is thoroughly discussed in \cite{diebold2013correlation}. In that paper, the authors study several leading models from the market microstructure literature. Unfortunately, their assumption of constant volatility is quite strong.

\smallskip
The remainder of the paper is structured as follows. The model is introduced in Section 2. The limit theory of the QMLE under small noise, the Hausman tests and the efficient price estimator are developed in Section 3. We discuss about measure of goodness of fit estimation, central limit theory under large noise and guidance for implementation of volatility estimation in Section 4. Section 5 performs a Monte Carlo experiment to assess finite sample performance of the tests and validation of the sequence to estimate volatility. Section 6 is devoted to an empirical study. We conclude in Section 7. Theoretical details and proofs can be found in the Appendix.

\section{Model}
\label{model}
For a given horizon time $T >0$, we make observations\footnote{All the considered quantities are implicitly or explicitly indexed by $n$. Consistency and convergence in law refer to the behavior as $n \rightarrow \infty$. A full specification of the model actually involves the stochastic basis $\calb=(\Omega,\proba,\calf,\F)$, where $\calf$ is a $\sigma$-field and $\F=(\calf_t)_{t\in[0,T]}$ is a filtration. We assume that all the processes are $\F$-adapted (either in a continuous or discrete meaning) and that the observation times $t_i$ are $\F$-stopping times. Also, when referring to It\^{o}-semimartingale, we automatically mean that the statement is relative to $\F$.} contaminated by the MMN at (possibly random) times $0 = t_0 \leq ...\leq t_{N} \leq T $ of the efficient log-price $X_t$, and we assume that we have the additive decomposition
\begin{eqnarray*}
\underbrace{Z_{t_i}}_{\text{observed price}}  =  \underbrace{X_{t_i}}_{\text{efficient price}} + \underbrace{\underbrace{\phi(Q_i, \theta_0)}_{\text{explicative part}} + \underbrace{\epsilon_{t_i}}_{\text{ residual noise}}}_{\text{MMN}}.
\label{eqDecompositionAdditive}
\end{eqnarray*}
Here the parameter $\theta_0 \in \Theta \subset \reels^d$, where $\Theta$ is a compact set. The impact function $\phi$ is known, of class $C^m$ in $\theta$ with $m > d/2 + 2 $, and $\epsilon_{t_i}$ corresponds to the remaining noise. Finally, $Q_i \in \reels^q$ includes  \textit{observable} information from the LOB such as the trade type $I_i$, the trading volume $V_i$ (\cite{glosten1988estimating}), the duration time between two trades $D_i$ (\cite{almgren2001optimal}), the quoted depth\footnote{The ask (bid) depth specifies the volume available at the best ask (bid)} $QD_i$ (\cite{kavajecz1999specialist}), the bid-ask spread $S_i$, the order flow imbalance\footnote{It is defined as the imbalance between supply and demand at the best bid and ask prices (including both quotes and cancellations)} $OFI_i$ (\cite{cont2014price}). The introduced MMN complies with the empirical evidence about autocorrelated noise\footnote{Although not with endogenous and/or heteroskedastic noise.} (see, e.g., \cite{kalnina2008estimating} and \cite{ait2011ultra}). Some examples of $\phi$ can be consulted on Table \ref{tableoverview}.
\subsubsection*{The efficient price}
The latent log-price $X_t$ is an It\^{o}-semimartingale of the form 
\begin{eqnarray}
\label{efficientPrice} dX_t & = & b_t dt + \sigma_t dW_t + dJ_t,\\
d \sigma_t & = & \widetilde{b}_t dt + \widetilde{\sigma}_t^{(1)} dW_t + \widetilde{\sigma}_t^{(2)} d\widetilde{W}_t + d\widetilde{J}_t,
\end{eqnarray}
with $(W_t, \widetilde{W}_t)$ which is a 2 dimensional standard Brownian motion, the drift $(b_t,\widetilde{b}_t)$ which is componentwise locally bounded, $(\sigma_t, \widetilde{\sigma}_t^{(1)}, \widetilde{\sigma}_t^{(2)})$\footnote{A nice review on the use of stochastic volatility in financial mathematics can be found in \cite{ghysels19965}.}
which is componentwise locally bounded, itself an It\^{o} process and $\inf_t (\min (\sigma_t, \widetilde{\sigma}_t^{(2)})) > 0$ a.s. We also assume that $(J_t,\widetilde{J}_t)$ is a 2 dimensional pure jump process\footnote{Jumps in volatility are a salient part of the data (see, e.g., \cite{todorov2011volatility} for empirical evidence.)} of finite activity.  

\subsubsection*{The observation times}
Crucial to the estimation is the robustness of the procedure when considering tick-time volatility instead of calendar-time volatility. For instance, \cite{patton2011data} (see, e.g., p. 299) compares empirically the accuracy of estimators and mentions that using tick-time sampling leads to more accurate volatility estimation, although the considered estimators are a priori not robust to such sampling procedure. Also, \cite{xiu2010quasi} and \cite{ait2016hausman} (Section 5, p. 17) compute the likelihood estimators estimating tick-time volatility even though the theory only covers the regular observation times framework.

\smallskip
We introduce the notation $\Delta := T/n$. We consider the random discretization scheme used in \cite{clinet2018efficient} (Section 4) and adapted from \cite{jacod2011discretization} (see Section 14.1). We assume that there exists an It\^{o}-semimartingale $\alpha_t > 0$ which satisfies Assumption 4.4.2 p. 115 in \cite{jacod2011discretization} and is locally bounded and locally bounded away from $0$, and i.i.d $U_i > 0$ that are independent with each other and from other quantities such that 
\bea 
t_0 & = & 0,\\
t_i & = & t_{i-1} + \Delta \alpha_{t_{i-1}}U_i.
\label{specIrregularGrid}
\eea
We further assume that $\esp U_i = 1$, and that for any $q >0$, $m_{q} := \esp U_i^q < \infty$, is independent of $n$. If we define $\pi_t := \sup_{i \geq 1} t_{i}-t_{i-1}$ and the number of observations before $t$ as $N(t) = \sup \{ i\in \naturels | 0 < t_{i}  \leq t \}$ we have that $\pi_t \to^\proba 0$ and that\footnote{Actually the convergence is $u.c.p$, i.e. uniformly in probability on $[0,t]$ for any $t \in [0,T]$. Equation (\ref{eqNumberJumps}) can be shown using Lemma 14.1.5 in \cite{jacod2011discretization}. The uniformity is obtained as a consequence of the fact that $N_n$ and $\int_0^.{\inv{\alpha_s}ds}$ are increasing processes and Property (2.2.16) in \cite{jacod2011discretization}.}
\bea 
 \frac{N(t)}{n} \to^{\proba} \inv{T}\int_0^t{\inv{\alpha_s}ds}.
 \label{eqNumberJumps}
 \eea 
When there is no room for confusion, we sometimes drop $T$ in the expression, i.e we use $N := N(T)$. 

\subsubsection*{The information}
 Given the process $X_t$, the observed information $Q_i$ is assumed to be conditionally stationary, i.e. for any $k$, $j$, $i_1$, $\cdots$, $i_k \in \naturels$ and for any continuous and bounded function $f$ we have
\bea 
\esp \l[\l. f(Q_{i_1+j},...,Q_{i_k+j}) \r| X \r] = \esp \l[\l. f(Q_{i_1},...,Q_{i_k}) \r| X \r] \textnormal{  } \textnormal{a.s.}
\eea
We introduce the difference between the explicative part taken in $\theta$ and in $\theta_0$ as
\bea 
W_i(\theta) := \phi(Q_i,\theta)-\phi(Q_i,\theta_0),
\label{defW}
\eea 
and for any $i,j,k,l \in \naturels$, and for any multi-indices $\boldsymbol{q} = (q_1,q_2)$, $\boldsymbol{r} = (r_1,r_2,r_3,r_4)$, where the subcomponents of $\boldsymbol q$ and $\boldsymbol r$ are themselves $d$ dimensional multi-indices, the following quantities conditioned on the price process
\bea 
\esp \l[\l. W_i(\theta) \r| X \r] &=& 0 \textnormal{   }\textnormal{a.s} ,
\label{assMean}\\
\rho_{j}^{\boldsymbol{q}}(\theta) &:= &\esp \l[ \l.  \frac{\partial^{q_1}W_{i}(\theta)}{\partial \theta^{q_1}}\frac{\partial^{q_2}W_{i+j}(\theta)}{\partial \theta^{q_2}} \r|X \r] = \esp \l[ \frac{\partial^{q_1}W_{i}(\theta)}{\partial \theta^{q_1}}\frac{\partial^{q_2}W_{i+j}(\theta)}{\partial \theta^{q_2}} \r] \textnormal{   }\textnormal{a.s},
\label{defRho}\\
\kappa_{j,k,l}^{\boldsymbol{r}}(\theta) &:=& \textnormal{cum}\l[\l. \frac{\partial^{r_1}W_{i}(\theta)}{\partial \theta^{r_1}},\frac{\partial^{r_2}W_{i+j}(\theta)}{\partial \theta^{r_2}}, \frac{\partial^{r_3}W_{i+k}(\theta)}{\partial \theta^{r_3}}, \frac{\partial^{r_4}W_{i+l}(\theta)}{\partial \theta^{r_4}}  \r| X \r] \nonumber \\ &=& \textnormal{cum}\l[ \frac{\partial^{r_1}W_{i}(\theta)}{\partial \theta^{r_1}},\frac{\partial^{r_2}W_{i+j}(\theta)}{\partial \theta^{r_2}}, \frac{\partial^{r_3}W_{i+k}(\theta)}{\partial \theta^{r_3}}, \frac{\partial^{r_4}W_{i+l}(\theta)}{\partial \theta^{r_4}}   \r] \textnormal{   } \textnormal{a.s}, \label{defKappa}
\eea 
where $\rho_{j}^{\boldsymbol{q}}(\theta)$ and $\kappa_{j,k,l}^{\boldsymbol{r}}(\theta)$ are assumed independent of $n$. Note that conditions (\ref{assMean})-(\ref{defKappa}) state that conditional moments of the information process  (and its derivatives with respect to $\theta$) up to the fourth order are independent of the efficient price. This is weaker than assuming the independence of $Q$ and $X$ (and thus it is weaker than the classical QMLE framework of \cite{xiu2010quasi} where the MMN is assumed independent of $X$). When $\boldsymbol{q}=0$ (respectively $\boldsymbol{r}=0$), we refer directly to $\rho_j(\theta)$ (respectively $\kappa_{j,k,l}(\theta)$) in place of $\rho_j^{\boldsymbol{q}}(\theta)$ (respectively $\kappa_{j,k,l}^{\boldsymbol{q}}(\theta)$). To ensure the weak dependence of the information over time and the identifiability of $\theta_0$, we also assume for any $i = 0, \cdots,m$ and $0 \leq |\boldsymbol{q}|, |\boldsymbol{r}| \leq m$ the following set of conditions:
\bea 
\sup_{\theta \in \Theta }\sum_{j=0}^{+\infty}{\l| \rho_{j}^{\boldsymbol q}(\theta)
\r|} & < & \infty \textnormal{   }\textnormal{a.s}, 
\label{assCorr}\\
\sup_{\theta \in \Theta }\sum_{j,k,l=0}^{+\infty}{\l| \kappa_{j,k,l}^{\boldsymbol r}(\theta) 
\r|} & < & \infty \textnormal{   } \textnormal{a.s}, 
\label{assCumulant}\\ 
\esp \l[\l.\sup_{\theta \in \Theta}\l|\frac{\partial^j \mu_i(\theta)}{\partial \theta^j}\r|^p \r| X\r] &<& \infty \textnormal{  a.s, for any } p \geq 1,\textnormal{ } 0 \leq j \leq 2, 
\label{assMoment}\\
\frac{\partial \rho_0(\theta)}{\partial \theta} = 0 & \Leftrightarrow & \theta = \theta_0. \label{assPartial}
\eea 


\begin{remark}
Conditions (\ref{assCorr})-(\ref{assCumulant}) ensure the weak dependence over time of the information process whereas Condition (\ref{assPartial}) implies the identifiability of $\theta_0$ for the QMLE. They are needed in order to derive the limit theory of the QMLE estimators related to $\theta_0$ that are defined in the next section. Note that we consider a setting where the information process is stationary when conditioned on the efficient price process, which was not assumed in \cite{li2016efficient}. In particular conditions (\ref{assCorr})-(\ref{assCumulant}) are stronger forms for stationary sequences of Condition (A.xi), while (\ref{assPartial}) replaces the identifiability assumption (A.x) in their paper. The need for stronger assumptions is due to the fact that in this work, in addition to the consistency with rate of convergence $N$, we also prove the central limit theory for the QMLE related to $\theta_0$. On the other hand, the moment condition (\ref{assMoment}) is weaker than the quite strong assumption (A.v) requiring that the information process is uniformly stochastically bounded.  
\end{remark}

\subsubsection*{The residual noise}

The remaining noise is assumed independent of all the other processes, i.i.d with $\esp[\epsilon_t]=0$ and $\esp[\epsilon_t^2] = a_0^2 >0$, and with finite fourth moment. 

\begin{remark}
Given the assumptions on the information process $Q_i$ and on the residual noise $\epsilon_{t_i}$, we have ruled out the case of an heteroskedastic MMN. Although empirical evidence indicates time dependence of the MMN (as pointed out in, e.g, \cite{hansen2006realized}), incorporating heteroskedasticity in our model is beyond the scope of this paper. Note also that we only allow for a weak form of endogeneity for the explicative part $\phi(Q_i,\theta)$ (its conditional moments of order $4$ or less should not depend on $X_t$). Again, we set aside stronger forms of endogeneity in this paper. Nevertheless, we have considered an endogenous and heteroskedastic residual noise in our simulation study and shown that the tests seem reasonably robust to such misspecification.  
\end{remark}
\section{Tests for the presence of residual noise} \label{sectionNoNoise}
\subsection{Small noise alternative case}
We first consider the simple semiparametric model where $X_t = \sigma_0 W_t$,  the observations are regular $t_{i+1} - t_i = \Delta$ which implies that $N = n$, the residual noise $\epsilon_{t_i}$ is normally distributed  with zero-mean and variance $a_0^2$. We further define $\Delta_N = T/N$ which in this simple model satisfies $\Delta_N = \Delta$. The null hypothesis is defined as $\mathcal{H}_0 : \{a_0^2=0, \phi = \Phi\}$ whereas the alternative is defined as $\mathcal{H}_1 : \{a_0^2:= \eta_0/n >0, \phi = \Phi\}$, where $\Phi := \Phi (Q_i, \theta_0 ) \neq 0$ and $\eta_0$ is a constant which does not depend on $n$. The cases of large noise alternative and $\phi = 0$ alternative are respectively delayed to Section \ref{largenoisealternativecase} and Section \ref{phi0alternativecase}. To ensure that our method is robust to general information, our strategy consists in considering two distinct likelihood functions conditioned on the information. We define the observed log returns $Y_i = Z_{t_i} - Z_{t_{i-1}}$, $Y = (Y_1, \cdots, Y_{N})^T$. Moreover, the returns of information are denoted by $\mu_i(\theta) = \phi(Q_{i}, \theta) - \phi(Q_{{i-1}}, \theta)$, $\mu(\theta) = (\mu_1(\theta), \cdots,\mu_N(\theta))^T$ and we further define $\widetilde{Y}(\theta) = Y - \mu(\theta)$. Key to our analysis is that $\widetilde{Y}(\theta)$ is known to the econometrician.

\smallskip
In the absence of residual noise, the observed returns can be expressed as 
\begin{eqnarray}
Y_i = \sigma_0 (W_{t_i} - W_{t_{i-1}}) + \mu_i (\theta_0). 
\end{eqnarray}
It is then clear that $\widetilde{Y} (\theta_0)$ is i.i.d normally distributed centered with variance $\sigma_0^2 \Delta_N$ and the log-likelihood can be expressed as
\begin{eqnarray}
\label{loglikexp}
l_{exp}(\sigma^2, \theta) = - \frac{N}{2} \textnormal{log}(\sigma^2 \Delta_N ) - \frac{N}{2} \log (2\pi) -\frac{1}{2\sigma^2 \Delta_N} \widetilde{Y}(\theta)^T  \widetilde{Y}(\theta).
\end{eqnarray}

When the residual noise is present, \cite{ait2005often} show that in the case where there is no information, i.e. 
\begin{eqnarray}
Y_i = \sigma_0 (W_{t_i} - W_{t_{i-1}}) + (\epsilon_{t_i} - \epsilon_{t_{i-1}}),
\end{eqnarray}
$Y$ features a MA(1) process so that the log-likelihood process of the model is
\bea 
l(\sigma^2, a^2) = - \frac{1}{2} \textnormal{log det}(\Omega) - \frac{N}{2} \log (2\pi) -\frac{1}{2} Y^T \Omega^{-1} Y,
\label{loglikparamait}
\eea 
where $\Omega$ is the matrix 
\bea\Omega = \left(\begin{matrix}
                    \sigma^2 \Delta_N + 2 a^2 & - a^2 & 0 & \cdots & 0 \\
                    - a^2 & \sigma^2 \Delta_N + 2 a^2 & - a^2 & \ddots & \vdots \\
                    0 & - a^2 & \sigma^2 \Delta_N + 2 a^2 & \ddots & 0\\
                    \vdots & \ddots & \ddots & \ddots & - a^2\\
                    0 & \cdots & 0 & - a^2 & \sigma^2 \Delta_N + 2 a^2
                  \end{matrix}\right). \\[12pt]
\label{defOmega}
\eea 
When incorporating non-null information, the model for the returns can be written as 
\begin{eqnarray}
Y_i = \sigma_0 (W_{t_i} - W_{t_{i-1}}) + \mu_i (\theta_0) + (\epsilon_{t_i} - \epsilon_{t_{i-1}}).
\end{eqnarray}
It is then immediate to see that $\widetilde{Y} (\theta_0)$ follows a MA(1) dynamic so that we can substitute the log-likelihood function by
\bea 
l_{err}(\sigma^2, \theta, a^2) = - \frac{1}{2} \textnormal{log det}(\Omega) - \frac{N}{2} \log (2\pi) -\frac{1}{2} \widetilde{Y}(\theta)^T \Omega^{-1} \widetilde{Y}(\theta).
\label{loglikerr}
\eea 

To assess the central limit theory, we consider the general framework specified in Section \ref{model} and define the quadratic variation as
$$T\overline{\sigma}_0^2 :=\int_0^T{\sigma_s^2ds} + \sum_{0 < s \leq T}\Delta J_s^2,$$ 
where $\Delta J_s = J_s - J_{s-}$, and we assume that $\overline{\sigma}_0^2 \in \big[ \underline{\sigma}^2 , \overline{\sigma}^2 \big]$ almost surely, where $\underline{\sigma}^2 >0$. This assumption is necessary to maximize the quasi likelihood function on a well-defined bounded space. This may seem to be a somewhat restrictive condition on the volatility process, but since $\overline{\sigma}^2$ can be taken arbitrarily large, it does not affect the implementation of the estimation procedure in practice. Under $\mathcal{H}_0$ and assuming null information, \cite{ait2016hausman} show that the QMLE associated to (\ref{loglikparamait}) is optimal with rate of convergence $n^{1/2}$. When incorporating information into the model, both QMLE related to (\ref{loglikexp}) and (\ref{loglikerr}) also turn out to converge with rate $n^{1/2}$. Formally, we assume that $\upsilon_0 :=(\overline{\sigma}_0^2,\theta_0) \in \Upsilon$, where $\Upsilon = \big[ \underline{\sigma}^2 , \overline{\sigma}^2 \big] \times \Theta $. We define $\widehat{\upsilon}_{exp} := (\widehat{\sigma}_{exp}^2, \widehat{\theta}_{exp})$ and $\widehat{\xi}_{err} := (\widehat{\sigma}_{err}^2, \widehat{\theta}_{err}, \widehat{a}_{err}^2)$ as respectively one solution to the equation $\partial_\upsilon l_{exp}(\upsilon) = 0$ on the interior of $\Upsilon$ and one solution to the equation $\partial_{\xi} l_{err}(\xi) = 0$ on $\Upsilon \times \big[ - \underline{\eta}/n, \overline{\eta}/n \big]$, where $\overline{\eta} > 0$ and $0 < \underline{\eta} < \underline{\sigma}^4/4 $. This corresponds to an extension of parameter space as $a^2$ can take negative values, as in \cite{ait2016hausman} (see the discussion at the bottom of p. 8). Such extension is needed because under $\calh_0$ and with the non-extended space $[0, \bar{\eta}/n]$, the parameter $a_0^2 = 0$ would lie on the boundary of the parameter space, making the above procedure inconsistent. In the following theorem, we give the joint limit distribution of $(\widehat{\upsilon}_{exp}, \widehat{\xi}_{err})$ assuming that the noise process is of order $1/\sqrt{n}$. We also specify the limit under $\calh_0$, i.e when there is no residual noise.

\begin{theorem*}  
\label{theoremnoerror}
(Joint central limit theorem for $(\widehat{\upsilon}_{exp},\widehat{\xi}_{err})$ under the small residual noise framework) Assume that $a_0^2 = \eta_0 T/n$ and that $\textnormal{cum}_4[\epsilon] = \calk T^2/n^2$ for some fixed $\eta_0 \geq 0$, $\calk \geq 0$, where $\textnormal{cum}_4[\epsilon]$ is the fourth order cumulant of $\epsilon_t$. Then, we have $\calg_T$-stably\footnote{The filtration ${\bf G} = (\mathcal{G}_t)_{0 \leq t \leq T}$ is defined as $\mathcal{G}_t := \sigma \l\{U_i^n, \alpha_s, X_s | (i,n) \in \naturels^2 ,  0 \leq s \leq t \r\}$.} in law that
$$\left(\begin{matrix} N^{1/2} \big(\widehat{\sigma}_{exp}^2 - \overline{\sigma}_0^2 - 2\tilde{\eta}_0 \big) \\ 
N \big(\widehat{\theta}_{exp} - \theta_0 - N^{-1}B_{\theta_0, exp} \big) \\
N^{1/2} \big(\widehat{\sigma}_{err}^2 - \overline{\sigma}_0^2 \big) \\
N \big( \widehat{\theta}_{err} - \theta_{0}- N^{-1}B_{\theta_0, err} \big)\\
N^{3/2} \big( \widehat{a}_{err}^2 - a_0^2 \big) 
\end{matrix}\right) \to \calm\caln \left( 0, \frac{\mathcal{Q}}{T} \times \left(\begin{matrix} 2  & 0 & 2 & 0 & 0 \\ 
0  & \frac{T \int_0^T{\sigma_s^2ds}}{\calq} U_{\theta_0}^{-1} & 0 & \frac{T \int_0^T{\sigma_s^2ds}}{\calq} U_{\theta_0}^{-1} & 0 \\
2  & 0 & 6 & 0 & -2T \\
0  & \frac{T \int_0^T{\sigma_s^2ds}}{\calq} U_{\theta_0}^{-1} & 0 & \frac{T \int_0^T{\sigma_s^2ds}}{\calq} U_{\theta_0}^{-1} & 0 \\
0  & 0 & -2T & 0 & T^2 
\end{matrix}\right)  + \textbf{V} \right),$$
\normalsize where $\mathcal{Q} = T^{-1}\int_0^T{\alpha_s^{-1}ds}\l\{\int_0^T{\sigma_s^4 \alpha_sds} + \sum_{0 <s \leq T} \Delta J_s^2 (\sigma_s^2\alpha_s + \sigma_{s-}^2\alpha_{s-}) \r\}$, $\tilde{\eta}_0 = (T^{-1}\int_0^T\alpha_s^{-1}ds)\eta_0$, $\widetilde{\calk} = (T^{-1}\int_0^T\alpha_s^{-1}ds)^2\calk$,
$$ U_{\theta_0} = \esp\l[\frac{\partial\mu_1\l(\theta_0\r)}{\partial \theta} . \frac{\partial\mu_1\l(\theta_0\r)}{\partial \theta}^T\r], $$
$\textbf{V}:= \textbf{V}(\calq, \overline{\sigma}_0^2, \tilde{\eta}_0, \widetilde{\calk},\theta_0)$ is an additional matrix due to the presence of residual noise of the form 
\small
$$ \textbf{V} = \left(\begin{matrix} \textbf{V}_{11}  & 0 & \textbf{V}_{13} & 0 & \textbf{V}_{15} \\ 
0  & \textbf{V}_{22} & 0 & \textbf{V}_{24} & 0 \\
\textbf{V}_{13}  & 0 & \textbf{V}_{33} & 0 & \textbf{V}_{35} \\
0  & \textbf{V}_{24} & 0 & \textbf{V}_{44} & 0 \\
\textbf{V}_{15}  & 0 & \textbf{V}_{35} & 0 & \textbf{V}_{55} 
\end{matrix}\right),$$
\normalsize and $B_{\theta_0,exp}$, $B_{\theta_0,err}$ are two bias terms due to the presence of jumps in the price process. The exact expression of $\textbf{V}$, $B_{\theta_0,exp}$, and $B_{\theta_0,err}$ can be found in Section \ref{avarTheoremSmallNoise}. \\
In particular, $\textbf{V}(\calq, \overline{\sigma}_0^2, 0,0,\theta_0) = 0$, and thus, under $\mathcal{H}_0$, we have $\calg_T$-stably in law that \small
$$\left(\begin{matrix} N^{1/2} \big(\widehat{\sigma}_{exp}^2 - \overline{\sigma}_0^2 \big) \\ 
N \big(\widehat{\theta}_{exp} - \theta_0 - N^{-1}B_{\theta_0,exp} \big) \\
N^{1/2} \big(\widehat{\sigma}_{err}^2 - \overline{\sigma}_0^2 \big) \\
N \big( \widehat{\theta}_{err} - \theta_{0} - N^{-1}B_{\theta_0,err} \big)\\
N^{3/2} \big( \widehat{a}_{err}^2 - 0 \big) 
\end{matrix}\right) \to \calm\caln \left( 0, \frac{\mathcal{Q}}{T} \times \left(\begin{matrix} 2  & 0 & 2 & 0 & 0 \\ 
0  & \frac{T \int_0^T{\sigma_s^2ds}}{\calq} U_{\theta_0}^{-1} & 0 & \frac{T \int_0^T{\sigma_s^2ds}}{\calq} U_{\theta_0}^{-1} & 0 \\
2  & 0 & 6 & 0 & -2T \\
0  & \frac{T \int_0^T{\sigma_s^2ds}}{\calq} U_{\theta_0}^{-1} & 0 & \frac{T \int_0^T{\sigma_s^2ds}}{\calq} U_{\theta_0}^{-1} & 0 \\
0  & 0 & -2T & 0 & T^2 
\end{matrix}\right) \right).$$
\normalsize 
\end{theorem*}
\begin{remark*} (Regular sampling) If observations are regular, $\mathcal{Q}$, $\tilde{\eta}_0$ and $\widetilde{\calk}$ can be specified as
\begin{eqnarray}
\mathcal{Q} = \int_0^T{\sigma_s^4 ds} + \sum_{0 < s \leq T} \Delta J_s^2 (\sigma_s^2 + \sigma_{s-}^2),\textnormal{  } \tilde{\eta}_0 = \eta_0,\textnormal{  }\widetilde{\calk} = \calk. 
\end{eqnarray}
\end{remark*}

We consider now the problem of testing $\mathcal{H}_0$ against $\mathcal{H}_1$. To do that we consider the Hausman statistics of the form
\begin{eqnarray}
S = N (\widehat{\sigma}_{exp}^{2} - \widehat{\sigma}_{err}^{2})^2/\widehat{V},
\end{eqnarray}
where $\widehat{V}$ is a consistent estimator of $AVAR(\widehat{\sigma}_{exp}^{2} - \widehat{\sigma}_{err}^{2})$ that will be defined in what follows. We aim to show that $S$ satisfies the key asymptotic properties
\begin{eqnarray}
\label{SH0} S \rightarrow \chi^2 & \text{ under } \mathcal{H}_0,\\
\label{SH1} S \rightarrow \infty & \text{ under } \mathcal{H}_1,
\end{eqnarray}
where $\chi^2$ is a standard chi-squared distribution. Actually, we can deduce (\ref{SH0}) from Theorem \ref{theoremnoerror} along with the consistency of $\widehat{V}$ and (\ref{SH1}) is relatively easy to obtain. As in \cite{ait2016hausman}, we consider three distinct scenarios, i.e.
\begin{enumerate}[(i)]
    \item constant volatility
    \item time-varying volatility and no price jump
    \item time-varying volatility and price jump
\end{enumerate}
This leads us to define two (one estimator is robust to two scenarios) distinct variance estimators $\widehat{V}_i$ and their affiliated statistics $S_i = N (\widehat{\sigma}_{exp}^{2} - \widehat{\sigma}_{err}^{2})^2/\widehat{V}_i$ in what follows. Due to the non regularity of arrival times, fourth power returns based estimators such as $\widehat{V}_5$ (defined in Section \ref{supplementaryvarest}) are inconsistent in general. We therefore consider bipower statistics, inspired by \cite{barndorff2004power} and \cite{barndorff2004econometric}.
If we assume (i) we have that 
$AVAR \big(\widehat{\sigma}_{exp}^2 - \widehat{\sigma}_{err}^2 \big) = 4 T^{-2} \sigma_0^4 \int_0^T\alpha_s^{-1}ds \int_0^T\alpha_s ds$, which can be estimated by 
\begin{eqnarray}
\widehat{V}_1 & = & \frac{4 N}{T^2} \sum_{i=2}^{N} \Delta \widehat{X}_{i}^2 \Delta \widehat{X}_{i-1}^2.
\end{eqnarray}
The estimator $\widehat{V}_1$ is also robust to (ii), where $AVAR(\widehat{\sigma}_{exp}^2 - \widehat{\sigma}_{err}^2) = 4 T^{-2}  \int_0^T\alpha_s^{-1}ds \int_0^T\sigma_s^4 \alpha_s ds$. Under (iii), $AVAR(\widehat{\sigma}_{exp}^2 - \widehat{\sigma}_{err}^2) = 4 T^{-2}\int_0^T{\alpha_s^{-1}ds} \big\{\int_0^T{\sigma_s^4 \alpha_sds} + \sum_{0 < s \leq T} \Delta J_s^2 (\sigma_s^2\alpha_s + \sigma_{s-}^2\alpha_{s-}) \big\}$. If we introduce $\widetilde{k}$ which is random and satisfies $\widetilde{k} \Delta_N \rightarrow^{\proba} 0$ and $\widetilde{u}_i = \widetilde{\alpha} (t_i - t_{i-1})^{\omega}$, we can estimate $AVAR(\widehat{\sigma}_{exp}^2 - \widehat{\sigma}_{err}^2)$ with
\begin{eqnarray}
\widehat{V}_2 & = & \frac{4 }{T} \Bigg\{ \frac{1}{\Delta_N} \sum_{i=2}^{N} \Delta \widehat{X}_{i}^2 \Delta \widehat{X}_{i-1}^2 \mathbf{1}_{\{\mid \Delta \widehat{X}_{i} \mid \leq \widetilde{u}_i \}} \mathbf{1}_{\{\mid \Delta \widehat{X}_{i-1} \mid \leq \widetilde{u}_{i-1} \}} \\ & & + \sum_{i=\widetilde{k}+1}^{N- \widetilde{k}} \Delta \widehat{X}_{i}^2 \mathbf{1}_{\{\mid \Delta \widehat{X}_{i} \mid > \widetilde{u}_i \}}
\big( \widehat{\sigma^{2}_{t_{i}} \alpha_{t_i}} +\widehat{\sigma^{2}_{t_{i-}} \alpha_{t_i-}} \big) \Bigg\} \text{ , where} \nonumber \\
\widehat{\sigma^{2}_{t_{i}} \alpha_{t_i}} & = & \frac{1}{\tilde{k} \Delta_N} \sum_{j=i+1}^{i+\widetilde{k}} \Delta \widehat{X}_j^2  \mathbf{1}_{\{\mid \Delta \widehat{X}_{j} \mid \leq \widetilde{u}_j \}} \text{ , } \widehat{\sigma^{2}_{t_{i}-} \alpha_{t_i-}}  =  \widehat{\sigma^{2}_{t_{i-\widetilde{k}-1}} \alpha_{t_{i-\widetilde{k}-1}}}. \nonumber 
\end{eqnarray}
We first show the consistency of the proposed estimators.
\begin{proposition*} \label{propAVAR}
For any $i = 1,2$ we have, as $n \to +\infty$, 
\beas 
\textnormal{under }\calh_0\textnormal{: }\widehat{V}_i &\to^\proba& AVAR(\widehat{\sigma}_{exp}^2 - \widehat{\sigma}_{err}^2),\\
\textnormal{under } \calh_1\textnormal{: } \widehat{V}_i &=& O_\proba(1),
\eeas 
if we assume the related framework.
\end{proposition*}
We then deduce asymptotic properties of the statistics. \begin{corollary*} \label{testConsistency}
Let $0 < \beta < 1$ and $c_{\beta}$ the associated $\beta$-quantile of the standard chi-squared distribution. Under the related framework, the test statistics $S_i$ satisfy
\begin{eqnarray}
\proba (S_i > c_{1-\beta} \mid \mathcal{H}_0) \rightarrow \beta \text{ and } \proba (S_i > c_{1-\beta} \mid \mathcal{H}_1) \rightarrow 1.
\end{eqnarray}
\end{corollary*}
When there is no residual noise in the model, following the procedure considered in \cite{li2016efficient} and \cite{chaker2017high}, we can estimate the efficient price as  
\begin{eqnarray}
\label{efficientpriceestimator2}
\widehat{X}_{t} = Z_{t_i} -  \phi(Q_i, \widehat{\theta}_{exp}) \text{, for } t  \in \big(t_{i-1},t_{i} \big].
\end{eqnarray}
By virtue of Theorem \ref{theoremnoerror}, we have that $\widehat{\theta}_{exp}$ is consistent and thus we can show the consistency of $\widehat{X}_{t}$. It is also immediate to see that
\begin{eqnarray}
\widehat{\sigma}_{exp}^{2} = T^{-1} \sum_{i=1}^N (\widehat{X}_{t_i} - \widehat{X}_{t_{i-1}})^2.
\label{formulaSigmaExp}
\end{eqnarray}
Formally, the volatility estimator (\ref{formulaSigmaExp}) expressed as a function of the estimated parameter is equal to the volatility estimator (also viewed as a function of the estimated parameter) considered in \cite{li2016efficient}. Moreover, given the shape of $l_{exp}$ in (\ref{loglikexp}), the QMLE $\widehat{\theta}_{exp}$ and the least square estimator (9) in the cited paper (p. 35) coincide, implying that both volatility estimators are equal. 
The need to correct for the price prior to using RV in (\ref{formulaSigmaExp}) can be understood looking at Table 3 (p. 1324) from \cite{diebold2013correlation}. In that table, the first column reports the limit of the naive RV. Accordingly, one can see that depending on the serial autocorrelation of $\phi$, there will be one or more extra autocorrelation terms in the limit. Subsequently, the use of the price estimation in (\ref{formulaSigmaExp}) permits to get rid of those additive terms.

\smallskip
The following corollary formally states the consistency of $\widehat{X}_{t}$ and the efficiency of RV when used on $\widehat{X}_{t}$. This corresponds exactly to Theorem 2 in \cite{li2016efficient}. This also corresponds to Theorem 4 (i) in \cite{chaker2017high} when $\phi$ is linear.
\begin{corollary*}\label{corollaryRV}
Under $\mathcal{H}_0$, the estimator $\widehat{X}_{t}$ is consistent, i.e. for any $t \in [0,T]$,
\begin{eqnarray}
\widehat{X}_{t} \rightarrow^\proba X_t. 
\end{eqnarray}
Furthermore, we have $\calg_T$-stably in law that \begin{eqnarray}
N^{1/2}\l(\sum_{i=1}^N (\widehat{X}_{t_i} - \widehat{X}_{t_{i-1}})^2 - \int_0^T \sigma_s^2ds - \sum_{0<s\leq T}\Delta J_s^2 \r) \rightarrow \calm\caln(0, 2T \mathcal{Q}).
\label{cltRVestimated}
\end{eqnarray}
In particular, when observations are regular and the efficient price is continuous, this can be written as
\begin{eqnarray}
N^{1/2}\l(\sum_{i=1}^N (\widehat{X}_{t_i} - \widehat{X}_{t_{i-1}})^2 -\int_0^T \sigma_s^2ds\r) \rightarrow \calm\caln\bigg(0, 2 T \int_0^T \sigma_s^4 ds \bigg).
\label{cltRVestimated2}
\end{eqnarray}
\end{corollary*}

It is interesting to remark that when $J=0$, convergence (\ref{cltRVestimated}) shows that RV on the estimated price is efficient in the sense that its AVAR attains the nonparametric efficiency bound derived in  \cite{renault2017efficient}. Indeed, taking $T=1$, note that our model of observation times falls under the setting of \cite{renault2017efficient} (see Assumption 2 and the short discussion below), where, in view of (2.10) on p. 447 in the aforementioned paper, we easily derive that $\alpha_s = T_s'^{-1}$. Thus, (\ref{cltRVestimated}) can be rewritten as 
\begin{eqnarray}
n^{1/2}\l(\sum_{i=1}^N (\widehat{X}_{t_i} - \widehat{X}_{t_{i-1}})^2 - \int_0^1 \sigma_s^2ds \r) \rightarrow \calm\caln\l(0, 2 \int_0^1 \sigma_s^4 T_s'^{-1}ds\r),
\label{cltRVestimated3}
\end{eqnarray}
which corresponds precisely to the efficiency bound (3.18) on p. 454 in \cite{renault2017efficient} in the case $g(u,\sigma^2) = \sigma^2$.

\subsection{Large noise alternative case}
\label{largenoisealternativecase}
If one can detect small noise, one can a-priori detect large noise. In this section, we consider the large noise alternative $\widetilde{\calh}_1: \{a_0^2 := \eta_0, \phi = \Phi \}$, where we recall that $\eta_0 > 0$ and $\Phi := \Phi (Q_i, \theta_0 ) \neq 0$. We show that Proposition \ref{propAVAR} and Corollary \ref{testConsistency} remain valid in what follows. We have removed the statements related to $\calh_0$ which obviously stay true.
\begin{proposition*} \label{propAVARlargenoise}
Under the related framework, for any $i = 1,2$ we have, as $n \to +\infty$, 
\beas 
\textnormal{under } \widetilde{\calh}_1\textnormal{: } \widehat{V}_i &=& O_\proba(N^2).
\eeas 
\end{proposition*}
\begin{corollary*} \label{testConsistencylargenoise}
Under the related framework, the test statistics $S_i$ satisfy
\begin{eqnarray}
\proba (S_i > c_{1-\beta} \mid \widetilde{\mathcal{H}}_1) \rightarrow 1.
\end{eqnarray}
\end{corollary*}

\subsection{The $\phi=0$ alternative case}
\label{phi0alternativecase}
So far we have assumed that the tests were conditional on a specific parametric model where $\phi$ is non-null, so that the null hypothesis and the alternative were considered under the constraint $\phi \neq 0$. In this part, we consider the pure i.i.d MMN alternative $\overline{\mathcal{H}}_1 : \{ \phi = 0,a_0^2 > 0 \}$, where the noise may be small ($a_0^2 = \eta_0 T/n$) or large ($a_0^2 = \eta_0T$) with $\eta_0 > 0$. Note that this is precisely the same alternative as that of \cite{ait2016hausman}. We prove in what follows that Proposition \ref{propAVAR} and Corollary \ref{testConsistency} remain true up to an innocuous assumption on the fitted model $\{\phi(.,\theta), \theta \in \Theta\}$, which is satisfied on all the models considered in this paper. Here again we have removed the statements related to $\calh_0$.

\begin{proposition*} \label{propAVARphi0}
Assume that there exists $\widetilde{\theta}$ in the interior of $\Theta$ such that $\phi(.,\widetilde{\theta}) = 0$. Then, under the related framework, for any $i = 1,2$ we have, as $n \to +\infty$, 
\beas 
\textnormal{under } \overline{\calh}_1\textnormal{ with } a_0^2 = \eta_0 T/n\textnormal{: } \widehat{V}_i = O_\proba(1),\\
\textnormal{under } \overline{\calh}_1\textnormal{ with } a_0^2 = \eta_0 T\textnormal{: }\widehat{V}_i = O_\proba(N_n^2).
\eeas 
\end{proposition*}
\begin{corollary*} \label{testConsistencyphi0}
Assume that there exists $\widetilde{\theta}$ in the interior of $\Theta$ such that $\phi(.,\widetilde{\theta}) = 0$. Then, under the related framework, the test statistics $S_i$ satisfy
\begin{eqnarray}
\proba (S_i > c_{1-\beta} \mid \overline{\mathcal{H}}_1) \rightarrow 1.
\end{eqnarray}
\end{corollary*}

\section{Goodness of fit}
The goal of this section is threefold. First, we introduce a measure of goodness of fit which can be used by the high frequency data user prior to testing to compare several models and assess if one or several candidates are worth testing. Second, we provide the central limit theory of the QMLE related to the model including the residual noise when assuming that it is present, and we deduce an estimator of the measure. Finally, we give a practical guidance to estimate volatility -this sequence is illustrated in the finite sample analysis that follows.

\subsection{Definition}
Prior to looking at the Hausman tests, it is safer to assume a model where the variance of the residual noise $a_0^2 > 0$ is non negligible. We introduce the proportion of variance explained as
\begin{eqnarray}
\label{defpropvar}
\pi_{V} := \frac{\esp \big[\phi(Q_0,\theta_0)^2 \big]}{\esp \big[\phi(Q_0,\theta_0)^2 \big] + a_0^2},
\end{eqnarray}
which is a measure of goodness of fit of the model. This measure is almost identical to $\pi_{exp}$ from Remark 8 (p. 37) in \cite{li2016efficient}.  The estimation of (\ref{defpropvar}) is based on the QMLE related to the model including the residual noise and given in (\ref{estpropvar}).

\subsection{Central limit theory}
Throughout the rest of this section we assume that the residual noise variance $a_0^2 > 0$ does not depend on $n$. When $\phi = 0$, this corresponds to a widespread assumption on the residual noise (which in this case corresponds exactly to the MMN). In this setting and further assuming that the volatility is constant, \cite{ait2005often} show that the MLE related to (\ref{loglikparamait}) is efficient with convergence rate $n^{1/4}$ and obtain the robustness of the MLE in case of departure from the normality of the noise. \cite{xiu2010quasi} shows that the procedure is also robust to time-varying volatility. We further investigate in \cite{clinet2018efficient} the behavior of the estimator when adding jumps to the price process and considering non regular stochastic arrival times. In what follows we show in particular that $\widehat{\sigma}_{err}^2$ converges at the same rate $n^{1/4}$. We assume that $\xi_0 :=(\overline{\sigma}_0^2,a_0^2,\theta_0)  \in \Xi$, where $\Xi = \big[ \underline{\sigma}^2 , \overline{\sigma}^2 \big] \times \big[ \underline{a}^2 , \overline{a}^2 \big] \times \Theta$ with $\underline{a}^2 > 0$. Finally, $\widehat{\xi}_{err}$ is defined as one solution to the equation $\partial_\xi l_{err}(\xi) = 0$ on the interior of $\Xi$.

\begin{theorem*} \label{theoremCLT} We have $\calg_T$-stably in law that
\begin{eqnarray}
\label{eqtheoremH1}
\left(\begin{matrix} N^{1/4} \big(\widehat{\sigma}_{err}^2 - \overline{\sigma}_0^2 \big) \\ 
N^{1/2} \big( \widehat{a}_{err}^2 - a_0^2 \big) \\
N^{1/2} \big( \widehat{\theta}_{err} - \theta_{0} \big)
\end{matrix}\right) \to \calm\caln \left( 0, \left(\begin{matrix} \frac{ 5 a_0 \mathcal{Q}}{ T^{3/2}\overline{\sigma}_0 } + \frac{3 a_0 \overline{\sigma}_0^3}{T^{1/2}} & 0 &0 \\ 
0 & 2 a_0^4 + \textnormal{cum}_4[\epsilon] &0 \\
0&0&  a_0^2 V_{\theta_0}^{-1} 
\end{matrix}\right)  \right),
\end{eqnarray}
 where the term $\textnormal{cum}_4[\epsilon]$ stands for the fourth order cumulant of $\epsilon$, and $V_{\theta_0}$ is the Fisher information matrix related to $\theta_0$ defined as
$$V_{\theta_0} = \esp\l[ \frac{\partial\phi\l(Q_{0},\theta_0\r)}{\partial \theta} . \frac{\partial\phi\l(Q_{0},\theta_0\r)}{\partial \theta}^T \r].$$
\end{theorem*}

\begin{remark*} (Variance gain when estimating the quadratic variation) If we assume that the information process $Q_i$ is i.i.d, it is also possible to directly estimate the quadratic variation and the global noise variance using the original QMLE of \cite{xiu2010quasi} and generalized to our setting with jumps and stochastic observation times in \cite{clinet2018efficient}. Denoting such estimator by $(\widehat{\sigma}_{ori}^2,\widehat{a}_{ori}^2)$, we have:
\small
\begin{eqnarray}
\left(\begin{matrix} N^{1/4} \big(\widehat{\sigma}_{ori}^2 - \overline{\sigma}_0^2 \big) \\ 
N^{1/2} \big( \widehat{a}_{ori}^2 - \widetilde{a}_0^2 \big) \\
\end{matrix}\right) \to \calm\caln \left( 0, \left(\begin{matrix} \frac{ 5 \widetilde{a}_0 \mathcal{Q}}{ T^{3/2}\overline{\sigma}_0 } + \frac{3 \widetilde{a}_0 \overline{\sigma}_0^3}{T^{1/2}} & 0  \\ 
0 & 2 \widetilde{a}_0^4 + \textnormal{cum}_4[\epsilon + \phi(Q_0,\theta_0)]  
\end{matrix}\right)  \right),
\normalsize
\label{cltClassical}
\end{eqnarray}
where $\widetilde{a}_0^2 = a_0^2 + \esp[\phi(Q_0,\theta_0)^2]$. Therefore, in view of (\ref{eqtheoremH1}) and (\ref{cltClassical}), accounting for the explicative part of the noise in the estimation process results in an asymptotic variance reduction for the volatility estimation of a factor 
$$ \frac{\widetilde{a}_0}{a_0} =  \sqrt{1 + \frac{\esp[\phi(Q_0,\theta_0)^2]}{a_0^2}} = \frac{1}{\sqrt{1 - \pi_V}}.$$
In particular, when the residual noise is negligible, i.e. in the limit $\pi_V \to 1$, we see that the asymptotic gain is infinite which is coherent with the fact that in such case the rate of convergence of $\widehat{\sigma}_{err}^2$ switches from $N^{1/4}$ to $N^{1/2}$ as in Theorem \ref{theoremnoerror}.
\end{remark*}

\begin{remark*} (Connection to the literature) \cite{li2016efficient} consider a shrinking noise in their Theorem 3, and thus they obtain a faster rate of convergence $n^{1/2}$ for the volatility estimator. On the other hand, \cite{chaker2017high} considers the same setting as ours in their Theorem 4, but as they are using a modification of the TSRV, they obtain the slower rate of convergence $n^{1/6}$.
\end{remark*}

\begin{remark*} (Regular sampling and continuous price case) When the observation times are regularly spaced and the efficient price is continuous, (\ref{eqtheoremH1}) can be specified as
\small
\begin{eqnarray}
\left(\begin{matrix} n^{1/4} \big(\widehat{\sigma}_{err}^2 - \overline{\sigma}_0^2 \big) \\ 
n^{1/2} \big( \widehat{a}_{err}^2 - a_0^2 \big) \\
n^{1/2} \big( \widehat{\theta}_{err} - \theta_{0} \big)
\end{matrix}\right) \to \calm\caln \left( 0, \left(\begin{matrix} \frac{ 5   a_0\int_0^T{\sigma_s^4 ds}}{ T\l(\int_0^T{\sigma_s^2ds}\r)^{1/2} } + \frac{3 a_0 \l(\int_0^T{\sigma_s^2ds}\r)^{3/2}}{T^2} & 0 &0 \\ 
0 & 2 a_0^4 + \textnormal{cum}_4[\epsilon] &0 \\
0&0&  a_0^2 V_{\theta_0}^{-1} 
\end{matrix}\right)  \right).
\normalsize
\end{eqnarray}

\end{remark*}

\begin{remark*} (Local QMLE) Using the local QMLE, we could further reduce the AVAR of the volatility obtained in (\ref{eqtheoremH1}). In the case of regular sampling and continuous price process, we could be as close as possible from the lower efficiency bound defined in \cite{reiss2011asymptotic}. The proofs of this paper would straightforwardly adapt. The case $\phi = 0$ is treated in \cite{clinet2018efficient}.
\end{remark*}
Based on Theorem \ref{theoremCLT}, we can consistently estimate $\pi_{V}$ as 
\begin{eqnarray}
\label{estpropvar}
\widehat{\pi}_{V} := \frac{(N+1)^{-1}\sum_{i=0}^{N} \phi(Q_i, \widehat{\theta}_{err})^2}{(N+1)^{-1}\sum_{i=0}^{N} \phi(Q_i, \widehat{\theta}_{err})^2 + \widehat{a}_{err}^2}.
\end{eqnarray}
 In accordance with our empirical findings (see Section \ref{sectionEmpiricalStudy}), we also investigate what happens in the case where $\pi_V$ approaches $1$, which corresponds to the small noise framework of Section \ref{sectionNoNoise}, where $a_0^2 = \eta_0T/n$ for some fixed $\eta_0 \geq 0$, and where $\esp[\phi(Q_0,\theta_0)^2] > 0$. It turns out that under this framework too $\widehat{\pi}_V$ converges to $\pi_V$. In particular, this also proves the consistency of $\widehat{\pi}_V$ in the case $a_0^2 =0$, that is $\pi_V =1$. More precisely, we have the following result.

\begin{lemma*}\label{lemmaConsistencyPi}(consistency of $\widehat{\pi}_V$)
In the large noise case $a_0^2 > 0$, we have 
$$ \widehat{\pi}_V = \pi_V + o_\proba(1).$$
In the small noise case $a_0^2 = \eta_0T/n$, $\eta_0 \geq 0$, with $\esp[\phi(Q_0,\theta_0)^2] > 0$, we have
$$ \widehat{\pi}_V = \pi_V + o_\proba\left(N^{-1}\right). $$
\end{lemma*}

\subsection{Practical guidance to estimate volatility}
\label{practicalguidance}
In this section, we provide a sequence -not theoretically validated  but which behaves correctly numerically in next section- to estimate volatility based on the introduced tests. As a matter of fact, the sequence may suffer from the so-called post-model selection problem (as in, e.g \cite{leeb2005model}). This is due to a possible lack of uniformity when pre-testing the presence of residual noise, and may affect the finite sample performance of the volatility estimator constructed from the sequence hereafter. A solution to that issue would be to prove that the proposed inference is uniformly valid with respect to the residual noise magnitude, in a similar way as in \cite{da2017moving}, Section 4.3-4.4. In practice, Section 5 of the present paper suggests that in finite sample, the post-model selection problem does not seem to impact much volatility estimation for the considered models. Moreover, we provide steps -here again completely ad hoc, but implemented in our empirical study- to the empirical researcher to investigate if it is worth considering a specific $\phi$ when implementing the tests. A rigorous statistical approach to choose $\phi$ among a class of competitive models in practice based on Bayesian Information Criterion is beyond the scope of this paper and can be found in \cite{clinet2018relation}.

\smallskip
We suggest the following sequence for volatility estimation:
\begin{itemize}
    \item If the original Hausman tests from \cite{ait2016hausman} are not rejected, it is reasonably safe to use RV on the raw data, even though it does not necessarily mean that there is no MMN -in our simulation study and empirical study, we find that the tests are rejected (almost) all the time when used at the highest frequency on fairly liquid stocks-. We emphasize that although not reported in our numerical study results, the original tests from \cite{ait2016hausman} when implemented with estimators not robust to autocorrelated MMN (such as RK, PAE, QMLE) are distorted by the presence of MMN of the form $\phi(Q_i, \theta_0) + \epsilon_{t_i}$ with $\phi \neq 0$. Accordingly, in line with our numerical study, we strongly advise the user to implement $\widehat{\sigma}_{err}^2$ as volatility estimator to be compared with RV. This requires a priori to know $\phi$ and $Q_i$. Another alternative, which does not need any preestimation, consists in using an estimator robust to autocorrelated MMN such as in \cite{da2017moving}. We did not implement this type of estimator in our numerical study. Finally, we insist on the fact that the theory related to such Hausman tests has not been investigated. 
    \item If the original tests are rejected and the tests considered in this paper are rejected, one should stick to $\widehat{\sigma}_{err}^2$.
    \item If the original tests are rejected and the tests of this paper are not rejected, then one should use $\widehat{\sigma}_{exp}^2$.
\end{itemize}
Finally, we recommend the empirical researcher the following steps to choose a specific $\phi$ prior to implementing the tests on several models:
\begin{itemize}
    \item The user may implement the original Hausman tests from \cite{ait2016hausman} on the raw data. In agreement with the sequence of volatility estimation, we advise the user to choose a volatility estimator robust to autocorrelated MMN.
    \item If the results seem to indicate the presence of MMN, the user should  estimate the ratio (\ref{defpropvar}).
    \item If the ratio turns out to be close to 100\%, then a proper investigation using the tests should be carried out.
\end{itemize}
To illustrate the method, we follow this procedure in our empirical study. 

\section{Finite sample performance}
\label{sectionFiniteSample}
We now conduct a Monte Carlo experiment to assess finite sample performance of the tests, and validity of the sequence to estimate volatility described in Section \ref{practicalguidance} -which a priori is subject to multiple testing, model selection and post model selection issues- by comparing it to some leading estimators from the literature. We simulate M=1,000 Monte Carlo days of high-frequency returns where the related horizon time $T = 1/252$ is annualized. One working day corresponds to 6.5 hours of trading activity, i.e. 23,400 seconds.

\subsubsection*{The efficient price}
We introduce the Heston model with U-shape intraday seasonality component and jumps in both price and volatility as
\begin{eqnarray*}
dX_t &=& b dt + \sigma_{t}dW_t + dJ_t,\\
\sigma_t &=& \sigma_{t-,U} \sigma_{t,SV},
\end{eqnarray*}
where
\begin{eqnarray*} 
\sigma_{t,U} & = & C + Ae^{-at/T} + De^{-c(1-t/T)} -  \beta\sigma_{\tau-,U}\mathbb{1}_{\{t \geq \tau \}},\\
d\sigma_{t,SV}^2 & = & \alpha(\bar{\sigma}^2 - \sigma_{t,SV}^2)dt + \delta \sigma_{t,SV}d\bar{W}_t, 
\end{eqnarray*}
with $b = 0.03$, $dJ_t = \nabla S_t d N_t$, $\nabla = T \bar{\sigma}^2$, the signs of the jumps  $S_t = \pm 1$ are i.i.d symmetric, $N_t$ is a homogeneous Poisson process with parameter $\bar{\lambda} = T$ so that the contribution of jumps to the total quadratic variation of the price process is around 50\%, $C=0.75$, $A=0.25$, $D=0.89$, $a=10$, $c=10$, the volatility jump size parameter $\beta=0.5$, the volatility jump time $\tau$ follows a uniform distribution on $[0,T]$, $\alpha = 5$, $\bar{\sigma}^2 = 0.1$, $\delta = 0.4$, $\bar{W}_t$ is a standard Brownian motion such that  $d\langle W,\bar{W} \rangle_t = \overline{\phi} dt$, $\overline{\phi} = -0.75$, $\sigma_{0,SV}^2$ is sampled from a Gamma distribution of parameters $(2\alpha\bar{\sigma}^2/\delta^2,\delta^2/2\alpha)$, which corresponds to the stationary distribution of the CIR process. To obtain more information about the model and values, see \cite{clinet2018efficient}. The model is inspired directly from \cite{andersen2012jump} and \cite{ait2016hausman}.

\subsubsection*{The observation times}
We consider three levels of sampling: tick by tick, 15 seconds, 30 seconds. The observation times are generated regularly except for the tick by tick case. For the latter, we assume that $\alpha_t = 1/( e^{\beta_1} + \{ e^{\beta_2} + e^{\beta_3} \}^2 (t/T - e^{\beta_2}/(e^{\beta_2} + e^{\beta_3}))^2)$, and that $U_i$ are following an exponential distribution with parameter $2T/23,400$.
We have that the rate of arrival times $\alpha_t^{-1}$ exhibits a usual U-shape intraday pattern, as pointed out in \cite{engle1998autoregressive} (see discussions in Section 5-6 and Figure 2) and \cite{chen2013inference} (see Section 5, pp. 1011-1017). We fix $\beta_1 = -0.84$, $\beta_2 = -0.26$ and $\beta_3 = -0.39$ following the empirical values exhibited in \cite{clinet2018statistical}, which implies that the sampling frequency is on average faster than one second. 

\subsubsection*{The information}
We implement two models: Roll model and the signed spread model. As in the simulation study from \cite{li2016efficient}, the trade indicator $I_i$ is simulated featuring a Bernoulli process with parameter $p=1/2$ and with an autocorrelation chosen equal to 0.3. We fix the parameter $\theta = 0.0001$ in the case of Roll model. For the signed spread model, we further simulate the spread $\overline{S}_i$ as an AR(1) process with mean $0.000125$, variance $10^{-10}$ and correlation parameter which amounts to $0.6$. The parameter is chosen equal to $\theta = 0.80$. The values of the parameters correspond roughly to the fitted values\footnote{Although not fully reported in the empirical study.}. 

\subsubsection*{The residual noise}
To assess finite sample performance of the tests, we consider two types of (finite sample) alternative. In $\mathcal{H}_1$, we assume that the residual noise is i.i.d normally distributed with zero-mean and variance $a_0^2 = 10^{-9}, 10^{-8}, 10^{-7}$. In $\mathcal{H}_2$, which in particular does not accommodate with the assumptions of this paper, we assume that the residual noise 
$$\epsilon_{t_i} = \l(\frac{a_0}{\sqrt{3}} + \nu N_1\r)( \text{sign} (\Delta X_i)\l|N_2\r| + I_i \l|N_3\r| + N_4),$$
where $N_1$, $N_2$, $N_3$ and $N_4$ are standard independent normally distributed variables, and $\nu = a_0^2$. With that specification, the residual noise has also variance approximately equal to $a_0^2$ (with $a_0^4 \ll a_0^2$), is serially correlated (since $I_i$ is serially correlated), heteroskedastic, endogenous as both correlated with the efficient returns and the explicative part of the MMN.

\smallskip
To validate the sequence to estimate volatility, we consider $a_0^2 = 0, 10^{-9}, \text{mix}$, where the latter corresponds to a setup where there is no residual noise for half of the days in the sample and a residual noise with variance $a_0^2 = 10^{-9}$ for the remaining half in the sample.

\subsubsection*{Remaining tuning parameters}
Although the likelihood-based estimators do not require any tuning parameter, we need to select some parameters for the truncation method used when computing $S_3$ and $S_5$. We choose $k = \lfloor n^{1/2} \rfloor$, $\omega = 0.48$, $\widetilde{\alpha} = \alpha_0 \widehat{\sigma}_{exp}$, $\alpha_0 = 4$, and $\widetilde{k} = \lfloor N^{1/2} \rfloor$, consistently with \cite{ait2016hausman} (except for $\alpha_0 = 4$ which was set equal to 3, because this was yielding too many jumps detection in our case).

\subsubsection*{Concurrent volatility estimators and simulated model considered for comparison}
We consider a group of eight concurrent volatility estimators which is a mix of estimators considered in this paper and leading estimators from the literature. $S$ corresponds to the sequence introduced in Section \ref{practicalguidance}. The QMLEexp is $\widehat{\sigma}_{exp}^2$, and is actually equal to estimated-price RV defined in \cite{li2016efficient} since both $\phi$ considered are linear. The QMLEerr is defined as $\widehat{\sigma}_{err}^2$. The E-QMLE is the two-step estimator with price estimation first and (regular) QMLE on the price estimates. We also have some popular estimators QMLE, PAE, RK, and RV.\footnote{Details on the choice of tuning parameters for the PAE and the RK can be obtained upon request to the authors.} In particular, excluding $S$, no tests  are applied prior or post to the estimation methods.

\smallskip
The simulated model considered features time-varying volatility but does not incorporate jumps in the price process as most methods considered are not robust to such environment. Moreover, the sampling times are regular (with high frequency sampling every second) since some methods may be badly affected when they are not.
\subsubsection*{Results}
We first discuss the results to assess finite sample performance of the tests. We compute $S_1$ and $S_2$ when dealing with the tick by tick simulated returns, and $S_3$, $S_4$ and $S_5$, which are introduced in the appendix, when looking at sparser observations. We report in Table \ref{tableSimu} and in Table \ref{tableSimuGen} the fraction of rejections of $\mathcal{H}_0$ at the 0.05 level for different scenarios. The statistics have desired fraction of rejections, and the power is reasonable (it is actually slightly better when the residual noise from the alternative has a general form which actually breaks the theoretical assumptions of this paper). There are two important lessons to take from this part. First, the power is much more satisfactory in the tick by tick case, and thus we insist that the high frequency data user should make inference using all the data available. Second, depending on the simulated scenario, i.e. constant volatility, time-varying volatility including jumps or not, the related statistic behave slightly better than the other statistics. One should confirm accordingly the type of data at hand prior to choosing the related statistic to use. This will heavily depend on data pre-processing, such as controlling for diurnal pattern in volatility (see \cite{christensen2018diurnal}) and/or removing jumps. 

\smallskip
We now discuss about the validity of the sequence given in Section \ref{practicalguidance}. Table \ref{tableSimuCompa} reports the bias, standard deviation and RMSE of the eight concurrent volatility estimators for the three scenarios, i.e. no residual noise, residual noise and a mix of both aforementioned scenarios. As expected from Theorem \ref{theoremnoerror}, QMLEexp leads the cohort when there is no residual noise. QMLEerr has approximately a standard deviation $\sqrt{3}$ times as big as that of QMLEexp, which is in line with the theorem. The sequence using the tests' RMSE is very close to that of QMLEexp, although slightly bigger, which is due to the fact that the test to assess if the MMN is fully explained by some variables of the limit order book are (falsely) rejected one time over twenty. In case of non-zero residual noise, QMLEerr performs the best which is not surprising since the estimation procedure is residual noise robust. The sequence using the tests' RMSE is almost the same as that of QMLEerr, which is explained by the fact that the tests from this paper are (rightly) rejected 499 days over 500. On the other hand, QMLEexp suffers since it is not residual noise robust. Overall, the sequence using the tests is leading the group (in terms of RMSE) in the mix scenario (which is the most realistic case). In particular, note that the first step in the sequence using the original tests of \cite{ait2016hausman} implemented with QMLEerr as volatility estimator compared with RV has no distortion on the finite sample results since tests are rejected all the time. QMLEerr comes second while being relatively close from the sequence using the tests. E-QMLE is virtually tied with QMLEerr, which can be explained by the fact that they are asymptotically equivalent. The other estimators perform more poorly.

\section{Empirical study} \label{sectionEmpiricalStudy}
Our main dataset consists of one calendar month (April 2011) of trades and quotes for thirty one CAC 40 constituents traded on the Euronext NV. The data for individual stocks were obtained from the TAQ data and the Order book data.\footnote{The data were obtained through Reuters and provided by the Chair of Quantitative Finance
of Ecole Centrale Paris.} We keep quotes corresponding to best bid/ask price, which are often referred to as Level 1 data. To obtain the information of trade type, we implement Section 3.4 in \cite{muni2016reconstruction}.\footnote{The code is available on our websites. A comparison with the simple and popular Lee-Ready procedure introduced in \cite{lee1991inferring} can be consulted in Section 5 of the cited paper.} The timestamp is rounded to the nearest millisecond. 

\smallskip
To prevent from opening and closing effects, we restrict our dataset starting each day at 9:30am and ending at 4pm. We consider the data in tick time, for an average
of 3,000 daily trades and a quote/trade ratio bigger than 20. The most active days include more than 10,000 trades, whereas the less liquid days are around 500 trades. Descriptive statistics on the individual stocks are detailed on Table \ref{descriptivestatistics}. Using the regular QMLE restricted to $\phi=0$, we find that the MMN variance lies within $1.30 \times 10^{-9}$ and $5.60 \times 10^{-8}$, taking the value $1.63 \times 10^{-8}$ on average.

\smallskip
We first implement the tests of \cite{ait2016hausman} on the observed price on tick-by-tick data. The results can be found on Table \ref{tableaittests}. The six tests consistently indicate that we reject those tests almost all the time, indicating that there seems to be MMN at the highest frequency for the stocks and days considered. Accordingly, we report in Table \ref{tableerror} the measure of goodness of fit of several leading models: Roll, Glosten-Harris, signed timestamp, signed spread model, signed quoted depth, the order flow imbalance, a linear combination of all the aforementioned models and a non-linear signed spread. The signed spread model incontestably dominates with an astonishing proportion of variance explained estimated around 99\%. This dominance is in fact consistent across sampling frequencies, stocks and over time, although not fully reported. Finally, this measure stagnates when sampling at sparser frequencies, and we argue that this is because there is (almost) no remaining noise. 

\smallskip
Among the concurrent models, the goodness of fit of Roll model is very decent with a four fifth proportion of variance explained at the highest frequency and increasing when diminishing the sampling frequency. Yet this feature hints that the model cannot be considered as reasonably free of residual noise when using tick by tick data. The measure related to Glosten-Harris model is slightly bigger, suggesting that the information on the volume helps to improve the fit to a certain extent. Those estimated values are in line with the results discussed in the empirical study of \cite{li2016efficient}. The fit is not as good on other models. We have tried many other alternative models (such as linear combinations of the aforementioned models) but have not found any significant improvement in the fit of the signed spread model, as reported in Table \ref{tableerror}. In particular adding a Roll component to it was not found to improve the fit much. 

\smallskip
We further investigate if the signed spread model can fairly be considered as free from residual noise by implementing the two tick-by-tick-robust Hausman tests. For each individual stock and test, the fraction of rejection at the 0.05 level is reported on Table \ref{tablenewtests}. Although not reported, the results are very similar when using the other three statistics. Over the thirty one stocks, the averaged-across-tests fraction lies within 0.00 and 0.11 for twenty eight constituents (hereafter denoted as the main group, and which features a proportion of variance explained bigger than 99\%), whereas picking at 0.11, 0.16 and 0.42 for the three remaining components (henceforth referred as the minor group, which features a proportion of variance explained slightly below 95\%). The observed fraction of rejections of the main group constituents can be considered as reasonably close to the theoretical threshold 0.05 hence we can not reject the null hypothesis for them and this indicates that stocks from the main group can be considered as fairly free from residual noise. On the contrary the rejection is clear for the three stocks from the minor group. 

\smallskip
An example of estimated efficient price can be seen on Figure \ref{figure_estimatedefficient}. To explore the properties of the estimated price, we recognize it as the given observed price to be tested in \cite{ait2016hausman}. Although their tests are by nature related to our tests, the estimators they are using differ from the ones considered in our work, and thus their tests can be regarded as a sensibly independent check on the efficiency of the estimated price. The fraction of rejections of their null hypothesis at the 0.05 level can be consulted on Table \ref{tableaittests}. When restricting to the main group constituents, the six tests range from 0.05 to 0.09 and are equal to 0.06 on average. When considering the minor group stocks, the same tests range from 0.33 to 0.42. This largely corroborates the fact that the main group stocks are most likely free from residual noise whereas the minor group constituents cannot be considered as such. One feature common to the minor group constituents, namely
France Telecom, Louis Vuitton and Schneider Electric, is that they are stocks with large ticks and small spread, which is almost always equal to 1 tick. This feature can be seen on Table \ref{descriptivestatistics}, as the three stocks share the top 3 in terms of smallest spread, and part of the top 5 when looking at the smallest ratio of price over tick size. Even for those stocks, our findings strongly indicate that the estimated price is much closer from efficiency than the observed price on which the tests are rejected 100\% of the time.

\smallskip
Another way to question the efficiency of the estimated price consists in inspecting the first lag of the autocorrelation function and the visual "signature plot" procedure of \cite{andersen2000great} (see also \cite{patton2011data}). This can be seen on Figure \ref{figure_ACF}-\ref{figure_signatureplot}. A satisfactory amelioration of the first lag of the autocorrelation function is noticeable, as it averages 0.02 when looking at the estimated price whereas -0.28 when taking the observed price. The signature plot is also acceptable as it is relatively flat.

\smallskip
Finally, the maximum likelihood estimation in both settings delivers very similar (the difference is at most equal to $10^{-3}$ even when considering the three stocks from the minor group) and stable estimates of the parameter which lies systematically between 0.60 and 0.90 with an average around 0.79 and a standard deviation slightly above 0.03 when sampling at the highest frequency. The estimation is also stable across sampling frequencies as the average values are $0.77$ and $0.76$ when considering the 15-second and 30-second frequency, respectively. Moreover, Figure \ref{figure_estimatedparameter} documents that the daily estimates averaged across stocks are also relatively stable over time. 

\smallskip
Note that we can also consider as in \cite{chaker2017high} (p. 15) the finite sample correction by instrumental variables. This slightly shifts estimation of the parameter towards the origin, with a decrease of magnitude one quarter of its value in the estimates. As it can improve finite sample properties, we implemented and chose to work with such finite sample correction, including in our implemented tests. Finally, we have considered variance estimators based on raw returns instead of estimated price returns for stability reason.

\section{Conclusion}
The paper introduces tests to assess if the market microstructure noise can be fully explained by the informational content of some variables from the limit order book. Two novel quasi-maximum likelihood estimators are extensively studied in the development. Subsequently, based on a common procedure the paper proposes an efficient price estimator.

\smallskip
We emphasize that the method can be easily implemented to assist anyone who is working with high frequency data. The empirical study should be taken as a reference for repeating the exercise, i.e. first testing among a class of candidates and then choosing one specific model based on the measure of goodness of fit. We hope this provides an alternative and reliable solution to the common dilemma between sparsing and using sophisticated noise-robust estimators.

\smallskip 
We also call attention to the fact that when the market microstructure noise is fully explained by the limit order book, other quantities beyond quadratic variation, such as pure integrated volatility (by truncation), integrated powers of volatility, high-frequency covariance or even volatility of volatility can be estimated following the same procedure as investigated in our recent work \cite{clinet2017estimation}.    

\smallskip
Finally, although we have checked that there is no major distortion of our tests in finite sample and that they can be useful to improve the precision of volatility estimation, challenging and interesting avenues for future research include a possible improvement of the method checking whether, or not, the testing procedure suffers from the classical post-model selection issue as presented in \cite{leeb2005model}.

\pagebreak
\section*{APPENDIX}
\section{Definition of supplementary variance estimators when observations are regular}
\label{supplementaryvarest}
In this section, we provide supplementary variance estimators in the case of regular observations. We consider the three aforementioned scenarios, i.e.
\begin{enumerate}[(i)]
    \item constant volatility
    \item time-varying volatility and no price jump
    \item time-varying volatility and price jump
\end{enumerate}
In the case (i), we have from Theorem \ref{theoremnoerror} that 
$AVAR \big(\widehat{\sigma}_{exp}^2 - \widehat{\sigma}_{err}^2 \big) = 4  \sigma_0^4$. This can be simply estimated by 
\begin{eqnarray}
\widehat{V}_3 & = & 4   (\widehat{\sigma}_{exp}^2)^2.
\end{eqnarray} 
Under (ii), we have $AVAR \big(\widehat{\sigma}_{exp}^2 - \widehat{\sigma}_{err}^2 \big) = 4 T^{-1} \int_0^T \sigma_s^4 ds$ which can be estimated by: 
\begin{eqnarray}
\widehat{V}_4 & = & \frac{4 n}{3 T^2} \sum_{i=1}^{n} \Delta \widehat{X}_{i}^4 \text{, with } \Delta \widehat{X}_{i} = \widehat{X}_{t_{i}} - \widehat{X}_{t_{i-1}},
\end{eqnarray}
where $\widehat{X}_{t_i}$ was previously defined in (\ref{efficientpriceestimator}). When assuming (iii), we have $AVAR  \big(\widehat{\sigma}_{exp}^2 - \widehat{\sigma}_{err}^2 \big) = 4 T^{-1} \big\{ \int_0^T \sigma_s^4 ds + \sum_{0 < s \leq T} \Delta J_s^2 (\sigma_s^2 + \sigma_{s-}^2)\big\}$. If we introduce $k \rightarrow \infty$ such that $k \Delta \rightarrow 0$ and $u = \widetilde{\alpha} \Delta^{\omega}$ with $0 < \omega < 1/2$, and $\tilde{\alpha} > 0$, the asymptotic variance can be estimated via
\begin{eqnarray}
\widehat{V}_5 & = & \frac{4 }{T} \Bigg\{ \frac{1}{3 \Delta} \sum_{i=1}^{n} \Delta \widehat{X}_{i}^4 \mathbf{1}_{\{\mid \Delta \widehat{X}_{i} \mid \leq u \}} + \sum_{i=k+1}^{n- k} \Delta \widehat{X}_{i}^2 \mathbf{1}_{\{\mid \Delta \widehat{X}_{i} \mid > u \}} 
\big( \breve{\sigma}^{2}_{t_{i}}  +\breve{\sigma}^{2}_{t_{i}-} \big) \Bigg\} \text{ , where}\\
\breve{\sigma}^{2}_{t_i} & = & \frac{1}{k \Delta} \sum_{j=i+1}^{i+k} \Delta \widehat{X}_j^2  \mathbf{1}_{\{\mid \Delta \widehat{X}_{j} \mid \leq u \}} \text{ , } \breve{\sigma}^{2}_{t_i-}  =  \breve{\sigma}^{2}_{t_{i-k-1}}. \nonumber
\end{eqnarray}
The estimator $\widehat{V}_3$ is based on the truncation method considered in \cite{mancini2009non}. The three variances estimators $\widehat{V}_i$ for $i = 3,4,5$ are identical to the ones introduced in \cite{ait2016hausman}, up to a scaling factor $T$. It is due to the fact that the authors scale their Hausman test statistics by $\Delta_n^{-1}$ whereas we used $N$ instead. Those three estimators satisfy the conditions of Proposition \ref{propAVAR} and Corollary \ref{testConsistency}. In the corresponding proofs, we also show for the case $i=3,4,5$.

\section{Expression of the asymptotic variance terms defined in Theorem \ref{theoremnoerror}}
\label{avarTheoremSmallNoise}
We recall that 
\begin{eqnarray*}
\overline{\sigma}_0^2 & = & T^{-1}\l(\int_0^T{\sigma_s^2ds} + \sum_{0 < s \leq T}\Delta J_s^2\r)$, $\tilde{\eta}_0 = \l(T^{-1}\int_0^T{\alpha_s^{-1}ds}\r) \eta_0,\\ \widetilde{\calk} & = & \l(T^{-1}\int_0^T{\alpha_s^{-1}ds}\r)^2\calk,\\\ \calq & = & T^{-1}\int_0^T{\alpha_s^{-1}ds}\l\{\int_0^T{\sigma_s^4 \alpha_sds} + \sum_{0 < s \leq T} \Delta J_s^2 (\sigma_s^2\alpha_s + \sigma_{s-}^2\alpha_{s-}) \r\}.
\end{eqnarray*}
Moreover, we also have 
\begin{eqnarray}
\phi_0 & = & 1 - \inv{2\tilde{\eta}_0} \l\{\sqrt{ \overline{\sigma}_0^2(4\tilde{\eta}_0+ \overline{\sigma}_0^2)} - \overline{\sigma}_0^2\r\}, \label{eqphi0}\\
\gamma_0^2 & = & \half \l\{2\tilde{\eta}_0+\overline{\sigma}_0^2 + \sqrt{\overline{\sigma}_0^2(4\tilde{\eta}_0+\overline{\sigma}_0^2)}\r\}, \label{eqgamma0}
\end{eqnarray}
so that $\overline{\sigma}_0^2 = \gamma_0^2(1-\phi_0)^2$ and $\tilde{\eta}_0 = \gamma_0^2\phi_0$. The components of $\textbf{V}$ are expressed as
\beas 
\textbf{V}_{11}(\overline{\sigma}_0^2, \tilde{\eta}_0, \widetilde{\calk}) &=& 4\widetilde{\calk} + 12\tilde{\eta}_0^2 + 8 \tilde{\eta}_0 \overline{\sigma}_0^2,\\
\textbf{V}_{13}(\gamma_0^2,\phi_0) &=&  4 \gamma_0^4 \phi_0^2 (1-\phi_0^2), \\
\textbf{V}_{15}(\gamma_0^2,\phi_0,\widetilde{\calk}) &=&  2\widetilde{\calk}T + 4\gamma_0^4T\phi_0 \\
\textbf{V}_{33}(\calq, \gamma_0^2,\phi_0) &=& \frac{-2\phi_0(3\phi_0^5-10\phi_0^3-2\phi_0^2+\phi_0-6)\calq}{(1-\phi_0^2)^3T}  + \frac{8 \gamma_0^4 \phi_0 (1-\phi_0)^3(\phi_0^2+\phi_0+1)}{(1+\phi_0)^3},\\
\textbf{V}_{35}(\calq, \gamma_0^2,\phi_0)&=& \frac{2\phi_0^2(\phi_0^3+\phi_0^2-2\phi_0-4)\calq}{(1-\phi_0)^2(1+\phi_0)^3}  - \frac{2 \gamma_0^4 T \phi_0 (1-\phi_0)^2 (\phi_0^2 + \phi_0+ 2)}{(1+\phi_0)^3}\\
\textbf{V}_{55}(\calq, \gamma_0^2,\phi_0,\widetilde{\calk}) &=& \frac{\phi_0^2(2-\phi_0^2)\calq T}{(1-\phi_0^2)^2} + \widetilde{\calk} T^2 + \frac{ \gamma_0^4 T^2 \phi_0 (1+\phi_0^2)(\phi_0^2 + 3\phi_0+4)}{(1+\phi_0)^3}.
\eeas 

\begin{remark}
When the volatility is constant, the observation times are regular and there are no jumps in the price process, we have $\overline{\sigma}_0^2 = \sigma^2$, $\calq = \sigma^4T$, and thus replacing $\phi_0$ and $\gamma_0^2$ using (\ref{eqphi0}) and (\ref{eqgamma0}), we have that the variance matrix for ($\widehat{\sigma}_{err}^2, \widehat{a}_{err}^2$) is of the form

\beas 
\left(\begin{matrix}   6\sigma^4 + \textbf{V}_{33} &-2\sigma^4  +\textbf{V}_{35}\\ 
-2\sigma^4 +\textbf{V}_{35} &\sigma^4 +\textbf{V}_{55} 
\end{matrix}\right), 
\eeas 
equal to 
\beas\left(\begin{matrix}  2\sigma^4+ 4\sqrt{\sigma^6(4\eta_0+\sigma^2)}   &-(\sigma^4 +2\sigma^2\eta_0 +\sqrt{\sigma^2(4\eta_0 + \sigma^2)})T \\ 
-(\sigma^4 +2\sigma^2\eta_0 +\sqrt{\sigma^2(4\eta_0 + \sigma^2)})T &(2\eta_0+\sigma^2)\l(2\eta_0+\sigma^2 +\sqrt{\sigma^2(4\eta_0+\sigma^2)}\r)T^2 + \calk T^2
\end{matrix}\right),
\eeas 
which corresponds to the limit variance of Theorem 2 p.371 of \cite{ait2005often} in the abovementioned framework.
\end{remark}

For the information part, we define for any $k \in \naturels$ the quantity $$\tilde{\rho}_ k = \half \esp \l[ \frac{\partial W_0(\theta_0)}{\partial \theta}\frac{\partial W_k(\theta_0)^T}{\partial \theta} + \frac{\partial W_k(\theta_0)}{\partial \theta}\frac{\partial W_0(\theta_0)^T}{\partial \theta}  \r],$$ along with the matrices 
\begin{eqnarray*}
U_{\theta_0} & = & 2(\tilde{\rho}_0 -\tilde{\rho}_1),\\ 
P_{\theta_0} & = & 2(1+\phi_0)^{-1}\l\{\tilde{\rho}_0 - (1-\phi_0)\sum_{k=1}^{+\infty}\phi_0^{k-1}\tilde{\rho}_k \r\}.
\end{eqnarray*}
Then, the asymptotic variance and covariance terms can be expressed as
\beas 
\textbf{V}_{22}(\theta_0, \tilde{\eta}_0) &=& 3\tilde{\eta}_0TU_{\theta_0}^{-1},\\
\textbf{V}_{24}(\theta_0, \gamma_0^2, \phi_0,\sum_{0<s\leq T}\Delta J_s^2) &=& 2 \gamma_0^2T(1-\phi_0^2)^{-1} U_{\theta_0}^{-1}\big((1-\phi_0)\big\{(\phi_0^4-4\phi_0^3+5\phi_0^2-\phi_0+1)\tilde{\rho}_0 \\ & & + (\phi_0^3 -\phi_0^2+3\phi_0-1)\tilde{\rho}_1 \big\} \\
&+& 2\phi_0(1-\phi_0)^2 \tilde{\rho}_2 + (2-\phi_0)(1-\phi_0)^4\sum_{k=2}^{+\infty}\phi_0^{k}\tilde{\rho}_k \big)P_{\theta_0}^{-1}\\
&-& \l[\frac{2U_{\theta_0}^{-1}}{1-\phi_0^2} \l\{(1-\phi_0)^3\tilde{\rho}_0 - (1-\phi_0)^2\tilde{\rho}_1 +(1-\phi_0)^3\sum_{k=2}^{+\infty}\phi_0^{k}\tilde{\rho}_k\r\}  P_{\theta_0}^{-1}- U_{\theta_0}^{-1}\r]\\ & & \times \sum_{0<s\leq T}\Delta J_s^2,\\
\textbf{V}_{44}(\theta_0, \gamma_0^2, \phi_0,\sum_{0<s\leq T}\Delta J_s^2) &=& \gamma_0^2T \l(P_{\theta_0}^{-1} - (1-\phi_0)^2U_{\theta_0}^{-1} \r) \\&-&\l[ \frac{2(1-\phi_0)^2}{(1-\phi_0^2)^2}P_{\theta_0}^{-1}\l(\frac{\tilde{\rho}_0}{1-\phi_0^2}+  \sum_{k=1}^{+\infty}\l\{\frac{2\phi_0^k}{1-\phi_0^2} -k\phi_0^{k-1} \r\}\tilde{\rho}_k\r)P_{\theta_0}^{-1} - U_{\theta_0}^{-1} \r]\\
& & \times \sum_{0<s\leq T}\Delta J_s^2. \\
\eeas
Finally, introducing $B_{\theta_0} = \sum_{0<s\leq T} \sum_{k=1}^{N_n} \phi_0^{|k-i_n(s)|}\frac{\partial \mu_k(\theta_0)}{\partial \theta} \Delta J_s$, and $i_n(s)$ is the only index such that $t_{i-1}^n < t \leq t_i^n$ the bias terms are expressed as 
\beas 
B_{\theta_0,exp} = U_{\theta_0}^{-1} B_{\theta_0} \textnormal{ and }
B_{\theta_0,err} = P_{\theta_0}^{-1} B_{\theta_0}.
\eeas 
In particular, under $\calh_0$, $\phi_0=0$ and $P_{\theta_0} = U_{\theta_0}$ so that
$$ B_{\theta_0,exp} = B_{\theta_0,err} = \sum_{0<s\leq T}  \frac{\partial \mu_{i_n(s)}(\theta_0)}{\partial \theta} \Delta J_sU_{\theta_0}^{-1}.$$
\section{Proofs}

\subsection{Simplification of the problem} \label{sectionSimplify}
We give an additional assumption which is harmless (see e.g. the discussion in \cite{clinet2018efficient}, Section A.1). Note that we can apply Girsanov theorem as all the assumptions on the information also hold on the risk-neutral probability.

\textbf { (H) } We have $b = \tilde{b} = 0$. Moreover $\sigma$, $\sigma^{-1}$, $\tilde{\sigma}^{(1)}$, $(\tilde{\sigma}^{(1)})^{-1}$, $\tilde{\sigma}^{(2)}$, $(\tilde{\sigma}^{(2)})^{-1}$, $\alpha$, $\alpha^{-1}$ are bounded. Given an \textit{a priori} number $\gamma >0$, we also have $\sup_{0 \leq i \leq N_n} U_i^n \leq n^{\gamma}$. \\

From now on, to avoid confusion as much as possible, we explicitly write the exponent $n$ in the expressions $Q_i^n$, $U_i^n$, etc. Note also that, by virtue of Lemma 14.1.5 in \cite{jacod2011discretization}, recalling the definition $\pi_t^n := \sup_{i \geq 1} t_{i}^n-t_{i-1}^n$, and $N_n(t) = \sup \{ i\in \naturels - \{0\} | t_{i}^n  \leq t \}$ we have 

\bea 
c >0 \implies n^{1-c}\pi_t^n \to^\proba 0,
\label{eqStepsize}
 \eea 
We sometimes refer to the continuous part of $X_t$ defined as 
\bea 
\tilde{X}_t :=  X_t - J_t.
\eea 

\medskip

We define $\calu := \sigma \l\{U_i^n | i,n \in \naturels \r\} \vee \sigma \l\{ \alpha_s | 0 \leq s \leq T\r\}$ the $\sigma$-field that generates the observation times and which is independent of $X$ and $Q$. We will often have to use the conditional expectation $\esp[.|\calu]$, that we hereafter denote for convenience by $\esp_{\calu}$. We also define the discrete filtration $\calg_i^n := \calf_{\ti{i}}^X \vee \calu$, and the continuous version $\calg_t := \calf_t^X \vee \calu$. Note that by independence from $\alpha$, $X$ admits the same It\^{o} semi-martingale dynamics in the extension $\G = (\calg_t)_{0 \leq t \leq T}$. 

\medskip 

Finally, all along the proofs, we recall that we write $N_n$ in place of $N_n(T)$, we define $K_n = N_n^{1/2 + \delta}$, for some $\delta >0$ to be adjusted, and we let $K$ be a positive constant that may vary from one line to the next. 

\subsection{Estimates for $\Omega^{-1}$}

We start this appendix by giving some useful estimates for the matrix $\Omega^{-1} := [\omega^{i,j}]_{i,j}$ which was defined in (\ref{defOmega}). Let us define $u_0 = \sqrt{\frac{\sigma^2T}{a^2}}$. Note that $\frac{\partial u_0}{\partial \sigma^2}  = \frac{u_0}{2\sigma^2}$, and  $\frac{\partial u_0}{\partial a^2} = -\frac{u_0}{2a^2}$. In all this section, the expression $O(1)$ means a (possibly random) function $f : (i,j,n,\xi) \to f(i,j,n,\xi) $, where $\xi = (\sigma^2,a^2,\theta) \in \Xi$, which is bounded uniformly in all its arguments and in $\omega \in \Omega$ under the constraint $i, j \leq N_n$, and which is $C^{\infty}$ on the compact $\Xi$, such that its partial derivatives $\frac{\partial^{\boldsymbol{\alpha}} O(1)}{\partial^{\boldsymbol{\alpha}}\xi}$ are also bounded. In particular, for any multi-index $\boldsymbol{\alpha}$ we have the useful property $\frac{\partial^{\boldsymbol{\alpha}} O(1)}{\partial^{\boldsymbol{\alpha}}\xi} = O(1)$. Finally, we define for $n \in \naturels$ the function $g_n : k \in \{1,...,2N_n\} \to k \wedge (2N_n-k)$. Note that $g_n(k)$ is always dominated by $N_n$.

\begin{lemma*} \label{expansionOmega} (expansions for $\Omega^{-1}$)
There exists $s >0$, such that uniformly in $(i,j) \in \{1,..., N_n\}^2$,
\small

\beas  
\omega^{i,j} &=& \frac{\sqrt{N_n}}{2u_0a^2}\l\{ \l(1 - \frac{u_0^3}{24N_n^{3/2}}|i-j| + O\l(N_n^{-1}\r) \r)e^{-u_0\frac{|i-j|}{\sqrt{N_n}}} -\l(1 - \frac{u_0^3}{24N_n^{3/2}}g_n(i+j)+ O\l(N_n^{-1}\r) \r)e^{-u_0\frac{g_n(i+j)}{\sqrt{N_n}}}  \r\}\\ 
&+& \bo{e^{-s\sqrt{N_n}}}.
\eeas 
\normalsize



\end{lemma*}

\begin{proof} 
From the relation $\omega^{i,j} = \omega^{N_n-i,N_n-j}$, it is sufficient to show the result with $2 \leq i + j  \leq N_n$. Under the change of variable 
\beas  
\phi = 1 - \inv{2a^2} \l\{\sqrt{ \sigma^2 \Delta_{N_n}(4a^2+\sigma^2\Delta_{N_n})} -\sigma^2\Delta_{N_n}   \r\} \textnormal{ and } \gamma^2 = \half \l\{2a^2+\sigma^2\Delta_{N_n} + \sqrt{\sigma^2\Delta_{N_n}(4a^2+\sigma^2\Delta_{N_n})}\r\},
\eeas 
we recall the expression of $\omega^{i,j}$
\bea  
\omega^{i,j} = \frac{\phi^{|i-j|}- \phi^{i+j} - \phi^{2N_n-i-j+2}+ \phi^{2N_n- |i-j|+2}}{\gamma^2(1-\phi^2)(1-\phi^{2N_n+2})},
\label{defCoeffOmega}
\eea 
taken from \cite{xiu2010quasi}, eq. (28) p 245. By a short calculation we also have the expansions
\bea 
\phi = 1 - \frac{u_0}{\sqrt{N_n}} + \frac{u_0^2}{2N_n} - \frac{u_0^3}{8N_n^{3/2}} + O\l(N_n^{-5/2}\r),
\eea 
and 
\bea 
\gamma^2 = a^2 + \frac{a^2u_0}{\sqrt{N_n}}  + \bo{ N_n^{-1}}.
\eea 
Now, for the first term in the numerator, we can write 
\beas 
\phi^{|i-j|} &=& \textnormal{exp}\l\{|i-j|\textnormal{log}\l(1 - \frac{u_0}{\sqrt{N_n}} +  \frac{u_0^2}{2N_n} - \frac{u_0^3}{8N_n^{3/2}} + O\l(N_n^{-5/2}\r)\r)\r\} \\
&=& \textnormal{exp}\l\{|i-j|\l( - \frac{u_0}{\sqrt{N_n}} + \frac{u_0^2}{2N_n} - \frac{u_0^3}{8N_n^{3/2}} -\half \l(-\frac{u_0}{\sqrt{N_n}} + \frac{u_0^2}{2N_n} -\frac{u_0^3}{8N_n^{3/2}}\r)^2  \r.\r. \\
&& \l.\l.+\inv{3}\l(-\frac{u_0}{\sqrt{N_n}} + \frac{u_0^2}{2N_n} -\frac{u_0^3}{8N_n^{3/2}}\r)^3 + O\l(N_n^{-2}\r)\r)\r\} \\
&=& \textnormal{exp}\l\{- \frac{u_0}{\sqrt{N_n}} |i-j| - \frac{u_0^3}{24N_n^{3/2}}|i-j|  +\underbrace{O\l(|i-j|N_n^{-2}\r)}_{O(N_n^{-1})}\r\} \\
&=& \textnormal{exp}\l\{- \frac{u_0}{\sqrt{N_n}} |i-j|\r\}\l(1 - \frac{u_0^3}{24N_n^{3/2}}|i-j| + O\l(N_n^{-1}\r) \r).
\eeas
Moreover, as $i+j \leq N_n$, similar calculations lead to the estimates
\beas 
\phi^{i+j} = \textnormal{exp}\l\{- \frac{u_0}{\sqrt{N_n}} (i+j)\r\}\l(1 - \frac{u_0^3}{24N_n^{3/2}}(i+j)+ O\l(N_n^{-1}\r) \r),
\eeas 
\beas 
\phi^{2N_n - i - j} = O\l(\textnormal{exp}\l\{- u_0\sqrt {N_n}\r\}\r),
\eeas 
and finally 
\beas 
\phi^{2N_n- |i-j|+2} = O\l(\textnormal{exp}\l\{- \frac{3u_0}{2}\sqrt {N_n}\r\}\r).
\eeas 
We also have the expansion 
\beas 
\gamma^2(1-\phi^2)(1-\phi^{2N_n+2}) = 2a^2 \frac{u_0}{\sqrt {N_n}} + \bo{N_n^{-3/2}}
\eeas 
by direct calculation. Overall we thus get 
\beas 
\omega^{i,j} = \frac{\sqrt{N_n}}{2u_0a^2}\l\{ \l(1 - \frac{u_0^3}{24N_n^{3/2}}|i-j| + O\l(N_n^{-1}\r) \r)e^{-u_0\frac{|i-j|}{\sqrt{N_n}}} -\l(1 - \frac{u_0^3}{24N_n^{3/2}}(i+j)+ O\l(N_n^{-1}\r) \r)e^{-u_0\frac{i+j}{\sqrt{N_n}}}  \r\},
\eeas
up to the terms related to $\phi^{2N_n - i - j}$ and $\phi^{2N_n- |i-j|+2}$, which we gather in $\bo{e^{-s\sqrt{N_n}}}$.

\end{proof} 

For a matrix $A = \l[ a_{i,j}\r]_{1 \leq i \leq j \leq n}\in \reels^{N_n \times N_n}$, we associate the matrix  $\dot{A} = [\dot{a}_{i,j}]_{0 \leq i \leq N_n ,1 \leq j \leq N_n } \in \reels^{(N_n+1)\times N_n }$ and $\ddot{A} = [\ddot{a}_{i,j}]_{0 \leq i \leq N_n ,0 \leq j \leq N_n } \in \reels^{(N_n+1)\times (N_n+1) }$ whose components respectively satisfy 

\bea 
\dot{a}_{i,j} = a_{i+1,j} - a_{i,j}, 
\eea 
and
\bea 
\ddot{a}_{i,j} = \dot{a}_{i,j+1} - \dot{a}_{i,j} = a_{i+1,j+1} - a_{i,j+1} + a_{i,j} - a_{i+1,j} ,
\eea
with the convention $a_{i,j} = 0$ when $i=0$ or $j=0$. We recall the following lemma taken from \cite{clinet2018statistical}.

\begin{lemma*}  \label{lemmaTransfo}
Let $y,z \in \reels^{N_n+1}$, with $y = (y_0,...,y_{N_n})^T$, $z = (z_0,...,z_{N_n})^T$, and define $$\Delta y = (\Delta y_1,...,\Delta y_{N_n}) := \l(y_1-y_0,...,y_{N_n} - y_{N_n-1}\r)^T \in \reels^{N_n},$$ and $\Delta z$ the same way.  Then we have the by-part summation identities

$$\Delta y^T A \Delta z = - y^T \dot{A} \Delta z = y^T \ddot{A} z.$$


\end{lemma*}




We define accordingly $\dot{\Omega}^{-1}$, and $\ddot{\Omega}^{-1}$. In the next lemmas we derive some estimates for such matrices.

\begin{lemma*} \label{expansionOmegaPoint} (expansions for $\dot{\Omega}^{-1}$) 
We have the approximation, uniform in $i \in \{0,...,N_n\}$, $j \in \{1,...,N_n\} $ 

\beas  
\dot{\omega}^{i,j} &=& - \frac{1}{2a^2}\l\{ \l(\textnormal{sgn}(i-j) -\frac{u_0}{2\sqrt{N_n}} - \frac{u_0^3}{24N_n^{3/2}}|i-j| + \bo{N_n^{-1}}\r) e^{-u_0\frac{|i-j|}{\sqrt{N_n}}} \r. \\
&-& \l.\l(1 -\frac{u_0}{2\sqrt{N_n}} - \frac{u_0^3}{24N_n^{3/2}}g_n(i+j) + \bo{N_n^{-1}}\r) e^{-u_0\frac{g_n(i+j)}{\sqrt{N_n}}} \r\} + \bo{e^{-s\sqrt{N_n}}}, 
\eeas  
where $\textnormal{sgn}(x) = \mathbb{1}_{\{x \geq 0\}} - \mathbb{1}_{\{x < 0\}} $.

\end{lemma*}

\begin{proof}
Once again we assume without loss of generality that $ i+j \leq N_n$. From Lemma \ref{expansionOmega} and the definition of $\omega^{i,j}$, some calculation gives, up to the term $\bo{e^{-s\sqrt{N_n}}}$,

\beas 
\dot{\omega}^{i,j} &=&  \frac{\sqrt{N_n}}{2u_0a^2}\l\{ \l(1 - \frac{u_0^3}{24N_n^{3/2}}|i+1-j| + \bo{N_n^{-1}}\r)  \l(e^{-u_0\frac{|i+1-j|}{\sqrt{N_n}}}-e^{-u_0\frac{|i-j|}{\sqrt{N_n}}}\r)  \r. \\
&&- \l.\l(1  - \frac{u_0^3}{24N_n^{3/2}}(i+j+1) + \bo{N_n^{-1}}\r)  \l(e^{-u_0\frac{i+j+1}{\sqrt{N_n}}}-e^{-u_0\frac{i+j}{\sqrt{N_n}}}\r) \r\}\\
&=&\frac{\sqrt{N_n}}{2u_0a^2}\l\{ \l(1 - \frac{u_0^3}{24N_n^{3/2}}|i-j| + \bo{N_n^{-1}}\r)e^{-u_0\frac{|i-j|}{\sqrt {N_n}}}  \l(e^{-u_0\frac{\textnormal{sgn(i-j)}}{\sqrt{N_n}}}-1\r)  \r. \\
&&- \l.\l(1  - \frac{u_0^3}{24N_n^{3/2}}(i+j) + \bo{N_n^{-1}}\r)e^{-u_0\frac{i+j}{\sqrt {N_n}}}  \l(e^{-\frac{u_0}{\sqrt n}}-1\r) \r\}\\
&=& \frac{\sqrt{N_n}}{2u_0a^2}\l\{ \l(1 - \frac{u_0^3}{24N_n^{3/2}}|i-j| + \bo{N_n^{-1}}\r)e^{-u_0\frac{|i-j|}{\sqrt {N_n}}}  \l(e^{-u_0\frac{\textnormal{sgn}(i-j)}{\sqrt {N_n}}}-1\r)  \r. \\
&&- \l.\l(1  - \frac{u_0^3}{24N_n^{3/2}}(i+j) + \bo{N_n^{-1}}\r)e^{-u_0\frac{i+j}{\sqrt {N_n}}}  \l(e^{-\frac{u_0}{\sqrt n}}-1\r) \r\}\\ 
&=& \frac{\sqrt{N_n}}{2u_0a^2}\l\{ \l(1 - \frac{u_0^3}{24N_n^{3/2}}|i-j| + \bo{N_n^{-1}}\r)e^{-u_0\frac{|i-j|}{\sqrt {N_n}}}  \l(-u_0\frac{\textnormal{sgn}(i-j)}{\sqrt {N_n}}+  \fruu +\bo{N_n^{-3/2}}\r)  \r. \\
&&- \l.\l(1  - \frac{u_0^3}{24N_n^{3/2}}(i+j) + \bo{N_n^{-1}}\r)e^{-u_0\frac{i+j}{\sqrt {N_n}}}\l(-\fru + \fruu + \bo{N_n^{-3/2}}\r) \r\},\\
\eeas 
and expanding the terms in parenthesis we get the result. 

\end{proof}

From the previous lemma we deduce by similar computations an expansion for $\ddot{\Omega}^{-1}$.

\begin{lemma*} \label{expansionOmegaPPoint} (expansions for $\ddot{\Omega}^{-1}$) 
If $i \neq j$, we have the approximation, uniform in $(i,j) \in \{0,...,N_n\}^2$ 

\bea  
\ddot{\omega}^{i,j} =  -\frac{u_0}{2a^2\sqrt{N_n}}\l(1 + O\l(N_n^{-1/2}\r)\r) \l\{e^{-u_0 \frac{|i-j|}{\sqrt {N_n}} } + e^{-u_0\frac{g_n(i+j)}{\sqrt{N_n}}}\r\} +\bo{e^{-s\sqrt{N_n}}}. 
\label{expansionOmegaij}
\eea  
Moreover, uniformly in $i \in \{0,...,N_n\}$, 

\bea 
\ddot{\omega}^{i,i} = \frac{1}{a^2}\l(1- \frac{u_0}{2\sqrt {N_n}} + O\l(N_n^{-1}\r)\r) +  \frac{u_0}{2a^2\sqrt{N_n}}\l(1 + O\l(N_n^{-1/2}\r)\r) e^{-u_0\frac{g_n(2i)}{\sqrt {N_n}}} + \bo{e^{-s\sqrt{N_n}}}.
\label{expansionOmegaii}
\eea

\end{lemma*}

 \subsection{Estimates for the efficient price $X$}

Hereafter, we adopt the same notation convention as in \cite{clinet2018efficient}, Section A.3. For a process $V$, and $t \in [0,T]$ we write  $\Delta V_t = V_t - V_{t-}$. We also write $\Delta V_{i}^n := V_{t_i^n} -V_{t_{i-1}^n}$. Finally, for interpolation purpose we sometimes write the continuous version $\Delta V_{i,t}^n := V_{t_i^n \wedge t} -V_{t_{i-1}^n \wedge t} $. Let us define
\bea 
\zeta_{i,t}^n := (\Delta \tilde{X}_{i,t}^n)^2 -\sigma_{t_{i-1}^n}^2(t_{i}^n\wedge t-t_{i-1}^n \wedge t), \textnormal{ and } \overline{\zeta}_{i,t}^n := \esp \l[\zeta_{i,t}^n | \calg_{i-1}^n\r].   
\eea 


We recall the following standard estimates.

\begin{lemma*}\label{lemmaEstimateX}
We have, for some constant $K >0$ independent of $i$,
\bea
 \esp \l[ \l. \sup_{t \in ]\ti{i-1},\ti{i}]}|\Delta \tilde{X}_{i,t}^n|^p \r| \calg_{i-1}^n \r]  \leq  Kn^{-p/2}(U_i^n)^{p/2},
 \label{eqDeltaX}
\eea 

\bea 
\l|\tilde{\zeta}_{i,t}^n \r| \leq K n^{-3/2} (U_i^n)^{3/2} , 
\label{eqZetaTilde}
\eea 

\bea
 \esp \l[ \l. \l(\zeta_{t,i}^n\r)^p \r| \calg_{i-1}^n \r]  \leq  Kn^{-p} (U_i^n)^{p},
 \label{eqZeta2}
\eea 
\bea
 \esp \l[ \l. \l|\int_{\ti{i-1} \wedge t }^{\ti{i} \wedge t}\sigma_s^2ds - \sigma_{t_{i-1}^n}^2(t_{i}^n\wedge t-t_{i-1}^n \wedge t)\r|^p \r| \calg_{i-1}^n \r]  \leq  Kn^{-3p/2}(U_i^n)^{3p/2}.
 \label{eqZeta3}
\eea 

\end{lemma*}

\subsection{Estimates for the information part}


In this section we derive some asymptotic results for the information part. We define for $\xi = (\sigma^2,a^2,\theta) \in \Xi$ 
\bea 
\G_n(\xi) = ( \mu(\theta_0)-\mu(\theta))^T  \Omega^{-1}  ( \mu(\theta_0)-\mu(\theta)),
\label{defGnH1}
\eea 
along with the asymptotic fields

\bea 
\G_{\infty,1}(\xi) = -\frac{u_0}{2a}\l\{\rho_0(\theta) + 2\sum_{k=1}^{+\infty}{\rho_k(\theta)}\r\},
\label{defG1}
\eea 
and
\bea 
\G_{\infty,2}(\xi) = \frac{\rho_0(\theta)}{a^2}.
\label{defG2}
\eea 

By Lemma \ref{expansionOmegaPPoint}, we have the following matrix decomposition for $\ddot{\Omega}^{-1}$. Let 

\bea 
\cale^- = \l[ e^{-u_0 \frac{|i-j|}{\sqrt {N_n}}}\r]_{0 \leq i,j \leq N_n}, \textnormal{ and } \cale^+ = \l[ e^{-u_0 \frac{g_n(i+j)}{\sqrt{N_n}}}\r]_{0 \leq i,j \leq N_n}.
\label{defCale}
\eea 
Then we have 

\bea 
\ddot{\Omega}^{-1} = \inv{a^2}\I_{N_n} -\frac{u_0}{2a^2\sqrt{N_n}}\l(1+ \bo{N_n^{-1/2}}\r) \l\{ \cale^- - \cale^+ \r\} + \bo{e^{-s\sqrt{N_n}}}\J_{N_n}, 
\label{decompositionMatrixDdot}
\eea 
where $\I_{N_n}, \J_{N_n} \in \reels^{N_n\times N_n}$ are respectively the identity matrix and the matrix whose components are all equal to $1$.



\begin{lemma*} \label{lemmaConsistencySquare}
Let $\boldsymbol{\alpha} = (\alpha_{0},\alpha_{1},\alpha_{2})$ be a multi-index such that $|\boldsymbol{\alpha}| \leq m$. if $\alpha_0 >0$, then we have 
\bea 
 \sup_{\xi \in \Xi} \espu \l[ \l(\inv{\sqrt{N_n}}\frac{\partial^{\boldsymbol{\alpha}} \G_n(\xi)}{\partial \xi^{\boldsymbol{\alpha}}} - \frac{\partial^{\boldsymbol{\alpha}} \G_{\infty,1}(\xi)}{\partial \xi^{\boldsymbol{\alpha}}} \r)^2 \r] \to^\proba 0.
 \label{llnGSigma}
\eea
 
Moreover, if $\alpha_0 = 0$, then we have 
\bea 
 \sup_{\xi \in \Xi} \espu \l[ \l(\inv{N_n}\frac{\partial^{\boldsymbol{\alpha}} \G_n(\xi)}{\partial \xi^{\boldsymbol{\alpha}}} - \frac{\partial^{\boldsymbol{\alpha}} \G_{\infty,2}(\xi)}{\partial \xi^{\boldsymbol{\alpha}}} \r)^2 \r] \to^\proba 0.  
 \label{llnGa}
\eea


\end{lemma*}
\begin{proof} 
First note that by Lemma \ref{lemmaTransfo}, $\G_n(\xi)$ has the representation
\bea 
\G_n(\xi) = W(\theta)^T \ddot{\Omega}^{-1} W(\theta),
\eea 
and thus by (\ref{decompositionMatrixDdot}) $\G_n(\xi)$ admits the decomposition
\beas 
\G_n(\xi) &=& \inv{a^2}\textnormal{Tr} \l(W(\theta)W(\theta)^T\r) -\frac{u_0}{2a^2\sqrt{N_n}}\l(1+ \bo{\inv{\sqrt{ N_n}}}\r) W(\theta)^T \l\{\cale^- - \cale^+\r\}W(\theta)\\ &+&\bo{e^{-s\sqrt{N_n}}}W(\theta)^T\J_{N_n}W(\theta).   
\eeas
Consider now some multi-index $\boldsymbol{\alpha}$ such that $|\boldsymbol{\alpha}| \leq m$, and first assume that $\alpha_0 > 0$. Let us denote 
\bea 
J_{\boldsymbol{\alpha}}^{+} = \frac{\partial^{\boldsymbol{\alpha}} }{\partial \xi^{\boldsymbol{\alpha}}}  \l\{ \frac{u_0}{2a^2} \cale^+ \r\},
\eea 
and a similar definition for $J_{\boldsymbol{\alpha}}^-$.
We show (\ref{llnGSigma}). First note that in that case $\frac{\partial^{\boldsymbol{\alpha}}}{\partial \xi^{\boldsymbol{\alpha}}} \l( \inv{a^2}\textnormal{Tr} \l(W(\theta)W(\theta)^T\r) \r) = 0$ since $\alpha_0 >0$ and $\inv{a^2}\textnormal{Tr} \l(W(\theta)W(\theta)^T\r)$ does not depend on $\sigma^2$. Moreover, it is immediate to see that the term in $\bo{e^{-s\sqrt{N_n}}}\frac{\partial^{\boldsymbol{\alpha}}}{\partial \xi^{\boldsymbol{\alpha}}}\l\{W(\theta)^T\J_{N_n}W(\theta)\r\}$ is negligible given the factor $e^{-s\sqrt{N_n}}$. Let us show now that we have  

\bea 
\sup_{\xi \in \Xi} \espu \l[\l(\frac{\partial^{\boldsymbol{\alpha}}}{\partial \xi^{\boldsymbol{\alpha}}}\frac{u_0}{2a^2  N_n}\l(1+ \bo{\inv{\sqrt {N_n}}}\r) W(\theta)^T  \cale^+W(\theta) \r)^2\r] \to^\proba 0,
\label{eqNorme2DerivSigma0}
\eea
which after some straightforward calculation is equivalent to showing that 
\bea 
\sup_{\xi \in \Xi} \inv{N_n^2}\espu \l[ \l(\frac{\partial^{\alpha_2} }{\partial \theta^{\alpha_2}}\l\{ W(\theta)^T J_{\boldsymbol{\alpha}}^+ W(\theta)\r\} \r)^2\r] \to^\proba 0.
\label{eqNorme2DerivSigma}
\eea 
By the classical variance-bias decomposition, (\ref{eqNorme2DerivSigma}) will be proved if we can show uniformly in $\xi \in \Xi$ that we have
\bea 
\inv{N_n}\frac{\partial^{\alpha_2} }{\partial \theta^{\alpha_2}}\espu \l[W(\theta)^T J_{\boldsymbol{\alpha}}^+  W(\theta)  \r] \to^\proba 0
\label{eqNorme2DerivSigmaEsp}
\eea 
on the one hand, and
\bea 
\inv{N_n^2}\frac{\partial^{\alpha_2} }{\partial \theta^{\alpha_2}}\textnormal{Var}_\calu \l[ W(\theta)^T J_{\boldsymbol{\alpha}}^+  W(\theta)\r] \to^\proba 0
\label{eqNorme2DerivSigmaVar}
\eea 
on the other hand. We start by (\ref{eqNorme2DerivSigmaEsp}). Recall that $K_n = N_n^{1/2+\delta}$ for some $0<\delta <1/2$. From the definition of $\cale^+$ in (\ref{defCale}), it is straightforward to see that $J_{\boldsymbol{\alpha},i,j}^+ = \bo{e^{-\frac{u_0}{2}N_n^{\delta}}}$ as soon as $i+j \geq K_n$, and $J_{\boldsymbol{\alpha},i,j}^+ = \bo{1}$ otherwise. Therefore, we have 
\bea 
\inv{N_n}\frac{\partial^{\alpha_2} }{\partial \theta^{\alpha_2}}\espu \l[W(\theta)^T J_{\boldsymbol{\alpha}}^+  W(\theta)  \r]  &=& \inv{N_n} \sum_{i,j=0}^{N_n}{J_{\boldsymbol{\alpha},i,j}^+ \frac{\partial^{\alpha_2} \rho_{|i-j|}(\theta) }{\partial \theta^{\alpha_2}} }. 
\eea
From the symmetry $J_{\boldsymbol{\alpha},N_n-i,N_n-j}^+ = J_{\boldsymbol{\alpha},i,j}^+$ we split the sum as 
\beas 
\inv{N_n} \sum_{i,j=0}^{N_n}{J_{\boldsymbol{\alpha},i,j}^+ \frac{\partial^{\alpha_2} \rho_{|i-j|}(\theta) }{\partial \theta^{\alpha_2}} } &=& \frac{2}{N_n} \sum_{0 \leq i+j < N_n}{J_{\boldsymbol{\alpha},i,j}^+ \frac{\partial^{\alpha_2} \rho_{|i-j|}(\theta) }{\partial \theta^{\alpha_2}} } + \underbrace{\inv{N_n} \sum_{i+j=N_n}{J_{\boldsymbol{\alpha},i,j}^+ \frac{\partial^{\alpha_2} \rho_{|i-j|}(\theta) }{\partial \theta^{\alpha_2}} }}_{\bo{e^{-\frac{u_0}{2}N_n^{\delta}}}},\\
&=& \frac{2}{N_n} \sum_{0 \leq i+j < K_n}{J_{\boldsymbol{\alpha},i,j}^+ \frac{\partial^{\alpha_2} \rho_{|i-j|}(\theta) }{\partial \theta^{\alpha_2}} } + \frac{2}{N_n} \sum_{K_n \leq i+j < N_n}{J_{\boldsymbol{\alpha},i,j}^+ \frac{\partial^{\alpha_2} \rho_{|i-j|}(\theta) }{\partial \theta^{\alpha_2}} } 
\\ & & + \bo{e^{-\frac{u_0}{2}N_n^{\delta}}},
\eeas 
and first we have 
\beas  
\frac{2}{N_n} \l|\sum_{K_n \leq i+j < N_n}{J_{\boldsymbol{\alpha},i,j}^+ \frac{\partial^{\alpha_2} \rho_{|i-j|}(\theta) }{\partial \theta^{\alpha_2}}}\r| &=& \bo{e^{-\frac{u_0}{2}N_n^{\delta}}} \times \inv{N_n}\l|\sum_{K_n \leq i+j < N_n}{ \frac{\partial^{\alpha_2} \rho_{|i-j|}(\theta) }{\partial \theta^{\alpha_2}}} \r|\\
&\leq& \bo{e^{-\frac{u_0}{2}N_n^{\delta}}} \times \underbrace{\sum_{0 \leq k < N_n}{\l(1 -\frac{k}{N_n}\r) \l|\frac{\partial^{\alpha_2} \rho_{k}(\theta) }{\partial \theta^{\alpha_2}}\r|}}_{\bo{1}}, \\
&=& \bo{e^{-\frac{u_0}{2}N_n^{\delta}}},
\eeas  
where the estimates $\sum_{0 \leq k < N_n}{\l(1 -\frac{k}{N_n}\r) \l|\frac{\partial^{\alpha_2} \rho_{k}(\theta) }{\partial \theta^{\alpha_2}}\r|} = \bo{1}$ uniformly in $\theta \in \Theta$ is a consequence of Assumption (\ref{assCorr}). Now we also have 
\beas  
\frac{2}{N_n} \l|\sum_{0 \leq i+j < K_n}{J_{\boldsymbol{\alpha},i,j}^+ \frac{\partial^{\alpha_2} \rho_{|i-j|}(\theta) }{\partial \theta^{\alpha_2}}}\r| &\leq& \bo{1} \times \frac{K_n}{N_n} \sum_{0 \leq k < K_n}{\l(1 -\frac{k}{K_n}\r) \l|\frac{\partial^{\alpha_2} \rho_{k}(\theta) }{\partial \theta^{\alpha_2}}\r|}\\
&=& \bo{N_n^{\delta-1/2}},
\eeas 
by the same argument. Thus we have proved (\ref{eqNorme2DerivSigmaEsp}). Now, using a similar formula as for \cite{mccullagh1987tensor}, Section 3.3 p. 61, (\ref{eqNorme2DerivSigmaVar}) can be expressed as the sum

\bea 
\inv{N_n^2} \textnormal{Var}_\calu \l[ \frac{\partial^{\alpha_2} }{\partial \theta^{\alpha_2}}\l\{W(\theta)^T J_{\boldsymbol{\alpha}}^+  W(\theta)\r\}\r] = V_1^+ + V_2^+,
\label{decompositionCullagh}
\eea 
where by the Leibniz rule,
\bea 
V_1^+ = \inv{N_n^2}\sum_{r_1,r_2 = 0}^{\alpha_2}\binom{\alpha_2}{r_1}\binom{\alpha_2}{r_2}\sum_{i,j,k,l = 0}^{N_n}J_{\boldsymbol{\alpha},i,j}^+J_{\boldsymbol{\alpha},k,l}^+  \kappa_{j-i,k-i,l-i}^{\boldsymbol{\beta}(\boldsymbol{r})}(\theta),
\label{decompositionCumulant}
\eea 
with $\boldsymbol{\beta}(\boldsymbol{r}) = (r_1,\alpha_2-r_1,r_2,\alpha_2-r_2)$, $\boldsymbol{r} = (r_1,r_2)$, where $r_1$ and $r_2$ are $d$ dimensional multi-indices such that $r_1,r_2 \leq \alpha_2$, and
\bea 
V_2^+ = \inv{N_n^2}\sum_{i,j,k,l =0}^{N_n}{J_{\boldsymbol{\alpha},i,j}^+J_{\boldsymbol{\alpha},k,l}^+ \l\{\rho_{|k-i|}^{(r_1,r_2)}(\theta) \rho_{|l-j|}^{(\alpha_2-r_1,\alpha_2-r_2)}(\theta) + \rho_{|l-i|}^{(r_1,\alpha_2-r_2)}(\theta) \rho_{|k-j|}^{(\alpha_1-r_1,r_2)}(\theta) \r\}}.
\label{decompositionRhoRho}
\eea
First, we have  
\beas  
\frac{4}{N_n^2}\l|\sum_{0 \leq i,j,k,l \leq N_n}{J_{\boldsymbol{\alpha},i,j}^+J_{\boldsymbol{\alpha},k,l}^+  \kappa_{j-i,k-i,l-i}^{\boldsymbol{\beta}(\boldsymbol{r})}(\theta) }\r|&\leq& \bo{1} \times \frac{1}{N_n^2} \sum_{0 \leq i,j,k,l \leq N_n}{\l| \kappa_{j-i,k-i,l-i}^{\boldsymbol{\beta}(\boldsymbol{r})}(\theta) \r|}.
\eeas 
Now, since we can swap the elements of $\boldsymbol{\beta}(\boldsymbol{r})$ without loss of generality we can assume that $\kappa_{j-i,k-i,l-i}^{\boldsymbol{\beta}(\boldsymbol{r})}(\theta)$ is symmetric in $i,j,k,l$ so that we have up to a multiplicative constant 
\beas 
\l|V_1^+\r| &\leq& O(1) \times \frac{1}{N_n^2} \sum_{r_1,r_2 = 0}^{\alpha_2}\sum_{0 \leq i\leq j, k, l < K_n}{\l| \kappa_{j-i,k-i,l-i}^{\boldsymbol{\beta}(\boldsymbol{r})}(\theta)\r|}\\
&\leq&\bo{1} \times \frac{1}{N_n} \sum_{0 \leq p,q,r < N_n }{\l(1- \frac{p \wedge q \wedge r}{N_n}\r)\l|\kappa_{p,q,r}^{\boldsymbol{\beta}(\boldsymbol{r})}(\theta)\r|}\\
&=&  \bo{N_n^{-1}},
\eeas
by Assumption (\ref{assCumulant}). On the other hand, following a similar path as for the bias case, we can reduce up to negligible terms the elementary terms of $\l|V_2^+\r|$ to 
\beas 
 &&\frac{4}{N_n^2}\l|\sum_{0 \leq i+j < K_n, 0 \leq k+l < K_n}{J_{\boldsymbol{\alpha},i,j}^+J_{\boldsymbol{\alpha},k,l}^+  \l\{\rho_{|k-i|}^{(r_1,r_2)}(\theta) \rho_{|l-j|}^{(\alpha_2-r_1,\alpha_2-r_2)}(\theta)\r\}}\r|\\
&\leq& \bo{1} \times \sum_{0 \leq i+j < K_n, 0 \leq k+l < K_n}{\l| \l\{\rho_{|k-i|}^{(r_1,r_2)}(\theta) \rho_{|l-j|}^{(\alpha_2-r_1,\alpha_2-r_2)}(\theta)\r\}\r|}\\
&\leq& \bo{1} \times \frac{K_n^2}{n^2} \sum_{p,q=0}^{K_n-1}{\l(1-\frac{p}{K_n}\r)\l(1-\frac{q}{K_n}\r) \l| \l\{\rho_{p}^{(r_1,r_2)}(\theta) \rho_{q}^{(\alpha_2-r_1,\alpha_2-r_2)}(\theta)\r\}\r|}\\
&=& \bo{N_n^{2\delta-1}},
\eeas 
where the last estimate is obtained by application of Assumption (\ref{assCorr}). Similar reasoning also shows the negligibility of the terms in $\rho_{|l-i|}^{(r_1,\alpha_2-r_2)}(\theta) \rho_{|k-j|}^{(\alpha_1-r_1,r_2)}(\theta)$. Thus we have proved (\ref{eqNorme2DerivSigmaVar}), so that (\ref{eqNorme2DerivSigma0}) holds true. To complete the proof of (\ref{llnGSigma}), it remains to show 
\bea 
\sup_{\xi \in \Xi} \esp_\calu \l[\l(\frac{\partial^{\boldsymbol{\alpha}}}{\partial \xi^{\boldsymbol{\alpha}}}\l\{\frac{u_0}{2a^2  N_n}\l(1+ \bo{\inv{\sqrt{N_n}}}\r) W(\theta)^T  \cale^-W(\theta) - \G_{\infty,1}(\xi) \r\} \r)^2\r] \to^\proba 0,
\label{eqNorme2DerivSigma1}
\eea 
which can be expressed as 
\bea 
\sup_{\xi \in \Xi} \inv{N_n^2}\frac{\partial^{\alpha_2} }{\partial \theta^{\alpha_2}}\esp_\calu \l[ \l( W(\theta)^T J_{\boldsymbol{\alpha}}^- W(\theta) - \frac{\partial^{\boldsymbol{\alpha}} \G_{\infty,1}(\xi)}{\partial \xi^{\boldsymbol{\alpha}}} \r)^2\r] \to^\proba 0.
\label{eqNorme2DerivSigma2}
\eea  
We adopt the same bias-variance approach as before, and first compute using the symmetry $J_{\boldsymbol{\alpha},i,j}^- = J_{\boldsymbol{\alpha},j,i}^- $,
\bea 
\inv{N_n}\frac{\partial^{\alpha_2} }{\partial \theta^{\alpha_2}}\esp_\calu \l[  W(\theta)^T J_{\boldsymbol{\alpha}}^- W(\theta) \r] = \frac{2}{N_n} \sum_{0 \leq i<j \leq N_n}{J_{\boldsymbol{\alpha},i,j}^- \frac{\partial^{\alpha_2} \rho_{|i-j|}(\theta) }{\partial \theta^{\alpha_2}}} + \frac{1}{N_n}\sum_{i = 0}^{N_n}{J_{\boldsymbol{\alpha},i,i}^- \frac{\partial^{\alpha_2} \rho_{0}(\theta) }{\partial \theta^{\alpha_2}}}.
\label{eqBiasJmoins}
\eea 
Now, we have immediately $J_{\boldsymbol{\alpha},i,i}^- = \frac{\partial^{\boldsymbol{\alpha}}}{\partial \xi^{\boldsymbol{\alpha}}}\l(\frac{u_0}{2a^2}\r)$, so that 
\bea 
\frac{1}{N_n}\sum_{i = 0}^{N_n}{J_{\boldsymbol{\alpha},i,i}^- \frac{\partial^{\alpha_2} \rho_{0}(\theta) }{\partial \theta^{\alpha_2}}} - \frac{\partial^{\boldsymbol{\alpha}}}{\partial \xi^{\boldsymbol{\alpha}}}\l(\frac{u_0}{2a^2} \rho_{0}(\theta)\r) \to^\proba 0  
\label{eqConvDiagonal}
\eea 
uniformly in $\xi \in \Xi$. Thus, by (\ref{eqBiasJmoins}) and (\ref{eqConvDiagonal}) we have
\beas
\esp_\calu \l[ \inv{N_n}\frac{\partial^{\alpha_2} }{\partial \theta^{\alpha_2}} \l\{W(\theta)^T J_{\boldsymbol{\alpha}}^- W(\theta) \r\}- \frac{\partial^{\boldsymbol{\alpha}} \G_{\infty,1}(\xi)}{\partial \xi^{\boldsymbol{\alpha}}} \r] = \frac{2}{N_n} \sum_{0 \leq i<j \leq N_n}{\l(J_{\boldsymbol{\alpha},i,j}^-  -\frac{\partial^{\boldsymbol{\alpha}}}{\partial \xi^{\boldsymbol{\alpha}}}\l(\frac{u_0}{2a^2}\r) \r)\frac{\partial^{\alpha_2} \rho_{|i-j|}(\theta) }{\partial \theta^{\alpha_2}}} \\+ \lop{1},
\eeas 
and noticing that $\sup_{\xi \in \Xi} \l\{J_{\boldsymbol{\alpha},i,j}^-  -\frac{\partial^{\boldsymbol{\alpha}}}{\partial \xi^{\boldsymbol{\alpha}}}\l(\frac{u_0}{2a^2} \r)\r\}= o_\proba(1)$, we easily conclude as before by the dominated convergence theorem along with Assumption (\ref{assCorr}) that 
\bea 
\sup_{\xi \in \Xi}\inv{N_n}\frac{\partial^{\alpha_2} }{\partial \theta^{\alpha_2}}\esp_\calu \l[  W(\theta)^T J_{\boldsymbol{\alpha}}^- W(\theta) - \frac{\partial^{\boldsymbol{\alpha}} \G_{\infty,1}(\xi)}{\partial \xi^{\boldsymbol{\alpha}}} \r] \to^\proba 0.
\eea 
Moreover, the variance term is treated as previously. 

 Finally, for (\ref{llnGa}), similar computations yield the result. Note that when $\alpha_0 = 0$, the diagonal terms of $\frac{\partial^\alpha \ddot{\Omega}^{-1}}{\partial \xi^\alpha }$ predominate, hence there is absence of higher order correlations $\rho_k(\theta)$, $k \geq 1$, in the limit.
\end{proof} 

To deal with the cross terms, we define in the same fashion as before 

\bea 
\K_n(\xi) = (\mu(\theta_0) - \mu(\theta))^T \Omega^{-1} \widetilde{Y}(\theta_0),
\label{defKnH1}
\eea 

\begin{lemma*} \label{lemmaConsistencyCross}
Let $\boldsymbol{\alpha} = (\alpha_0,\alpha_1,\alpha_2)$ be a multi-index such that $|\boldsymbol{\alpha}| \leq m$. Then, if $\alpha_0 = 0$,  we have 

\bea 
\sup_{\xi \in \Xi} \esp_\calu \l[ \l(\inv{\sqrt{N_n}} \frac{\partial^{\boldsymbol{\alpha}} \K_n(\xi)}{\partial \xi^{\boldsymbol{\alpha}}}\r)^2\r] \to^\proba 0.
\label{eqCross1}
\eea 
If $\alpha_0 > 0$, we have 
\bea 
\sup_{\xi \in \Xi} \esp_\calu \l[ \l(\inv{ N_n} \frac{\partial^{\boldsymbol{\alpha}} \K_n(\xi)}{\partial \xi^{\boldsymbol{\alpha}}}\r)^2\r] \to^\proba 0.
\label{eqCross2}
\eea

\end{lemma*}

\begin{proof}
Let us first show (\ref{eqCross1}). By Lemma \ref{lemmaTransfo}, we have the representation 
\bea 
\K_n(\xi) =  W(\theta)^T\dot{\Omega}^{-1}\Delta X + W(\theta)^T \ddot{\Omega}^{-1}\epsilon. 
\eea
Accordingly, we start by showing that 
\bea 
\sup_{\xi \in \Xi} \esp_\calu \l[ \l(\inv{\sqrt{N_n}} \frac{\partial^{\boldsymbol{\alpha}} }{\partial \xi^{\boldsymbol{\alpha}}} \l\{W(\theta)^T\dot{\Omega}^{-1}\Delta X\r\}\r)^2\r] \to^\proba 0.
\label{crossTermDX}
\eea
In the case where $X$ is continuous, that is $J=0$, we have
\beas  
\esp_\calu \l[ \l(\inv{\sqrt{N_n}} \frac{\partial^{\boldsymbol{\alpha}} }{\partial \xi^{\boldsymbol{\alpha}}} \l\{W(\theta)^T\dot{\Omega}^{-1}\Delta X\r\}\r)^2\r] &=& \inv {N_n} \esp_\calu \l[ \esp \l[ \l. \l( \frac{\partial^{\boldsymbol{\alpha}} }{\partial \xi^{\boldsymbol{\alpha}}} \l\{W(\theta)^T\dot{\Omega}^{-1}\Delta X\r\}\r)^2 \r| \calu \vee \sigma(X) \r]\r]\\
&=& \inv {N_n} \esp_\calu \l[ \sum_{0 \leq i,k \leq N_n, 1 \leq j,l \leq N_n }{\frac{\partial^{\alpha_2} \rho_{|i-k|}(\theta) }{\partial \theta^{\alpha_2}}} \frac{\partial^{\boldsymbol{\alpha}} \dot{\omega}^{i,j}}{\partial \xi^{\boldsymbol{\alpha}}}\frac{\partial^{\boldsymbol{\alpha}} \dot{\omega}^{k,l}}{\partial \xi^{\boldsymbol{\alpha}}} \Delta X_j^n\Delta X_l^n \r]\\
&=& \inv {N_n}  \sum_{0 \leq i,k \leq N_n, 1 \leq j \leq N_n }{\frac{\partial^{\alpha_2} \rho_{|i-k|}(\theta) }{\partial \theta^{\alpha_2}}} \frac{\partial^{\boldsymbol{\alpha}} \dot{\omega}^{i,j}}{\partial \xi^{\boldsymbol{\alpha}}}\frac{\partial^{\boldsymbol{\alpha}} \dot{\omega}^{k,j}}{\partial \xi^{\boldsymbol{\alpha}}} \esp_\calu \int_{\ti{j-1}}^{\ti j}{\sigma_s^2ds}. 
\eeas 
Now, using $\frac{\partial^{\boldsymbol{\alpha}} \dot{\omega}^{k,j}}{\partial \xi^{\boldsymbol{\alpha}}} = \bo{1}$, and that $\esp_\calu \int_{\ti{j-1}}^{\ti j}{\sigma_s^2ds} = \bop{n^{-1+\gamma}}$ by Assumption \textbf{(H)} where $\gamma>0$ can be taken arbitrary small, we get 
\beas 
\esp_\calu \l[ \l(\inv{\sqrt{N_n}} \frac{\partial^{\boldsymbol{\alpha}} }{\partial \xi^{\boldsymbol{\alpha}}} \l\{W(\theta)^T\dot{\Omega}^{-1}\Delta X\r\}\r)^2\r] &=& \bop{N_n^{-1} n^{-1+\gamma}} \times \sum_{0 \leq i,k \leq N_n, 1\leq j \leq N_n}{\l|\frac{\partial^{\boldsymbol{\alpha}} \dot{\omega}^{i,j}}{\partial \xi^{\boldsymbol{\alpha}}} \r|\l|\frac{\partial^{\alpha_2} \rho_{|i-k|}(\theta) }{\partial \theta^{\alpha_2}} \r|}\\
\eeas
and summing first over $k$ and using Assumption (\ref{assCorr}), we get 
\beas 
\esp_\calu \l[ \l(\inv{\sqrt{N_n}} \frac{\partial^{\boldsymbol{\alpha}} }{\partial \xi^{\boldsymbol{\alpha}}} \l\{W(\theta)^T\dot{\Omega}^{-1}\Delta X\r\}\r)^2\r] &=& \bop{N_n^{-1} n^{-1+\gamma}} \times \sum_{0 \leq i \leq N_n, 1 \leq j \leq N_n}{\l\{e^{-u_0\frac{|i-j|}{\sqrt {N_n}}}+ e^{-u_0\frac{i+j}{\sqrt {N_n}}}\r\}}\\
&=& \bop{N_n^{1/2} n^{-1+\gamma}} \to^\proba 0,
\eeas 
uniformly in $\xi \in \Xi$ by direct calculation for the last estimate. Now, when $J \neq 0$ for $n$ sufficiently large, using the finite activity property of the jump process it is easy to see that an additional term appears in the quadratic expression,
\bea  
\inv {N_n}  \sum_{0 \leq i,k \leq N_n} {\frac{\partial^{\alpha_2} \rho_{|i-k|}(\theta) }{\partial \theta^{\alpha_2}}}\sum_{j = 1}^{ N_J } \frac{\partial^{\boldsymbol{\alpha}} \dot{\omega}^{i,i_j}}{\partial \xi^{\boldsymbol{\alpha}}}\frac{\partial^{\boldsymbol{\alpha}} \dot{\omega}^{k,i_j}}{\partial \xi^{\boldsymbol{\alpha}}} \esp_\calu\Delta J_{\tau_j}^2,
\label{additionalCrossTermJump}
\eea 
where $N_J$ is the finite number of jumps of $J$ on $[0,T]$, $(\tau_j)_{1 \leq j\leq N_J}$ the related jump times, and $i_j$ is the only index such that $\ti{i_j-1} < \tau_j \leq \ti{i_j}$. Given the estimate of $\dot{\omega}^{i,j}$ of the previous section and the definition of $i_j$, we immediately see that the coefficients $\frac{\partial^{\boldsymbol{\alpha}}\dot{\omega}^{i,i_j}}{\partial \xi^{\boldsymbol{\alpha}}} = \bo{e^{-v\sqrt{N_n}}}$ for some $v>0$, and thus (\ref{additionalCrossTermJump}) is negligible so that we have (\ref{crossTermDX}). Now we show that we have
\bea 
\sup_{\xi \in \Xi} \esp_\calu \l[ \l(\inv{\sqrt{N_n}} \frac{\partial^{\boldsymbol{\alpha}} }{\partial \xi^{\boldsymbol{\alpha}}} \l\{W(\theta)^T\ddot{\Omega}^{-1}\epsilon\r\}\r)^2\r] \to^\proba 0.
\label{crossTermEpsilon}
\eea
By independence, we immediately have 
\bea 
\esp_\calu \l[ \l(\inv{\sqrt{N_n}} \frac{\partial^{\boldsymbol{\alpha}} }{\partial \xi^{\boldsymbol{\alpha}}} \l\{W(\theta)^T\ddot{\Omega}^{-1}\epsilon\r\}\r)^2\r] &=& \frac{a_0^2}{N_n} \sum_{0 \leq i,j,k \leq N_n} {\frac{\partial^{\alpha_2} \rho_{|i-k|}(\theta) }{\partial \theta^{\alpha_2}}\frac{\partial^{\boldsymbol{\alpha}} \ddot{\omega}^{i,j}}{\partial \xi^{\boldsymbol{\alpha}}}\frac{\partial^{\boldsymbol{\alpha}} \ddot{\omega}^{k,j}}{\partial \xi^{\boldsymbol{\alpha}}}},
\eea 
and from here using that $\frac{\partial^{\boldsymbol{\alpha}} \ddot{\omega}^{k,j}}{\partial \xi^{\boldsymbol{\alpha}}} = \bo{\inv{\sqrt{N_n}}}$ when $\alpha_0 >0$ on the one hand, then summing over $k$ and applying Assumption (\ref{assCorr}) and finally computing the explicit sum of exponential terms leads to 
\bea 
\esp_\calu \l[ \l(\inv{\sqrt{N_n}} \frac{\partial^{\boldsymbol{\alpha}} }{\partial \xi^{\boldsymbol{\alpha}}} \l\{W(\theta)^T\ddot{\Omega}^{-1}\epsilon\r\}\r)^2\r] &=& \bop{N_n^{-1/2}}
\eea 
uniformly in $\xi \in \Xi$. Thus we have proved (\ref{eqCross1}). Finally, (\ref{eqCross2}) is proved similarly. 
\end{proof}




\subsection{Proof of Theorem \ref{theoremCLT}}
We derive the limit theory for $\widehat{\xi}_{n,err}$ (hereafter denoted by $\widehat{\xi}_{n}$) when $a_0^2 >0$. In our terminology, we recall that we have for any $\xi \in \Xi$ and up to an additive constant term

\bea 
l_n(\xi) = \underbrace{- \half \textnormal{log det}(\Omega) -\half \widetilde{Y}(\theta_0)^T \Omega^{-1} \widetilde{Y}(\theta_0)}_{l_n^{(\sigma^2,a^2)}(\xi)} \underbrace{- \half \G_n(\xi) - \K_n(\xi)}_{l_n^{(\theta)}(\xi)},
\eea 
with $\Omega^{-1} = [\omega^{i,j}]_{1 \leq i \leq N_n,1 \leq i \leq N_n}$, where we also recall that the definition of $\omega^{i,j}$ can be found in (\ref{defCoeffOmega}). Moreover, note that $l_n^{(\sigma^2,a^2)}$ does not depend on $\theta$ and corresponds precisely to the quasi log-likelihood with no information as studied in \cite{xiu2010quasi} and extended to a more general setting in \cite{clinet2018efficient}. On the other hand, $l_n^{(\theta)}$ is the additional part incorporating $\theta$ and depends on the whole vector $\xi = (\sigma^2,a^2,\theta)$. Following a similar procedure as in \cite{xiu2010quasi} and \cite{clinet2018efficient}, we introduce the approximate log-likelihood random field as 

\bea 
\overline{l}_n(\xi) = \underbrace{-\half \textnormal{log det}(\Omega) -\half  \textnormal{Tr}\l(\Omega^{-1} \l\{\Sigma_0^c + \Sigma_0^d \r\} \r)  }_{\overline{l}_n^{(\sigma^2,a^2)}(\xi)} \underbrace{- \half \G_n(\xi)}_{\overline{l}_n^{(\theta)}(\xi)},
\eea 
with 
\beas 
\Sigma_0^c &=& \left(\begin{matrix}
                    \int_0^{t_{1}^n}{\sigma_s^2ds} + 2 a^2 & - a^2 & 0 & \cdots & 0 \\
                    - a^2 & \int_{t_{1}^n}^{t_{2}^n}{\sigma_s^2ds} +  2 a^2 & - a^2 & \ddots & \vdots \\
                    0 & - a^2 & \int_{t_{2}^n}^{t_{3}^n}{\sigma_s^2ds}+ 2 a^2 & \ddots & 0\\
                    \vdots & \ddots & \ddots & \ddots & - a^2\\
                    0 & \cdots & 0 & - a^2 & \int_{t_{N_n-1}^n}^{t_{N_n}^n}{\sigma_s^2ds} + 2 a^2
                  \end{matrix}\right), \\[12pt]
\eeas 
and
\beas 
\Sigma_0^d = \textnormal{diag}\l(\sum_{0 < s\leq t_1^n}{\Delta J_s^2} ,\sum_{t_1^n < s\leq t_2^n}{\Delta J_s^2},\cdots, \sum_{t_{N_n-1}^n < s\leq t_{N_n}^n}{\Delta J_s^2}\r). 
\eeas 
Consider the diagonal scaling matrix 
\bea 
\Phi_n = \textnormal{diag}(N_n^{1/2},N_n,N_n).
\eea 
Define also for $\xi \in \Xi$ the scaled scores  $ \Psi_n^{(\sigma^2,a^2)}(\xi) = -\Phi_n^{-1} \frac{\partial l_n^{(\sigma^2,a^2)}(\xi)}{\partial \xi}$, $ \Psi_n^{(\theta)}(\xi) = -\Phi_n^{-1} \frac{\partial l_n^{(\theta)}(\xi)}{\partial \xi}$, and $\Psi_n =  \Psi_n^{(\sigma^2,a^2)}+ \Psi_n^{(\theta)}$. Accordingly, the approximate scores $\overline{\Psi}_n$, $\overline{\Psi}_n^{(\sigma^2,a^2)}$, and $\overline{\Psi}_n^{(\theta)}$ admit the same definition replacing $l_n$ by $\overline{l}_n$. We start by a technical lemma to ensure the uniform convergence of some random fields.

\begin{lemma*} \label{lemmaSobolev}
Let $X_n(\xi)$ be a sequence of random variables of class $C^m$ in $\xi \in \Xi_n \subset \reels^d$, each $\Xi_n$ convex compact, such that $2m > d$, and $\calu$ a sub $\sigma$-field of the general $\sigma$-field. For any multi-index $\boldsymbol{\alpha}$ such that $0 \leq |\boldsymbol{\alpha}| \leq m$, we assume that
\beas  
\sup_{\xi \in \Xi_n} \esp_\calu \l[  \partial_\xi^{\boldsymbol{\alpha}} X_n(\xi)^2\r] = o_\proba(\textnormal{Leb}(\Xi_n)).
\eeas 
Then we have the uniform convergence
\bea 
\esp_\calu \l[ \sup_{\xi \in \Xi_n} |X_n(\xi)| \r] \to^\proba 0.
\eea 
\end{lemma*}

\begin{proof} 
By Theorem 4.12 Part I case A (taking $j=0$, $p=2$) from \cite{adams2003sobolev}, we apply Sobolev's inequality and define some constant $M$ such that we have 
\beas 
\esp_\calu \l[ \sup_{\xi \in \Xi_n} |X_n(\xi)| \r] &\leq& M \sum_{\boldsymbol{\alpha} | |\boldsymbol{\alpha}| \leq m} \l(\int_{\Xi_n} \esp \l[   \partial_\xi^{\boldsymbol{\alpha}} X_n(\xi) ^2 \r] d\xi \r)^{1/2}\\
&\leq& M \textnormal{Leb}(\Xi_n)^{1/2} \sum_{\boldsymbol{\alpha} | |\boldsymbol{\alpha}| \leq m} \l( \sup_{\xi \in \Xi_n} \esp \l[   \partial_\xi^{\boldsymbol{\alpha}} X_n(\xi) ^2 \r]  \r)^{1/2} \to^\proba 0.\\
\eeas 
\end{proof} 

\begin{lemma*} \label{lemmaUnifConvScore} (Asymptotic score) For any $\xi \in \Xi$, let

\beas  
\Psi_{\infty}(\xi) = \l( \begin{matrix}  -\frac{\sqrt T}{8a\sigma^3}\l(\overline{\sigma}_0^2 -\sigma^2\r)-\frac{\sqrt T}{8a^3\sigma}\l(a^2-a_0^2-\rho_0(\theta) - 2\sum_{k=1}^{+\infty}{\rho_k(\theta)}\r) \\ \inv{2a^4}\l(a^2-a_0^2-\rho_0(\theta)\r)\\ \inv{2a^2} \frac{\partial \rho_0(\theta)}{\partial \theta} \end{matrix} \r).
\eeas
We have 

\bea \sup_{\xi \in \Xi} \l|\Psi_n(\xi)-\Psi_\infty(\xi)\r| \to^\proba 0.
\label{eqConsistency0}
\eea 

\end{lemma*}

\begin{proof} 
Since $\Psi_n = \Psi_n^{(\sigma^2,a^2)} + \Psi_n^{(\theta)}$, the lemma will be proved if we can show that 
\bea 
\Psi_n^{(\sigma^2,a^2)}(\xi) \to^\proba  \l( \begin{matrix} -\frac{\sqrt T}{8a\sigma^3}\l(\overline{\sigma}_0^2 -\sigma^2\r)-\frac{\sqrt T}{8a^3\sigma}\l(a^2-a_0^2\r)\\ \inv{2a^4}\l(a^2-a_0^2\r)\\ 0 \end{matrix} \r)
\label{convScoreSigma}
\eea 
and 
\bea 
\Psi_n^{(\theta)}(\xi) \to^\proba  \l( \begin{matrix} \frac{\sqrt T}{8a^3\sigma}\l(\rho_0(\theta) + 2\sum_{k=1}^{+\infty}{\rho_k(\theta)}\r) \\ -\frac{\rho_0(\theta)}{2a^4}\\ \inv{2a^2}\frac{\partial \rho_0(\theta)}{\partial \theta} \end{matrix} \r) 
\label{convScoreTheta}
\eea 
uniformly in $\xi \in \Xi$. Note that (\ref{convScoreSigma}) is a consequence of Lemma A.3 in \cite{clinet2018efficient}. For (\ref{convScoreTheta}), by Lemma \ref{lemmaConsistencyCross} combined with Lemma \ref{lemmaSobolev}, we have uniformly in $\xi \in \Xi$ that $\Psi_n^{(\theta)}(\xi) - \overline{\Psi}_n^{(\theta)}(\xi)  = o_\proba(1)$. Thus it is sufficient to show that we have the convergence (\ref{convScoreTheta}) for $\overline{\Psi}_n^{(\theta)}$. Combining Lemma \ref{lemmaConsistencySquare} and Lemma \ref{lemmaSobolev}, we obtain 
\bea 
\sup_{\xi \in \Xi} \l\{ \inv{2\sqrt{N_n}}\frac{\partial \G_n(\xi)}{\partial \sigma^2} + \frac{u_0}{8a^2\sigma^2}\l\{\rho_0(\theta) + 2\sum_{k=1}^{+\infty}{\rho_k(\theta)}\r\}  \r\} \to^\proba 0,
\label{llnSigma}
\eea 

\bea 
\sup_{\xi \in \Xi} \l\{ \inv{2N_n}\frac{\partial \G_n(\xi)}{\partial a^2}  +  \frac{\rho_0(\theta)}{2a^4} \r\} \to^\proba 0,
\label{llna}
\eea 
and
\bea 
\sup_{\xi \in \Xi} \l\{ \inv{2N_n} \frac{\partial \G_n(\xi)}{\partial \theta} - \inv{2a^2}\frac{\partial \rho_0(\theta)}{\partial \theta} \r\} \to^\proba 0,
\label{llnPartial}
\eea 
and recalling that $\overline{\Psi}_n^{(\theta)} =  \half \Phi_n^{-1}\frac{\partial \G_n(\xi)}{\partial \xi}$, we get (\ref{convScoreTheta}).

\end{proof} 

\begin{theorem*} (consistency).
If $\widehat{\xi}_n = (\widehat{\sigma}_n^2,\widehat{a}_n^2,\widehat{\theta}_n)$ is the QMLE, we have 

\bea 
\widehat{\xi}_n \to^\proba \xi_0 := \l(\overline{\sigma}_0^2, a_0^2, \theta_0\r). 
\label{eqConsistency}
\eea 
\label{thmConsistencyH1}
\end{theorem*}

\begin{proof} 
 We extend the proof of Lemma A.4 in \cite{clinet2018efficient}, that is, since we already have
\bea 
\sup_{\xi \in \Xi} \l|\Psi_n(\xi)-\Psi_\infty(\xi)\r| \to^\proba 0
\eea 
by Lemma \ref{lemmaUnifConvScore}, we show that for any $\epsilon >0$
\bea 
\inf_{ \xi \in \Xi : \| \xi-\xi_0 \| \geq \epsilon} \| \Psi_\infty(\xi)\|^2>0 = \| \Psi_\infty(\xi_0 )\|^2  \textnormal{  }\proba\textnormal{-a.s}.
\label{conditionIdentifiability}
\eea 
Given the form of $\Psi_\infty$, the equality $\Psi_\infty(\xi_0 ) = 0$ is immediate. Note also that the left hand side inequality of (\ref{conditionIdentifiability}) will be automatically satisfied if we show that $\| \Psi_\infty(\xi ) \|^2 > 0$ as soon as $\xi \neq \xi_0$ by a continuity argument since $\Xi$ is compact. Let us then take $\xi \in \Xi - \{\xi_0\}$ such that $\Psi_\infty(\xi ) = 0$, and assume first that $\theta \neq \theta_0$. In that case, we have 

$$ \| \Psi_\infty(\xi ) \|^2 \geq \inv{4a^4} \l\|\frac{\partial \rho_0(\theta)}{\partial \theta} \r\|^2 > 0,$$
by (\ref{assPartial}), which leads to a contradiction. We thus get $\theta = \theta_0$ and in a similar way, we also have 
$$ 0 = \| \Psi_\infty(\sigma^2,a^2,\theta_0 ) \|^2 \geq \inv{4a^8}\l(a^2-a_0^2\r)^2,$$
that implies $a^2 = a_0^2$. Finally, the first component of $\Psi_\infty$ leads to the domination 
$$ 0 = \| \Psi_\infty(\sigma^2,a_0^2,\theta_0 ) \|^2 \geq \frac{T}{64a_0^2\sigma^6 }\l(\overline{\sigma}_0^2-\sigma^2\r)^2,$$
so that we can conclude $\sigma^2 = \overline{\sigma}_0^2$.

\end{proof} 

Let $H_n$ be the scaled Hessian matrix of the likelihood field, defined as 

\bea 
H_n(\xi) = - \Phi_n^{-1/2} \frac{\partial^2 l_n(\xi)}{\partial \xi^2} \Phi_n^{-1/2},
\label{defHn}
\eea 
and similarly $H_n^{(\sigma^2,a^2)}$, $H_n^{(\theta)}$, $ \overline{H}_n$, $\overline{H}_n^{(\sigma^2,a^2)}$, $\overline{H}_n^{(\theta)}$.

\begin{lemma*} (Fisher information) \label{lemmaFisherConsistency}
For $\xi_0 =(\sigma_0^2,a_0^2,\theta_0)$, let $\Gamma(\xi_0)$ be the matrix 

\bea 
\Gamma(\xi_0) = \l( \begin{matrix}  \frac{\sqrt T}{8a_0 \sigma_0^3} & 0 & 0\\ 
0 & \inv{2a_0^4} &0 \\
0& 0 & a_0^2V_{\theta_0}\end{matrix} \r).
\eea 
We have, for any ball $V_n$, centered on $\xi_0$ and shrinking to $\{\xi_0\}$,

\bea 
\sup_{\xi_n \in V_n} \l\| H_n(\xi_n) - \Gamma(\xi_0) \r\| \to^\proba 0.
\eea 

\end{lemma*}

\begin{proof}
As for the proof of Lemma \ref{lemmaUnifConvScore}, we have 
\bea 
H_n^{(\sigma^2,a^2)}(\xi) \to^\proba   \l( \begin{matrix}  \frac{\sqrt T}{8a_0 \sigma_0^3} & 0 & 0\\ 
0 & \inv{2a_0^4} &0 \\
0& 0 & 0\end{matrix} \r)
\eea 
uniformly in $\xi \in \Xi$ by Lemma A.5 in \cite{clinet2018efficient}, so that by the identity $H_n = H_n^{(\sigma^2,a^2)} +H_n^{(\theta)}$ the lemma will be proved if we can show 
\bea 
H_n^{(\theta)}(\xi) \to^\proba  \l( \begin{matrix} 0 & 0 & 0\\ 
0 & 0 &0 \\
0& 0 & a_0^2V_{\theta_0}\end{matrix} \r),
\eea 
uniformly in $\xi \in \Xi$. Since $H_n^{(\theta)} - \overline{H}_n^{(\theta)} =- \Phi_n^{-1/2} \frac{\partial^2 \K_n}{\partial \xi^2} \Phi_n^{-1/2}$, an immediate application of Lemma \ref{lemmaConsistencyCross} and Lemma \ref{lemmaSobolev} yields $\sup_{\xi \in \Xi }\l\{H_n^{
(\theta)}(\xi)- \overline{H}_n^{(\theta)}(\xi) \r\} \to^\proba 0$. Moreover, by Lemma \ref{lemmaConsistencySquare} and Lemma \ref{lemmaSobolev}, we have 
\bea 
\sup_{\xi \in \Xi} \l\{\overline{H}_n^{(\theta)}(\xi) - \overline{H}_\infty^{(\theta)}(\xi)  \r\} \to^\proba 0,
\eea 
with 
\bea 
\overline{H}_\infty^{(\theta)}(\xi) = \half \l( \begin{matrix}  \frac{\partial^2 \G_{\infty,1}(\xi)}{\partial \l(\sigma^2\r)^2}&0&0\\0&\frac{\partial^2 \G_{\infty,2}(\xi)}{\partial \l(a^2\r)^2}&\frac{\partial^2 \G_{\infty,2}(\xi)}{\partial a^2 \partial \theta }\\0&\frac{\partial^2 \G_{\infty,2}(\xi)}{\partial \theta \partial a^2  }&\frac{\partial^2 \G_{\infty,2}(\xi)}{\partial \theta^2} \end{matrix} \r).
\eea 
Thus, by continuity of $\overline{H}_\infty^{(\theta)}$, we deduce 
\bea 
\sup_{\xi_n \in V_n} \l\{\overline{H}_n^{(\theta)}(\xi_n) - \overline{H}_\infty^{(\theta)}(\xi_0)  \r\} \to^\proba 0,
\eea 
and moreover it is immediate to check that 
\bea 
\overline{H}_\infty^{(\theta)}(\xi_0) =  \l( \begin{matrix}  0&0&0\\0&0&0\\0&0&a_0^2V_{\theta_0} \end{matrix} \r),
\eea
in view of of definitions (\ref{defW}), (\ref{defRho}), (\ref{defG1}) and (\ref{defG2}).
\end{proof}

We now extend the notations of \cite{clinet2018efficient} (A.27)-(A.30) p. 128, and we define a few processes involved in the derivation of the central limit theorem. For $(\beta) \in \{(\sigma^2),(a^2),(\theta)\}$, and $t \in [0,T]$, 
\bea 
M_{1}^{(\beta)}(t) & := & \sum_{i=1}^{N_n(t)}{\frac{ \partial \omega^{i,i}}{\partial \beta}\l\{\l(\Delta X_{i,t}^n\r)^2 - \int_{t_{i-1}^n \wedge t }^{t_{i}^n \wedge t}{\sigma_s^2ds} - \sum_{t_{i-1}^n \wedge t <s \leq t_{i}^n \wedge t} \Delta J_s^2\r\}},
\label{eqM1}\\
 M_{2}^{(\beta)}(t) & := & \sum_{i=1}^{N_n(t)} \l\{\sum_{1 \leq j < i} \frac{ \partial \omega^{i,j}}{\partial \beta} \Delta X_{j,t}^n\r\} \Delta X_{i,t}^n,  
\label{eqM2}\\
 M_{3}^{(\beta)}(t) & := &- 2\sum_{i=0}^{N_n(t)} \l\{\sum_{j = 1}^{N_n(t)} \frac{ \partial \dot{\omega}^{i,j}}{\partial \beta} \Delta X_{j,t}^n\r\} \epsilon_{t_i^n}, 
\label{eqM3}\\
 M_{4}^{(\beta)}(t) & := & \sum_{i=0}^{N_n(t)}{\frac{ \partial \ddot{\omega}^{i,i}}{\partial \beta}\l\{\epsilon_{t_i^n}^2 - a_0^2\r\} } + 2\sum_{i=0}^{N_n(t)} \l\{\sum_{0 \leq j < i} \frac{ \partial \ddot{\omega}^{i,j}}{\partial \beta} \epsilon_{t_j^n}\r\} \epsilon_{t_i^n},
\label{eqM4}\\
M_5^{(\beta)}(t) &:= &\sum_{i=0}^{N_n(t)} \sum_{j=0}^{N_n(t)} {\ddot{\omega}^{i,j} \frac{\partial W_j(\theta_0)}{\partial \beta} \epsilon_{\ti{i}}} +  \sum_{i=1}^{N_n(t)} \sum_{j=0}^{N_n(t)} {\dot{\omega}^{i,j} \frac{\partial W_j(\theta_0)}{\partial \beta} \Delta X_{i,t}^n}. 
\label{eqM5}
\eea
We also define the three-dimensional vectors $M_i(t) := \l(M_{i}^{(\sigma^2)}(t),M_{i}^{(a^2)}(t),M_{i}^{(\theta)}(t)\r)^T$ for $i \in \{1,...,5\}$. In all the definitions (\ref{eqM1})-(\ref{eqM5}), the terms involving the parameters such as $\Omega^{-1}$, $\dot{\Omega}^{-1}$, $\ddot{\Omega}^{-1}$... are evaluated at point $\xi := (\sigma^2,a_0^2,\theta_0)$, for some $\sigma^2 \in [\underline{\sigma}^2,\overline{\sigma}^2]$. For $i \in \{1,...,4\}$, when properly scaled, the processes $M_i(T)$ admit limit distributions whose expressions can be found in Lemma A.6, A.7, A.9 and A.10 of \cite{clinet2018efficient}. We complete these results and show that $M_5(T)$ tends in distribution conditioned on $\calg_T$ to a normal distribution. Before stating the results, we recall that for a $\sigma$-field $\calh$, a random vector $Z$ and a sequence of random vectors $Z_n$ in $\reels^b$, $Z_n$ is said to converge in law towards $Z$ conditioned on $\calh$ if we have for any $u \in \reels^b$ 

\bea 
\esp \l[\l.e^{iu^TZ_n} \r| \calh \r] \overset{\proba}{\rightarrow} \esp \l[\l.e^{iu^T Z} \r| \calh \r].
\eea 
\begin{lemma*}\label{lemmaMtheta}

We have, conditionally on $\calg_T$, the convergence in distribution
\bea 
N_n^{-1/2} M_5(T) \to \caln \l(0,\l( \begin{matrix}  0&0&0\\0&0&0\\0&0&a_0^{-2}V_{\theta_0} \end{matrix} \r) \r). 
\eea 
\end{lemma*}

\begin{proof} 
Since $M_5^{(\sigma^2)}(T) =M_5^{(a^2)}(T) = 0$, it is sufficient to prove the marginal CLT for $M_5^{(\theta)}(T)$. The proof is conducted in two steps.\\

\medskip 

\textbf{Step 1.} We introduce 
\bea 
\tilde{M}_5^{(\theta)}(T) := a_0^{-2}\sum_{i=0}^{N_n}  { \frac{\partial W_i(\theta_0)}{\partial \theta} \epsilon_{\ti{i}}}.
\eea 
We show that $N_n^{-1/2}\l\{M_5^{(\theta)}(T) - \tilde{M}_5^{(\theta)}(T)  \r\} \to^\proba 0$. We decompose 
\bea 
N_n^{-1/2}\l\{M_5^{(\theta)}(T) - \tilde{M}_5^{(\theta)}(T)  \r\} = R_n^{(1)}+R_n^{(2)}+R_n^{(3)},
\eea 
where 
\bea 
R_n^{(1)} = N_n^{-1/2} \sum_{i=1}^{N_n} \sum_{j=0}^{N_n} {\dot{\omega}^{i,j} \frac{\partial W_j(\theta_0)}{\partial \theta} \Delta X_{i,t}^n},
\eea 
\bea 
R_n^{(2)} = N_n^{-1/2} \sum_{i=0}^{N_n} \sum_{j \neq i} {\ddot{\omega}^{i,j} \frac{\partial W_j(\theta_0)}{\partial \theta} \epsilon_{\ti{i}}},
\eea 
and 
\bea 
R_n^{(3)} = N_n^{-1/2} \sum_{i=0}^{N_n} \l\{\ddot{\omega}^{i,i} - a_0^{-2}\r\} \frac{\partial W_i(\theta_0)}{\partial \theta} \epsilon_{\ti{i}}.
\eea 
Now, by a similar proof as in Lemma \ref{lemmaConsistencyCross}, we have easily $R_n^{(1)} \to^\proba 0$. Moreover, by application of Lemma \ref{expansionOmegaPPoint}, Assumption \ref{assCorr} and the independence of $\epsilon$ and $W(\theta_0)$, we easily get that $R_n^{(2)} \to^\proba 0$. Finally, by Lemma \ref{expansionOmegaPPoint} we have $\ddot{\omega}^{i,i} - a_0^{-2} = \bo{N_n^{-1/2}}$ so that we obtain directly $\esp_\calu R_n^{(3)} = 0$ by independence of $\frac{\partial W(\theta_0)}{\partial \theta}$ and $\epsilon$ and for each component $R_{n,k}^{(3)}$ of $R_n^{(3)}$, $1 \leq k \leq d$, 
\beas  
\textnormal{Var}_\calu R_{n,k}^{(3)} &\leq& \bo{N_n^{-2}} \times \sum_{i,j = 0}^{N_n} \esp_\calu \l[\frac{\partial W_i(\theta_0)}{\partial \theta_k} \epsilon_{\ti{i}}\frac{\partial W_j(\theta_0)}{\partial \theta_k} \epsilon_{\ti{j}}\r] \\ 
&=& \bo{N_n^{-1}},
\eeas  
and thus $ R_n^{(3)} \to^\proba 0$. \\

\medskip

\textbf{Step 2.} Now we show that $N_n^{-1/2} \tilde{M}_5^{(\theta)}(T) \to \caln\l(0, a_0^{-2}V_{\theta_0}\r)$ conditionally on $\tilde{\calg}_T = \calg_T \vee \{Q_i^n | i,n,\in \naturels\}$ (and so conditionally on $\calg_T$). To do so, we will apply a conditional version of Theorem 5.12 from \cite{KallenbergFoundation2002}(p. 92) with respect to $\tilde{\calg}_T$. Accordingly, we first remark that we have the representation $N_n^{-1/2} \tilde{M}_5^{(\theta)}(T) = \sum_{i=0}^{N_n} \chi_{i}^n$ where $\chi_i^n = N_n^{-1/2}a_0^{-2} \frac{\partial W_i(\theta_0)}{\partial \theta} \epsilon_{t_i^n}$ are rowwise conditionally independent and centered given $\tilde{\calg}_T$. To get the desired convergence in distribution, it is thus sufficient to show that 
\bea 
N_n^{-1} \sum_{i=0}^{N_n} \esp \l[  \chi_i^n (\chi_i^n)^T  \l| \tilde{\calg}_T \r. \r] \to^\proba  a_0^{-2} V_{\theta_0},
\label{condOrdre2M5}
\eea 
and Lindeberg's condition, for some $p >0$, 
$$ \sum_{i=0}^{N_n} \esp \l[ \|\chi_i^n\|^p \l| \tilde{\calg}_T\r.\r] \to^\proba 0.$$
For (\ref{condOrdre2M5}), we immediately note that 
\bea 
N_n^{-1} \sum_{i=0}^{N_n} \esp \l[  (\chi_i^n)^2 \l| \tilde{\calg}_T \r. \r] &=& a_0^{-4}N_n^{-1}\sum_{i=0}^{N_n}\frac{\partial W_i(\theta_0)}{\partial \theta} \frac{\partial W_i(\theta_0)^T}{\partial \theta} \to^\proba a_0^{-4} V_{\theta_0} 
\label{condLindebergM5}
\eea 
where we have used the independence of $\epsilon$ and $Q$, and where the final step is a straightforward consequence of conditions (\ref{assCorr}) and (\ref{assCumulant}). Moreover Lindeberg's condition (\ref{condLindebergM5}) for $p=4$ is also easily verified using moment conditions and stationarity of the information process.

\end{proof}

\begin{lemma*} \label{lemmaCLTPsi}
We have for any $\sigma^2 \in [\underline{\sigma}^2, \overline{\sigma}^2]$, taking $\xi = (\sigma^2,a_0^2,\theta_0)$, stably in $\calg_T$, the convergence in distribution

\beas
\Phi_n^{1/2}\l\{\Psi_n(\xi) - \overline{\Psi}_n(\xi)\r\} \to  \calm\caln\l(0,\l( \begin{matrix}  \inv{4a_0}\l(\frac{5\calq}{16\sigma^7T^{1/2}}+ \frac{\overline{\sigma}_0^2 \sqrt T}{8\sigma^5} + \frac{\sqrt T}{16 \sigma^3}\r) & 0 & 0\\ 
0 & \inv{2a_0^{4}} + \frac{\textnormal{cum}_4[\epsilon]}{4a_0^8} &0 \\
0& 0 & a_0^{-2}V_{\theta_0}\end{matrix} \r)\r),
\eeas 
where $\calq = \qterm$.
\end{lemma*}

\begin{proof}
Note that we have the decomposition
\beas 
2\Phi_n^{1/2}\l\{\Psi_n(\xi) - \overline{\Psi}_n(\xi)\r\} = \Phi_n^{-1/2} \l\{M_1(T)+2M_2(T)+M_3(T)+M_4(T) +  2M_5(T)\r\}.
\eeas
First we show that  $\Phi_n^{-1/2}\l\{M_3(T) + M_4(T) + 2M_5(T)\r\}$ converges in distribution conditioned on $\calg_T$. By Lemma A.9 and A.10 in \cite{clinet2018efficient}, $\Phi_n^{-1/2}M_3(T)$ and $\Phi_n^{-1/2}M_4(T)$ both tend in distribution conditioned on $\calg_T$ to mixed normal distributions of respective asymptotic variances
\bea
V_3 = \l(\begin{matrix}  \frac{\overline{\sigma}_0^2 \sqrt T}{8\sigma^5a_0}  & 0 & 0\\ 
0 & 0 &0 \\
0& 0 & 0\end{matrix} \r),
\label{eqV3}
\eea
and
\bea 
V_4 = \l(\begin{matrix} \frac{\sqrt T}{16a_0\sigma^2} & 0 & 0\\ 
0 & \frac{2}{a_0^{4}} + \frac{\textnormal{cum}_4[\epsilon]}{a_0^8} &0 \\
0& 0 & 0\end{matrix} \r).
\label{eqV4}
\eea 
Moreover, by independence of $\epsilon$ with the other processes and the fact that $\frac{\partial W_j(\theta_0)}{\partial \theta}$ and $\Delta X$ are uncorrelated, we deduce that the conditional covariance terms between any pair $M_i(T)$,$M_j(T)$ $i \neq j$, $i,j \in \{3,4,5\} $ are null, so that along with the marginal convergence obtained in Lemma \ref{lemmaMtheta} this automatically yields the convergence in law conditioned on $\calg_T$  
\beas 
\Phi_n^{1/2} \l\{M_3(T)+M_4(T) + 2M_5(T) \r\} \to \calm\caln\l(0,\l( \begin{matrix}  \inv{a_0}\l(\frac{\overline{\sigma}_0^2 \sqrt T}{8\sigma^5} + \frac{\sqrt T}{16 \sigma^3}\r) & 0 & 0\\ 
0 & \frac{2}{a_0^{4}} + \frac{\textnormal{cum}_4[\epsilon]}{a_0^8} &0 \\
0& 0 & 4a_0^{-2}V_{\theta_0}\end{matrix} \r)\r). 
\eeas 
By Slutsky's lemma and Lemma A.6 and A.7 from \cite{clinet2018efficient}, we also have the $\calg_T$-stable convergence in distribution of $M_1(T) + 2M_2(T)$ towards a mixed normal distribution of random variance 
\bea 
V_2 = \l(\begin{matrix} \frac{5\calq}{16T^{1/2}\sigma^7a_0} & 0 & 0\\ 
0 & 0 &0 \\
0& 0 & 0\end{matrix} \r).
\label{eqV2}
\eea 
Finally, by application of Proposition A.8 from \cite{clinet2018efficient}, we deduce the joint $\calg_T$-stable convergence of $\l(M_1(T) + 2M_2(T), M_3(T)+M_4(T) + 2M_5(T)\r)$, hence the convergence of the global term $\Phi_n^{1/2}\l\{\Psi_n(\xi) - \overline{\Psi}_n(\xi)\r\}$, and we are done.  

\end{proof}

We are now ready to prove the central limit theorem. 

\begin{proof}[Proof of Theorem \ref{theoremCLT}] 
The proof follows exactly the proof of Theorem A.12 in \cite{clinet2018efficient} in the case $B = 1$, and replacing $\xi = (\sigma^2,a^2)$ by $\xi = (\sigma^2,a^2,\theta)$, and the score function along with the Fisher information by their three dimensional counterparts. 

\end{proof}

\subsection{Proof of Theorem \ref{theoremnoerror}}

In the case $a_0^2 = \eta_0T/n$, that is in the small residual noise framework, we are interested in deriving the joint law of two estimators, which are $\widehat{\xi}_{n, err} = (\widehat{\sigma}_{n, err}^2, \widehat{\theta}_{n, err},\widehat{a}_{n, err}^2)$ and $\widehat{\upsilon}_{n, exp} = (\widehat{\sigma}_{n, exp}^2, \widehat{\theta}_{n, exp})$, which is obtained under the no-residual noise constraint $a = 0$ in the quasi-likelihood function. For the former estimator, the limit theory is fairly different from the case of a fixed non-zero noise. Indeed, we are going to show that the rate of convergence $(N_n^{1/4},N_n^{1/2},N_n^{1/2})$ is changed to $(N_n^{1/2},N_n^{3/2},N_n)$. Moreover, the noise and price increments being of the same order, complex interaction terms now appear in the limit variance of the estimators. To derive a CLT for $\l(\widehat{\upsilon}_{n, exp},\widehat{\xi}_{n, err}\r)$, let us reformulate a bit the problem and, introducing $\widehat{\eta}_{n,err} := \Delta_{N_n}^{-1} \widehat{a}_{n, err}^2$, $\widehat{u}_{n,err} := N_n^{1/2}(\widehat{\theta}_{n,err}-\theta_0)$, and $\widehat{u}_{n,exp} := N_n^{1/2}(\widehat{\theta}_{n,exp}-\theta_0)$, we are now interested in showing that 
\bea 
\widehat{\nu}_n = \l(\widehat{w}_{n,exp},\widehat{\zeta}_{n,err} \r) := (\widehat{\sigma}_{n,exp}^2, \widehat{u}_{n,exp},\widehat{\sigma}_{n,err}^2, \widehat{u}_{n,err}, \widehat{\eta}_{n,err}  )
\eea 
admits a CLT with rate $N_n^{1/2}$. We note that $\widehat{\zeta}_{n,err}$ is the QMLE related to the new log-likelihood
\bea 
\call_{n,err}(\sigma^2,u,\eta) := -\half \textnormal{log det}(\Lambda) - \half (Y - \mu( \theta_0 + N_n^{-1/2}u ))^T \Lambda^{-1} (Y-\mu(\theta_0 + N_n^{-1/2}u)),
\label{logliknew}
\eea 
where now $\Omega$ is replaced by $\Lambda$ with
\bea 
\Lambda = \Delta_{N_n}\left(\begin{matrix}
                    \sigma^2  + 2 \eta & - \eta & 0 & \cdots & 0 \\
                    - \eta & \sigma^2  + 2 \eta & - \eta & \ddots & \vdots \\
                    0 & - \eta & \sigma^2 + 2 \eta & \ddots & 0\\
                    \vdots & \ddots & \ddots & \ddots & - \eta\\
                    0 & \cdots & 0 & - \eta & \sigma^2  + 2 \eta 
                  \end{matrix}\right), 
\eea
and $\widehat{w}_{n,exp}$ is one maximizer in the variables $(\sigma^2,u)$ of 
\bea
\call_{n,exp}(\sigma^2,u) := \call_{n,err}(\sigma^2,u,0),
\label{equivalenceloglik}
\eea 
that is (\ref{logliknew}) with $\Lambda = \Delta_{N_n}\sigma^2 \I_{N_n}$, where $\I_{N_n} \in \reels^{N_n \times N_n}$ is the identity matrix. We now adopt somewhat similar notations to \cite{ait2016hausman}, proof of Theorem 1, and we simplify the problem introducing the change of variables 
\beas  
\phi = 1 - \inv{2\eta} \l\{\sqrt{ \sigma^2(4\eta+\sigma^2)} -\sigma^2\r\} \textnormal{ and } \gamma^2 = \half \l\{2\eta+\sigma^2 + \sqrt{\sigma^2(4\eta+\sigma^2)}\r\},
\eeas
and $\phi = 0$ when $\eta = 0$. We have
\bea
\lambda^{i,j} = \Delta_{N_n}^{-1}\gamma^{-2}\frac{\phi^{|i-j|}-\phi^{i+j}-\phi^{2N_n + 2 - (i+j)}- \phi^{2N_n+2-|i-j|}}{(1-\phi^2)(1-\phi^{2n+2})},
\eea 
with the convention $0^0 = 1$. We are thus going to derive the asymptotic properties of $$\tilde{\nu}_{n} = (\tilde{w}_{n,exp},\tilde{\zeta}_{n,err}) = (\widehat{\sigma}_{n,exp}^2, \widehat{u}_{n,exp},\widehat{\gamma}_{n,err}^2,  \widehat{u}_{n,err},\widehat{\phi}_{n,err}),$$
obtained by maximization of the log-likelihood functions seen respectively as functions of $\zeta := (\gamma^2,u,\phi)$ and $w:=(\sigma^2,u)$. Given the form of $\phi$ and $\eta$, we see that there exists $\overline{\phi} \in (0,1)$ and $\overline{\gamma}^2 > \underline{\gamma}^2 >0 $ such that the optimizations are respectively conducted on the sets $\Xi_{n,err} := \l[\underline{\gamma}^2,\overline{\gamma}^2\r]  \times \{u \in \reels_+^d | \theta_0 + N_n^{-1/2}u \in \Theta \} \times \l[ - \overline{\phi}, \overline{\phi} \r]$, and $\Xi_{n,exp} := \l[\underline{\gamma}^2,\overline{\gamma}^2\r]  \times \{u \in \reels_+^d | \theta_0 + N_n^{-1/2}u \in \Theta \}$. Then, we will get back to $\widehat{\zeta}_{n,err}$ by the delta method. 

\smallskip

We keep similar notations as in the previous part and we write $\Psi_{n,err}(\zeta)  = -N_n^{-1}\frac{\partial \call_{n,err}}{\partial \zeta}$ the score function for the first experiment, and similarly $\Psi_{n,exp}(w)  = -N_n^{-1}\frac{\partial \call_{n,exp}}{\partial w}$. Sometimes we will consider the joint process $\Psi_n = \l(\Psi_{n,exp},\Psi_{n,err}\r)$. We also naturally adapt the notations (\ref{defGnH1}) and (\ref{defKnH1}) to the small noise context 
\bea 
\G_n(\zeta) = ( \mu(\theta_0)-\mu(\theta_0+N_n^{-1/2}u))^T  \Lambda^{-1}  ( \mu(\theta_0)-\mu(\theta_0+N_n^{-1/2}u)),
\label{defGnH0}
\eea
and 
\bea 
\K_n(\zeta) = (\mu(\theta_0) - \mu(\theta_0+N_n^{-1/2}u))^T \Lambda^{-1} \l\{\Delta X + \Delta \epsilon\r\}. 
\label{defKnH0}
\eea 

We first give the limit of both estimators when $a_0^2 = \eta_0 T/n$.
\begin{theorem*}\label{thmConsistencyH0} (consistency) Assume that $a_0^2 = \eta_0 T/n$. Define $\tilde{\eta}_0 := \eta_0T^{-1}\int_0^T{\alpha_s^{-1}ds}$. Let $\phi_0 := 1 - \inv{2\tilde{\eta}_0} \l\{\sqrt{ \overline{\sigma}_0^2(4\tilde{\eta}_0+ \overline{\sigma}_0^2)} - \overline{\sigma}_0^2\r\} $, and $\gamma_0^2 := \half \l\{2\tilde{\eta}_0+\overline{\sigma}_0^2 + \sqrt{\overline{\sigma}_0^2(4\tilde{\eta}_0+\overline{\sigma}_0^2)}\r\}$. Let
$$ \nu_0 := (w_{0,exp},\zeta_{0,err}) = (\overline{\sigma}_0^2 + 2\tilde{\eta}_0, 0,\gamma_0^2 ,0,\phi_0).$$
We have 
$$\tilde{\nu}_{n} \to^\proba \nu_0.$$

In particular, under the null hypothesis $\eta_0=0$, both estimators are consistent as we have
$$ \nu_0 = (\overline{\sigma}_0^2, 0, \overline{\sigma}_0^2,0,0).$$
\end{theorem*}

\begin{proof}
We show separately the consistency of $\tilde{\zeta}_{n,err}$ and $\tilde{w}_{n,exp}$. Let us start with $\tilde{\zeta}_{n,err}$. The methodology for the proof is the same as those of Lemma \ref{lemmaUnifConvScore} and Theorem \ref{thmConsistencyH1}. First, let us define
\bea 
\tilde{\rho}_k := \half \esp \l[ \frac{\partial W_0(\theta_0)}{\partial \theta}\frac{\partial W_k(\theta_0)^T}{\partial \theta} + \frac{\partial W_k(\theta_0)}{\partial \theta}\frac{\partial W_0(\theta_0)^T}{\partial \theta}  \r].
\eea 
Lemma \ref{lemmaConsistencySquare} and \ref{lemmaConsistencyCross} are now easily adapted to 

\bea 
 \sup_{\zeta \in \Xi_{n,err}} \espu \l[ \l(\inv{N_n}\frac{\partial^{\boldsymbol{\alpha}} \G_{n}(\zeta)}{\partial \zeta^{\boldsymbol{\alpha}}} - \frac{\partial^{\boldsymbol{\alpha}} \G_{\infty,err}(\zeta)}{\partial \zeta^{\boldsymbol{\alpha}}} \r)^2 \r] = o_\proba(N_n^{-1/2}),  
 \label{llnGH0}
\eea
for any multi-index $\boldsymbol{\alpha}$ such that $|\boldsymbol{\alpha}| \leq m$, and with 
\beas 
\G_{\infty,err}(\zeta) &:=& \frac{
2}{T\gamma^2(1+\phi)}u^T\l\{\tilde{\rho}_0 - (1-\phi)\sum_{k=1}^{+\infty}{\phi^{k-1} \tilde{\rho}_k}\r\}u \\
&=& \frac{1}{\gamma^2 T}u^T P_{\theta_0, \phi} u
\eeas 
with $P_{\theta_0, \phi} := 2(1+\phi)^{-1}\{\tilde{\rho}_0 - (1-\phi) \sum_{k=1}^{+\infty} \phi^{k-1}\tilde{\rho}_k\}$,
and
\bea 
 \sup_{\zeta \in \Xi_{n,err}} \espu \l[ \l(\inv{N_n}\frac{\partial^{\boldsymbol{\alpha}} \K_n(\zeta)}{\partial \zeta^{\boldsymbol{\alpha}}}\r)^2 \r] = o_\proba(N_n^{-1/2}).  
 \label{llnKH0}
\eea
Therefore, adapting the reasoning from \cite{ait2016hausman} (p. 45) to our setting (by (\ref{eqStepsize}) we have that the step size of the observation grid $\pi_T^n \to^\proba 0$) and combining them with (\ref{llnGH0}), (\ref{llnKH0}) and Lemma \ref{lemmaSobolev}, we get the convergence
\bea 
\sup_{\zeta \in \Xi_{n,err}} \l|\Psi_{n,err}(\zeta) - \Psi_{\infty,err}(\zeta) \r| \to^\proba 0,
\label{convSupErr}
\eea 
where 
\beas  
\Psi_{\infty,err}(\zeta) = \l( \begin{matrix} \inv{2\gamma^2} - \frac{\overline{\sigma}_0^2+2(1-\phi)\tilde{\eta}_0}{2\gamma^4(1-\phi^2) } +\half \frac{\partial \G_{\infty,err}(\zeta)}{\partial \gamma^2}\\\half \frac{\partial \G_{\infty,err}(\zeta)}{\partial u}  \\ \phi\frac{\overline{\sigma}_0^2+2(1-\phi)\tilde{\eta}_0}{\gamma^2(1-\phi^2)^2} - \frac{\tilde{\eta}_0}{\gamma^2(1-\phi^2)} +\half \frac{\partial \G_{\infty,err}(\zeta)}{\partial \phi}\end{matrix} \r).
\eeas
Now, by a classical statistical argument (see e.g \cite{van2000asymptotic}, Theorem 5.9) along with (\ref{convSupErr}), the consistency of $\widehat{\zeta}_{n,err}$ will be proved if for any $\epsilon>0$, $\inf_{\zeta \in \Xi_{\infty,err} : |\zeta-\zeta_0|>\epsilon} \|\Psi_{\infty,err}(\zeta)\| >0$ where $\|x\| = \sqrt{\sum_{i} x_i^2}$, and where $\Xi_{\infty,err} = \l[\underline{\gamma}^2,\overline{\gamma}^2\r]  \times \reels^d \times \l[ - \overline{\phi}, \overline{\phi} \r]$, and if $\Psi_{\infty,err}(\zeta_0)=0$. The second assertion is immediate. To prove the former, let us take $b > 0 $ an arbitrary number and consider $\Xi_{b,err} :=\l[\underline{\gamma}^2,\overline{\gamma}^2\r]  \times [-b,b]^d \times \l[ - \overline{\phi}, \overline{\phi} \r]$. We are going to show that $\inf_{\zeta \in \Xi_{\infty,err} - \Xi_{b,err}  : |\zeta-\zeta_0|>\epsilon} \|\Psi_{\infty,err}(\zeta)\| >0$  on the one hand, and $\inf_{\zeta \in  \Xi_{b,err}  : |\zeta-\zeta_0|>\epsilon} \|\Psi_{\infty,err}(\zeta)\| >0$ on the other hand. In the first case, by hypothesis $u^Tu \geq b^2$, and writing $\tilde{M}(\phi) = \tilde{\rho}_0 - (1-\phi)\sum_{k=1}^{+\infty}{\phi^{k-1} \tilde{\rho}_k}$, we automatically have that $\tilde{M}(\phi)$ is a symmetric positive matrix for any $\phi \in [-\overline{\phi}, \overline{\phi}]$ as a simple consequence of the Cauchy-Schwarz inequality and the stationarity of the information process. Thus, writing $\tilde{c}(\phi) > 0$ the minimal eigenvalue of $M(\phi)^2$ and $\tilde{c} = \min_{\phi \in [-\overline{\phi}, \overline{\phi}]} \tilde{c}(\phi)$, we get that in the first case $$ \|\Psi_{\infty,err}(\zeta)\| > \frac{\tilde{c}}{T \overline{\gamma}^2(1+ \overline{\phi})}b^2 > 0.$$ In the second case, since $\Xi_{b,err}$ is a compact space a continuity argument shows that it is sufficient to prove that $\Psi_{\infty,err}(\zeta) = 0$ if and only if $\zeta = \zeta_{0,err}$. Let thus $\zeta \in \Xi_{b,err}$ such that the score at point $\zeta$ is null. Given the shape of $\G_{\infty,err}$ and the positivity of $M(\phi)^2$ for any $\phi$, $\Psi_{\infty,err}^2(\zeta) = 0$ yields $u=0$. Then 
$$0 = \Psi_{\infty,err}^3(\zeta) = \phi\frac{\overline{\sigma}_0^2+2(1-\phi)\tilde{\eta}_0}{\gamma^2(1-\phi^2)^2}-\frac{\tilde{\eta}_0}{\gamma^2(1-\phi^2)}$$ 
yields the second order equation $\tilde{\eta}_0\phi^2 - (\overline{\sigma}_0^2 + 2\tilde{\eta}_0)\phi + \tilde{\eta}_0=0$, which in turn implies $\phi = \phi_0$ (the other root being non-admissible). Finally $$0 = \Psi_{\infty,err}^1(\zeta) = \inv{2\gamma^2} - \frac{\overline{\sigma}_0^2+2(1-\phi)\tilde{\eta}_0}{2\gamma^4(1-\phi^2) }$$ gives $\gamma^2 = \gamma_0^2$ replacing $\phi_0$ by its expression, and thus $\zeta = \zeta_{0,err}$. In particular, when $\eta_0 = \tilde{\eta}_0 = 0$, this gives $\phi_0=0$ and $\gamma_0^2 = \overline{\sigma}_0^2$.\\

\medskip 

We now derive the limit of $\widehat{w}_{n,exp}$. By the relation (\ref{equivalenceloglik}) we also immediately deduce that 

\beas  
\sup_{\zeta \in \Xi_{n,exp}} \l|\Psi_{n,exp}(w) - \Psi_{\infty,exp}(w) \r| \to^\proba 0,
\eeas  
where 
\beas  
\Psi_{\infty,exp}(w) = \l( \begin{matrix} \inv{2\sigma^2} - \frac{\overline{\sigma}_0^2+2\tilde{\eta}_0}{2\sigma^4 } +\half \frac{\partial \G_{\infty,exp}(w)}{\partial \sigma^2}\\\half \frac{\partial \G_{\infty,exp}(w)}{\partial u}  \end{matrix} \r),
\eeas
and 
\beas 
\G_{\infty,exp}(w) = \frac{2}{T\sigma^2}\l\{u^T\tilde{\rho}_0u - 
u^T\tilde{\rho}_1u\r\} = \inv{T\sigma^2}u^TU_{\theta_0}u.
\eeas 
From there, a similar reasoning to the case applied to $\widehat{\zeta}_{n,err}$ yields the convergence in probability of $\widehat{w}_{n,exp}$ towards $\l(\overline{\sigma}_0^2 + 2\tilde{\eta}_0,0\r)$. In particular, the volatility component is consistent under the null hypothesis $\eta_0=0$ and inconsistent otherwise whereas the information estimator is consistent in both situations.
\end{proof}

We now introduce, as in the previous section, 
\bea 
H_{n,err}(\zeta) = -N_n^{-1} \frac{\partial^2 \call_{n,err}(\zeta)}{\partial \zeta^2},
\label{defHn2err}
\eea 
and
\bea 
H_{n,exp}(w) =- N_n^{-1} \frac{\partial^2 \call_{n,exp}(w)}{\partial w^2}.
\label{defHn2exp}
\eea 
\begin{lemma*} (Fisher information) \label{lemmaFisherH0} Let $\Gamma_{err}(\zeta_{0,err})$ and $\Gamma_{exp}(w_{0,exp})$ be the matrices

\bea 
\Gamma_{err}(\zeta_{0,err}) = \l( \begin{matrix}  \half \gamma_0^{-4} & 0 & 0\\ 
0 &  \gamma_0^{-2}T^{-1}P_{\theta_0} &0 \\
0& 0 &  (1-\phi_0^2)^{-1}\end{matrix} \r)
\eea 
and
\bea 
\Gamma_{exp}(w_{0,exp}) = \l( \begin{matrix}  \half(\overline{\sigma}_0^2 + 2 \tilde{\eta}_0)^{-2} & 0 &\\ 
0 & (\overline{\sigma}_0^2 + 2 \tilde{\eta}_0)^{-1}T^{-1}U_{\theta_0}\end{matrix} \r),
\eea 
where $P_{\theta_0} := P_{\theta_0,\phi_0}$. We have, for any balls $V_{n,err}$ and $V_{n,exp}$ respectively centered on $\zeta_{0,err}$ and $w_{0,exp}$ and shrinking to $\{\zeta_{0,err}\}$ and $\{w_{0,err}\}$,

\bea 
\sup_{\zeta_n \in V_{n,err}} \l\| H_{n,err}(\zeta_n) - \Gamma_{err}(\zeta_{0,err}) \r\| \to^\proba 0
\eea 
and
\bea 
\sup_{w_n \in V_{n,exp}} \l\| H_{n,exp}(w_n) - \Gamma_{exp}(w_{0,exp}) \r\| \to^\proba 0.
\label{fisherExp}
\eea 
\end{lemma*}



\begin{proof}
Adapting the reasoning to get the third and fourth equations (p. 45) in \cite{ait2016hausman} (replacing $a_0^2 \Delta_n^{\gamma_0-1}$ by $\tilde{\eta}_0$), taking the first derivatives of the score functions, and using (\ref{llnGH0}) and (\ref{llnKH0}) for the information part, we directly obtain the convergence for any $b > 0$ 

\bea 
\sup_{\zeta \in \Xi_{b,err}} \l\| H_{n,err}(\zeta) - H_{\infty,err}(\zeta) \r\| \to^\proba 0,
\eea
where
\beas 
H_{\infty,err}(\zeta) := \l( \begin{matrix}  -\inv{2\gamma^4} +\frac{\overline{\sigma}_0^2 + 2(1-\phi)\tilde{\eta}_0}{\gamma^6(1-\phi^2)} & 0 & -\frac{\overline{\sigma}_0^2\phi - \tilde{\eta}_0 (\phi-1)^2}{\gamma^4(1-\phi^2)^2}\\ 
0 & 0 &0 \\
-\frac{\overline{\sigma}_0^2\phi- \tilde{\eta}_0 (\phi-1)^2}{\gamma^4(1-\phi^2)^2}& 0 & \frac{1}{\gamma^2}\l\{\frac{\overline{\sigma}_0^2+2(1-3\phi)\tilde{\eta}_0}{(1-\phi^2)^2}+\frac{4\phi^2 \l(\overline{\sigma}_0^2 + 2(1-\phi)\tilde{\eta}_0\r)}{(1-\phi^2)^3}\r\}\end{matrix} \r) + \half\frac{\partial^2 \G_{\infty,err}(\zeta)}{\partial \zeta^2}.
\eeas 
By continuity of $H_{\infty,err}$, we immediately deduce that 
\bea 
\sup_{\zeta_n \in V_{n,err}} \l\| H_{n,err}(\zeta_n) - H_{\infty,err}(\zeta_{0,err}) \r\| \to^\proba 0,
\eea 
and moreover from the definition of $H_{\infty,err}(\zeta)$ we have 
\bea
H_{\infty,err}(\zeta_{0,err}) =  \Gamma_{err}(\zeta_{0,err}), 
\eea 
since 
$$\half\frac{\partial^2 \G_{\infty,err}(\zeta_{0,err})}{\partial \zeta^2} = \l( \begin{matrix}  0 & 0 &0\\ 
0 & \gamma_0^{-2}T^{-1}P_{\theta_0} &0 \\
0& 0 & 0\end{matrix} \r), $$
and using the relations $\overline{\sigma}_0^2 = \gamma_0^2(1-\phi_0)^2$ and $\tilde{\eta}_0 = \gamma_0^2 \phi_0$. Convergence (\ref{fisherExp}) is proved in the same way.
\end{proof}

Let $\alpha^{i,j}=\phi^{|i-j|}-\phi^{i+j}-\phi^{2N_n + 2 - (i+j)}- \phi^{2N_n+2-|i-j|}$ and $\beta^{i,j} = \partial \alpha^{i,j}/\partial \phi$. We define, similarly to (\ref{eqM1})-(\ref{eqM5}) for $t \in [0,T]$, the martingales

\beas 
S_{1}(t) & := & \sum_{i=1}^{N_n(t)}{\alpha^{i,i}\l\{\l(\Delta X_{i,t}^n\r)^2 - \int_{t_{i-1}^n \wedge t }^{t_{i}^n \wedge t}{\sigma_s^2ds} - \sum_{t_{i-1}^n \wedge t <s \leq t_{i}^n \wedge t} \Delta J_s^2\r\}},
\label{eqS1}\\
 S_{2}^{(\textbf{a})}(t) & := & \sum_{i=1}^{N_n(t)} \l\{\sum_{1 \leq j < i} \textbf{a}^{i,j} \Delta X_{j,t}^n\r\} \Delta X_{i,t}^n, \textnormal{ \textbf{a}} \in \{\alpha,\beta\},  
\label{eqS2}\\
 S_{3}^{(\textbf{a})}(t) & := &- 2\sum_{i=0}^{N_n(t)} \l\{\sum_{j = 1}^{N_n(t)}    \dot{\textbf{a}}^{i,j} \Delta X_{j,t}^n\r\} \epsilon_{t_i^n}, \textnormal{ \textbf{a}} \in \{\alpha,\beta\},   
\label{eqS3}\\
 S_{4}^{(\textbf{a})}(t) & := & \sum_{i=0}^{N_n(t)}{  \ddot{\textbf{a}}^{i,i} \l\{\epsilon_{t_i^n}^2 - n^{-1}\eta_0\r\} } + 2\sum_{i=0}^{N_n(t)} \l\{\sum_{0 \leq j < i}   \ddot{\textbf{a}}^{i,j} \epsilon_{t_j^n}\r\} \epsilon_{t_i^n}, \textnormal{ \textbf{a}} \in \{\alpha,\beta\},
\label{eqS4}\\
S_5(t) &:= & - N_n^{-1/2}\sum_{j=1}^{N_n(t)} \l\{ \sum_{i = 0}^{N_n(t)} \dot{\alpha}^{i,j}\frac{\partial W_i(\theta_0)}{\partial \theta} \r\}\Delta X_{j,t}^n + N_n^{-1/2}\sum_{i=0}^{N_n(t)}   \sum_{j = 0}^{N_n(t)} \ddot{\alpha}^{i,j}{\frac{\partial W_j(\theta_0)}{\partial \theta} \epsilon_{t_i^n}}, 
\label{eqS5}
\eeas
and note that, up to exponentially negligible terms, we have the representation
\beas  
\Psi_n(\nu_0) = \left( \small \begin{matrix}  \frac{1}{2\gamma_0^4(1+\phi_0^2)^{2}T}(S_1(T) + 2S_2^{(\alpha)}(T) + S_3^{(\alpha)}(T) + S_4^{(\alpha)}(T))_{|\phi = 0} \\\frac{1}{\gamma_0^2(1+\phi_0^2)T} S_5(T)_{|\phi=0}\\\frac{1}{2\gamma_0^4(1-\phi_0^2)T}(S_1(T) + 2S_2^{(\alpha)}(T) + S_3^{(\alpha)}(T) + S_4^{(\alpha)}(T))_{|\phi=\phi_0}\\ \frac{1}{\gamma_0^2(1-\phi_0^2)T}S_5(T)_{|\phi=\phi_0}\\-\frac{1}{2\gamma_0^2(1-\phi_0^2)T}\l(\frac{2\phi_0}{1-\phi_0^2}\l\{S_1(T) + 2S_2^{(\alpha)}(T) + S_3^{(\alpha)}(T) + S_4^{(\alpha)}(T) \r\} + 2S_2^{(\beta)}(T) + S_3^{(\beta)}(T) + S_4^{(\beta)}(T)\r)_{|\phi=\phi_0}   \end{matrix}\right).
\label{martingaleRepS}
\eeas 

\begin{lemma*}
Let $\phi \in ]-1,1[$. We have,  $\calg_T$-stably in law, that  
\beas 
N_n^{1/2}S_1(T) \to \calm\caln \l(0, 2T\calq\r).
\eeas 
\label{lemmaS1}
\end{lemma*}

\begin{proof}
Note that $\alpha^{i,i} = 1 - \phi^{2i} - \phi^{2N_n + 2 - 2i} - \phi^{2N_n +2}$. Since $|\phi| < 1$, standard calculations involving Burkholder-Davis-Gundy inequalities, assumption \textbf{(H)} and the finite activity property of the jumps easily yield
$$N_n^{1/2} S_1(T) = N_n^{1/2}\widetilde{S}_1(T) + o_\proba(1),$$
where 
$$ \widetilde{S}_1(T) = \sum_{i=1}^{N_n(t)}{\l\{\l(\Delta X_{i,t}^n\r)^2 - \int_{t_{i-1}^n \wedge t }^{t_{i}^n \wedge t}{\sigma_s^2ds} - \sum_{t_{i-1}^n \wedge t <s \leq t_{i}^n \wedge t} \Delta J_s^2\r\}}.$$
Now, by a straightforward adaptation of Lemma 1 (in the case $q=0$) in the proof of Theorem 2 Appendix A.3 p.40 in \cite{ait2016hausman} in the case of an irregular grid of the form (\ref{specIrregularGrid}), and following the same line of reasoning as the proof of Lemma A.7 in \cite{clinet2018efficient}, we conclude that $\calg_T$-stably in law,
$$ N_n^{1/2}\widetilde{S}_1(T) \to \calm\caln \l(0, 2T\calq\r),$$
and we are done.
\end{proof}

\begin{lemma*}
Let $\phi \in ]-1,1[$ and define $S_2(T) = (S_2^{(\alpha)}(T), S_2^{(\beta)}(T))$. We have,  $\calg_T$-stably in law, that  
\beas 
N_n^{1/2}S_2(T) \to \calm\caln \l(0, T\calq \left( \begin{matrix} \frac{\phi^2}{1-\phi^2}& \frac{\phi}{(1-\phi^2)^2}\\\frac{\phi}{(1-\phi^2)^2}&\frac{1+\phi^2}{(1-\phi^2)^3}\end{matrix}\right)\r).
\eeas 
\label{lemmaS2}
\end{lemma*}
\begin{proof}
Again, given the shapes of $\alpha$ and $\beta$, by standard calculations on martingale increments, introducing $\widetilde{\alpha}^{i,j} = \phi^{|i-j|}$ and $\widetilde{\beta}^{i,j} = |i-j|\phi^{|i-j|-1}$, we easily have
$$ N_n^{1/2}S_2^{(\alpha)}(T) =N_n^{1/2}\widetilde{S}_2^{(\alpha)}(T)  + o_\proba(1) $$
where 
$$ \widetilde{S}_2^{(\alpha)}(T) = \sum_{i=1}^{N_n(t)} \l\{\sum_{1 \leq j < i} \widetilde{\alpha}^{i,j} \Delta X_{j,t}^n\r\} \Delta X_{i,t}^n$$
and a similar statement for $S_2^{(\beta)}(T)$. Proving a central limit theorem for $\widetilde{S}_2(T)$ now boils down to following exactly the same calculations as for $M_2^{(\sigma^2)}(T)$ in the large noise case (see the proof of Lemma A.7 in \cite{clinet2018efficient}) but replacing $\frac{\partial \omega^{i,j}}{\partial \sigma^2}$ by $\widetilde{\alpha}^{i,j} = \phi^{|i-j|}$ and $\widetilde{\beta}^{i,j} = |i-j|\phi^{|i-j|-1}$. In particular, a careful inspection of the proof shows that all the calculations remain valid replacing the scalar $\frac{5}{64T^{3/2}\sigma^7a_0} =\lim_n N_n^{-3/2}\Delta_n \sum_{j =1 }^{N_n}\l(\frac{\partial \omega^{N_n,j}}{\partial \sigma^2}\r)^2$ in the expression of the asymptotic variance, by 
the $2 \times 2$ matrix 
\beas 
\lim_nT\left( \begin{matrix}  \sum_{j=1}^{N_n-1} \l(\widetilde{\alpha}^{N_n,j}\r)^2  &\sum_{j=1}^{N_n-1} \widetilde{\alpha}^{N_n,j} \widetilde{\beta}^{N_n,j} \\\sum_{j=1}^{N_n-1} \widetilde{\alpha}^{N_n,j} \widetilde{\beta}^{N_n,j}&\sum_{j=1}^{N_n-1} \l(\widetilde{\beta}^{N_n,j}\r)^2\end{matrix}\right) &=&T\left( \begin{matrix} \sum_{k=1}^{+\infty} \phi^{2k} & \sum_{k=1}^{+\infty} k\phi^{2k-1}\\\sum_{k=1}^{+\infty} k\phi^{2k-1} & \sum_{k=1}^{+\infty} k^2\phi^{2k-2}\end{matrix}\right) \\ &=& T\left( \begin{matrix} \frac{\phi^2}{1-\phi^2}& \frac{\phi}{(1-\phi^2)^2}\\\frac{\phi}{(1-\phi^2)^2}&  \frac{1+\phi^2}{(1-\phi^2)^3}\end{matrix}\right),
\eeas 
which yields the $\calg_T$-stable convergence in distribution
\beas 
N_n^{1/2}\widetilde{S}_2(T) \to \calm\caln \l(0, T\calq \left( \begin{matrix} \frac{\phi^2}{1-\phi^2}& \frac{\phi}{(1-\phi^2)^2}\\\frac{\phi}{(1-\phi^2)^2}&\frac{1+\phi^2}{(1-\phi^2)^3}\end{matrix}\right)\r).
\eeas
\end{proof}

\begin{lemma*}
\label{lemmaS3}
Let $\phi \in ]-1,1[$ and define $S_3(T) = (S_3^{(\alpha)}(T), S_3^{(\beta)}(T))$. We have,  conditioned on $\calg_T$ the convergence in distribution  
\beas 
N_n^{1/2}S_3(T) \to \calm\caln \l(0, 4\overline{\sigma}_0^2 \tilde\eta_0T^2 \left( \begin{matrix} \frac{2(1-\phi)^2}{1-\phi^2}& \frac{-2(1-\phi)^2}{(1-\phi^2)^2}\\\frac{-2(1-\phi)^2}{(1-\phi^2)^2}&\frac{4(1-\phi)^2}{(1-\phi^2)^3}\end{matrix}\right)\r). 
\eeas                           
\end{lemma*}

\begin{proof}
As in the previous lemma, we introduce the coefficients $\widetilde{\alpha}^{i,j} = \phi^{|i-j|}$ and $\widetilde{\beta}^{i,j} = |i-j|\phi^{|i-j|-1}$. Note that $\dot{\widetilde{\alpha}}^{i,j} = \phi^{|i-j|}(\phi-1)$ for $i \geq j$ and $\dot{\widetilde{\alpha}}^{i,j} = \phi^{|i-j|}(\phi^{-1}-1)$ for $j \geq i+1$. Moreover, if $i \geq j$, $\dot{\widetilde{\beta}}^{i,j}= \phi^{|i-j|}+|i-j|\phi^{|i-j|-1}(\phi-1)$ and if $j \geq i+1$, $\dot{\widetilde{\beta}}^{i,j}= - \phi^{|j-i|-2} + |j-i|\phi^{|j-i|-2}(1-\phi)$. Given the exponential shape of the coefficients, we easily show as for the previous lemma that 
$$ N_n^{1/2}S_3^{(\alpha)}(T) =N_n^{1/2}\widetilde{S}_3^{(\alpha)}(T)  + o_\proba(1) $$
where 
$$ \widetilde{S}_3^{(\alpha)}(T) = - 2\sum_{i=0}^{N_n(t)} \l\{\sum_{j = 1}^{N_n(t)}    \dot{\widetilde{\alpha}}^{i,j} \Delta X_{j,t}^n\r\} \epsilon_{t_i^n}$$
and a similar definition for $\beta$. Now, as for $S_2(T)$, we adapt the proof of $M_3^{(\sigma^2)}(T)$ from the large noise case (proof of Lemma A.9 in \cite{clinet2018efficient}). Again, all the calculations remain valid except that now $\frac{\partial \dot{\omega}^{i,j}}{\partial \sigma^2}$ should be replaced by $\dot{\widetilde{\alpha}}^{i,j}$ and $\dot{\widetilde{\beta}}^{i,j}$, and accordingly, in the limiting variance the scalar $\frac{1}{8\sqrt T \sigma^5 a_0} = \lim_n a_0^2 N_n^{-1/2}\sum_{j=1}^{N_n}\l(\frac{\partial \dot{\omega}^{N_n/2,j}}{\partial \sigma^2} \r)^2 $ is replaced by the $2 \times 2$ matrix 
\beas 
\lim_n \frac{4\eta_0 T N_n}{n}\left( \begin{matrix}  \sum_{j=1}^{N_n-1} \l(\dot{\widetilde{\alpha}}^{N_n/2,j}\r)^2  &\sum_{j=1}^{N_n-1} \dot{\widetilde{\alpha}}^{N_n/2,j} \dot{\widetilde{\beta}}^{N_n/2,j} \\\sum_{j=1}^{N_n-1} \dot{\widetilde{\alpha}}^{N_n/2,j} \dot{\widetilde{\beta}}^{N_n/2,j}&\sum_{j=1}^{N_n-1} \l(\dot{\widetilde{\beta}}^{N_n/2,j}\r)^2\end{matrix}\right) &=&4\tilde{\eta}_0 T\left( \begin{matrix} \frac{2(1-\phi)^2}{1-\phi^2}& \frac{-2(1-\phi)^2}{(1-\phi^2)^2}\\\frac{-2(1-\phi)^2}{(1-\phi^2)^2}&\frac{4(1-\phi)^2}{(1-\phi^2)^3}\end{matrix}\right),
\eeas 
where the last step is obtained by direct calculation on the coefficients, and because $\eta_0 T N_n/n \to^\proba \tilde{\eta}_0T$ by definition of $\tilde{\eta}_0$.
\end{proof}

\begin{lemma*}
\label{lemmaS4}
Let $\phi \in ]-1,1[$, and define $S_4(T) = (S_4^{(\alpha)}(T), S_4^{(\beta)}(T))$. We have,  conditioned on $\calg_T$ the convergence in distribution  
\beas 
N_n^{1/2}S_4(T) \to \caln \l(0,  4\l(\widetilde{\calk}+ 2 \tilde\eta_0^2\r)T^2\left( \begin{matrix} (1-\phi)^2& -(1-\phi)  \\-(1-\phi)  &1\\ \end{matrix}\right) +4\tilde{\eta}_0^2T^2(1-\phi)^4\left( \begin{matrix}   \frac{1}{1-\phi^2}&   -\frac{ (\phi+2)}{(1-\phi^2)^2}  \\ -\frac{ (\phi+2)}{(1-\phi^2)^2} & \frac{ \phi^2+4\phi+5}{(1-\phi^2)^3}\\ \end{matrix}\right) \r).
\eeas
\end{lemma*}

\begin{proof}
As for the previous lemmas, introducing $\widetilde{\alpha}^{i,j} = \phi^{|i-j|}$ and $\widetilde{\beta}^{i,j} = |i-j|\phi^{|i-j|-1}$, we have by standard calculation the approximation
$$N_n^{1/2}S_4^{(\alpha)}(T) = N_n^{1/2}\widetilde{S}_4^{(\alpha)}(T) +o_\proba(1)$$
where 
$$  \widetilde{S}_{4}^{(\alpha)}(t) =  \underbrace{\sum_{i=0}^{N_n(t)}{  \ddot{\widetilde{\alpha}}^{i,i} \l\{\epsilon_{t_i^n}^2 - n^{-1}\eta_0T\r\} }}_{U^{(\alpha)}(t)} + \underbrace{2\sum_{i=0}^{N_n(t)} \l\{\sum_{0 \leq j < i}   \ddot{\widetilde{\alpha}}^{i,j} \epsilon_{t_j^n}\r\} \epsilon_{t_i^n}}_{V^{(\alpha)}(t)},$$
and a similar statement for $S_4^{(\beta)}(T)$. Moreover, we have $\ddot{\widetilde{\alpha}}^{i,j} = -(1-\phi)^2\phi^{|i-j|-1}$ for $i \neq j$, and $\ddot{\widetilde{\alpha}}^{i,j} = 2(1-\phi)$ for $i=j$. Similarly, we have $\ddot{\widetilde{\beta}}^{i,j} = (1-\phi^2)\phi^{|i-j|-2}-|i-j|(1-\phi)^2\phi^{|i-j|-2}$ for $i \neq j$ and $\ddot{\widetilde{\beta}}^{i,j} = -2$ for $i=j$. Now, defining $U = (U^{(\alpha)},U^{(\beta)})$ and $V = (V^{(\alpha)},V^{(\beta)})$, we have that $U$ and $V$ are uncorrelated sums of martingale increments so that it is sufficient to prove that conditionally on $\calg_T$

\beas 
N_n^{1/2}U(T) \to \calm\caln \l(0,  4\l(\widetilde{\calk}+ 2 \tilde\eta_0^2\r)T^2\left( \begin{matrix} (1-\phi)^2& -(1-\phi)  \\-(1-\phi)  &1\\ \end{matrix}\right)\r)
\eeas 
and 
\beas 
N_n^{1/2}V(T) \to  \calm\caln \l(0,  4\tilde{\eta}_0^2T^2(1-\phi)^4\left( \begin{matrix}   \frac{1}{1-\phi^2}&   -\frac{ (\phi+2)}{(1-\phi^2)^2}  \\ -\frac{ (\phi+2)}{(1-\phi^2)^2} & \frac{ \phi^2+4\phi+5}{(1-\phi^2)^3}\\ \end{matrix}\right)\r).
\eeas
The first limit is an immediate consequence of the fact that $U(T)$ is a sum of centered independent and identically distributed variables and the fact that $\ddot{\widetilde{\alpha}}^{i,i} = 2(1-\phi)$ and $\ddot{\widetilde{\beta}}^{i,i} = -2$, and that $\epsilon$ admits a finite fourth order moment. As for $V(T)$ a similar argument to that of $\widetilde{S}_2^{(\alpha)}$ in the proof of Lemma \ref{lemmaS2} yields the convergence in distribution of $N_n^{1/2}V(T)$ to a normal limit with variance matrix
\beas
\lim_n 4\tilde{\eta}_0^2T^2 \left( \begin{matrix}  \sum_{j=0}^{N_n-1} \l(\ddot{\widetilde{\alpha}}^{N_n,j}\r)^2  &\sum_{j=0}^{N_n-1} \ddot{\widetilde{\alpha}}^{N_n,j} \ddot{\widetilde{\beta}}^{N_n,j} \\\sum_{j=0}^{N_n-1} \ddot{\widetilde{\alpha}}^{N_n,j} \ddot{\widetilde{\beta}}^{N_n,j}&\sum_{j=0}^{N_n-1} \l(\ddot{\widetilde{\beta}}^{N_n,j}\r)^2\end{matrix}\right) &=& 4\tilde{\eta}_0^2T^2\left( \begin{matrix}   \frac{ (1-\phi)^4}{1-\phi^2}&   -\frac{ (1-\phi)^4(\phi+2)}{(1-\phi^2)^2}  \\ -\frac{ (1-\phi)^4(\phi+2)}{(1-\phi^2)^2} & \frac{ (1-\phi)^4(\phi^2+4\phi+5)}{(1-\phi^2)^3}\\ \end{matrix}\right)
\eeas 
by direct calculation on the coefficients. Finally, the convergences are both conditional on $\calg_T$ because the process $\epsilon$ is independent of $\calg_T$.

\end{proof}
\begin{lemma*}
Let $\phi \in ]-1,1[$. We have the stable convergence in distribution 
\beas 
N_n^{1/2}\l(S_5(T) - N_n^{-1/2} B_{\theta_0, \phi} \r) \to \calm\caln \l(0, A_{\theta_0,\phi} \r)
\eeas 
where 
$B_{\theta_0,\phi} = \sum_{0<s\leq T} \sum_{k=1}^{N_n} \phi^{|k-i_n(s)|}\frac{\partial \mu_k(\theta_0)}{\partial \theta} \Delta J_s$, and $i_n(s)$ is the only index such that $t_{i-1}^n < t \leq t_i^n$, and 
\beas A_{\theta_0,\phi} &=& 2\int_0^T{\sigma_s^2ds}(1-\phi)^2\l(\frac{\tilde{\rho}_0}{1-\phi^2} +   \sum_{k=1}^{+\infty} \l\{\frac{2\phi^k}{1-\phi^2} - k\phi^{k-1} \r\}\tilde{\rho}_k\r) \\
&+& 2\tilde{\eta}_0T(1-\phi)^3\l( \frac{\phi+3}{1-\phi^2}\tilde{\rho}_0 + \sum_{k=1}^{+\infty} \l\{ \frac{2(1-\phi) \phi^k}{1-\phi^2} -4\phi^{k-1} + (k-1)(1-\phi)\phi^{k-2}\r\}\tilde{\rho}_k \r).
\eeas 
\label{lemmaS5}
\end{lemma*}

\begin{proof}
We apply the same line of reasoning as for $S_1,\cdots,S_4$. Let $\widetilde{\alpha}^{i,j} = \phi^{|i-j|}$ and $\widetilde{\beta}^{i,j} = |i-j|\phi^{|i-j|-1}$. By standard moment calculation, we have 
$$N_n^{1/2}S_5 (T) = N_n^{1/2}\widetilde{S}_5(T) +o_\proba(1)$$
where 
$$  \widetilde{S}_{5}(t) =  - \underbrace{N_n^{-1/2}\sum_{j=1}^{N_n(t)} \l\{ \sum_{i = 0}^{N_n(t)} \dot{\widetilde{\alpha}}^{i,j}\frac{\partial W_i(\theta_0)}{\partial \theta} \r\}\Delta X_{j,t}^n}_{U(t)} + \underbrace{N_n^{-1/2}\sum_{i=0}^{N_n(t)}   \sum_{j = 0}^{N_n(t)} \ddot{\widetilde{\alpha}}^{i,j}{\frac{\partial W_j(\theta_0)}{\partial \theta} \epsilon_{t_i^n}}}_{V(t)}.$$
Let us assume for now that there are no jumps in the price, i.e. $J = 0$. $U$ and $V$ being uncorrelated sums of martingale increments, all we have to prove is that we have the $\calg_T$-stable marginal convergences in distribution

\beas 
N_n^{1/2}U(T) \to \calm\caln \l(0, 2\int_0^T{\sigma_s^2ds}(1-\phi)^2\l(\frac{\tilde{\rho}_0}{1-\phi^2} +   \sum_{k=1}^{+\infty} \l\{\frac{2\phi^k}{1-\phi^2} - k\phi^{k-1} \r\}\tilde{\rho}_k\r)   \r)
\eeas 
and 
\beas 
N_n^{1/2}V(T) \to  \caln \l(0, 2\tilde{\eta}_0T(1-\phi)^3\l( \frac{\phi+3}{1-\phi^2}\tilde{\rho}_0 + \sum_{k=1}^{+\infty} \l\{ \frac{2(1-\phi) \phi^k}{1-\phi^2} -4\phi^{k-1} + (k-1)(1-\phi)\phi^{k-2}\r\}\tilde{\rho}_k \r)  \r).
\eeas
We start with $U$. In that case, we are going to apply Theorem 2-1 from \cite{Jacod1997} to the continuous $\tilde{\calg}_t$-martingale $N_n^{1/2}U$, where $\tilde{\calg}_t := \calg_t \vee \{Q_i^n, i,n \in \naturels\} $ (Note that $X$ and $W$ are still respectively an It\^{o} process and a Brownian motion under $\tilde{\calg}$ in view of the assumptions). Condition (2.8) is satisfied with $B=0$. For condition (2.9), note that we have for any $t \in [0,T]$
\beas 
N_n\langle U,U\rangle_t &=& \sum_{j=1}^{N_n} \l(\sum_{i,i^{'}=0}^{N_n} \dot{\widetilde{\alpha}}^{i,j}\dot{\widetilde{\alpha}}^{i^{'},j} \frac{\partial W_i(\theta_0)}{\partial \theta}\frac{\partial W_{i^{'}}(\theta_0)^T}{\partial \theta}\r) \int_{t_{j-1}^n \wedge t}^{t_j^n \wedge t}\sigma_s^2ds,
\eeas 
which, by similar calculations as for the proof of Lemma \ref{lemmaConsistencySquare} converges in probability to the limit 
\beas 
C_t &=& \lim_n \l(\sum_{i,i^{'}=0}^{N_n} \dot{\widetilde{\alpha}}^{i,N_n/2}\dot{\widetilde{\alpha}}^{i^{'},N_n/2} \esp \l[\frac{\partial W_i(\theta_0)}{\partial \theta}\frac{\partial W_{i^{'}}(\theta_0)^T}{\partial \theta}\r]\r) \int_0^t{\sigma_s^2ds},\\
&=& \l(\frac{2(1-\phi^2)}{1-\phi^2}\tilde{\rho}_0 + 2(1-\phi)^2\sum_{k=1}^{+\infty} \l\{\frac{2\phi^k}{1-\phi^2} - k\phi^{k-1} \r\}\tilde{\rho}_k\r) \int_0^t{\sigma_s^2ds}
\eeas 
by direct calculation on the coefficients $\dot{\widetilde{\alpha}}^{i,j}$, assumptions \textbf{(H)}, (\ref{assCorr}) and (\ref{assCumulant}), and recalling that $\tilde{\rho}_k = \half \esp \l[ \frac{\partial W_0(\theta_0)}{\partial \theta}\frac{\partial W_k(\theta_0)^T}{\partial \theta} + \frac{\partial W_k(\theta_0)}{\partial \theta}\frac{\partial W_0(\theta_0)^T}{\partial \theta}  \r]$ for any $k \in \naturels$. Similarly, we also have Condition (2.10), i.e that $N_n^{1/2} \langle U, W \rangle_t \to^\proba 0$. Finally Condition (2.11) comes from the continuity of $U$ and Condition (2.12) is automatically satisfied since for any bounded martingale $\overline{N}$ orthogonal to $W$ we have $\langle U, \overline{N} \rangle_t = 0$, which yields the $\tilde{\calg}_T$ (and so $\calg_T$) stable convergence. Now we turn to $V$. Note that proving the central limit theorem for $V(T)$ boils down to adapting the reasoning of the proof of $M_5(T)$ in Lemma \ref{lemmaMtheta} replacing $\ddot{\omega}^{i,j}$ by $\ddot{\widetilde{\alpha}}^{i,j}$. A careful inspection of the proof shows that the calculation remains valid but the limiting variance is now expressed as 

\beas 
&&\lim_n \frac{\eta_0T}{n}N_n \sum_{i,i^{'}=0}^{N_n} \ddot{\widetilde{\alpha}}^{i,N_n/2}\ddot{\widetilde{\alpha}}^{i^{'},N_n/2} \esp \l[\frac{\partial W_i(\theta_0)}{\partial \theta}\frac{\partial W_{i^{'}}(\theta_0)^T}{\partial \theta}\r] \\
&=& 2\tilde{\eta}_0 T (1-\phi)^3\l( \frac{\phi+3}{1-\phi^2}\tilde{\rho}_0 + \sum_{k=1}^{+\infty} \l\{ \frac{2(1-\phi) \phi^k}{1-\phi^2} -4\phi^{k-1} + (k-1)(1-\phi)\phi^{k-2}\r\}\tilde{\rho}_k \r)
\eeas 
by direct calculation on the coefficients $\ddot{\widetilde{\alpha}}^{i,N_n/2}$.

\smallskip
When $J\neq 0$, we immediately see that there is an additional term in $U(T)$ of the form $$-\sum_{j=1}^{N_n(t)} \l\{ \sum_{i = 0}^{N_n(t)} \dot{\widetilde{\alpha}}^{i,j}\frac{\partial W_i(\theta_0)}{\partial \theta} \r\}\Delta J_{j,t}^n$$
which, by the finite activity property of $J$ is equal to $-\sum_{j=1}^{N_J} \l\{ \sum_{i = 0}^{N_n} \dot{\widetilde{\alpha}}^{i,i(\tau_j)}\frac{\partial W_i(\theta_0)}{\partial \theta} \r\}\Delta J_{\tau_j}$
for $n$ sufficiently large and where $N_J$ is the random number of times $J$ jumps on $[0,T]$, and $\tau_1,\cdots\tau_J$ are the associated jumping times. Now, note that by the by-part summation formula of Lemma \ref{lemmaTransfo} and the definition of $\dot{\widetilde{\alpha}}^{i,j}$, this term is equal to $B_{\theta_0,\phi}$ up to exponentially negligible terms and we are done.
\end{proof}

\begin{proof}[Proof of Theorem \ref{theoremnoerror}.]
Let us first derive the limit of $N_n^{1/2}(\tilde{\nu}_n - \nu_0)$. To do so, note that we have the martingale representation, up to exponentially negligible terms,
\beas  
\Psi_n(\nu_0) = \left( \tiny \begin{matrix}\frac{1}{2\gamma_0^4(1+\phi_0^2)^{2}T}(S_1(T) + 2S_2^{(\alpha)}(T) + S_3^{(\alpha)}(T) + S_4^{(\alpha)}(T))_{|\phi = 0} \\\frac{1}{\gamma_0^2(1+\phi_0^2)T} S_5(T)_{|\phi=0}\\\frac{1}{2\gamma_0^4(1-\phi_0^2)T}(S_1(T) + 2S_2^{(\alpha)}(T) + S_3^{(\alpha)}(T) + S_4^{(\alpha)}(T))_{|\phi=\phi_0}\\ \frac{1}{\gamma_0^2(1-\phi_0^2)T} S_5(T)_{|\phi=\phi_0}\\-\frac{1}{2\gamma_0^2(1-\phi_0^2)T}\l(\frac{2\phi_0}{1-\phi_0^2}\l\{S_1(T) + 2S_2^{(\alpha)}(T) + S_3^{(\alpha)}(T) + S_4^{(\alpha)}(T) \r\} + 2S_2^{(\beta)}(T) + S_3^{(\beta)}(T) + S_4^{(\beta)}(T)\r)_{|\phi=\phi_0}  \end{matrix}\right).
\label{eqPsiH0}
\eeas
Moreover, the first order condition on $\Psi_n$ yields
\bea 
0 = \Psi_n(\tilde{\nu}_n) = \Psi_n(\nu_0) + H_n(\bar{\nu}_n)\l(\tilde{\nu}_n - \nu_0 \r),
\label{eqFirstOrder}
\eea 
where $\bar{\nu}_n \in \l[\tilde{\nu}_n,\nu_0\r]$. We reformulate (\ref{eqFirstOrder}) as 
\bea
\Gamma(\nu_0)^{-1}H_n(\bar{\nu}_n)N_n^{1/2}\l(\tilde{\nu}_n - \nu_0 \r) = - \Gamma(\nu_0)^{-1}N_n^{1/2}\Psi_n(\nu_0). 
\eea
Thus, by Lemma \ref{lemmaFisherH0}, Slutsky's Lemma for stable convergence, the above martingale representation for $\Psi_n(\nu_0)$ along with lemmas (\ref{lemmaS1}), (\ref{lemmaS2}), (\ref{lemmaS3}), (\ref{lemmaS4}) and (\ref{lemmaS5}) (the brackets $\langle S_i, S_j\rangle_T$ for $i \neq j$ are all negligible so that we have the joint $\calg_T$-stable convergence of the family $(S_i)_{i=1\cdots5}$), we deduce that $\calg_T$-stably in law

\beas 
N_n^{1/2}\l(\tilde{\nu}_n - \nu_0 - N_n^{-1/2}\widetilde{B}_{\theta_0}\r) \to \calm \caln \l(0 , \textbf{U}\r),
\eeas 
where $\widetilde{B}_{\theta_0}=(0,U_{\theta_0}^{-1}B_{\theta_0},0,(1-\phi_0^2)^{-1}P_{\theta_0}^{-1}B_{\theta_0},0,0)^T$, $\textbf{U}$ is the matrix
\footnotesize
$$  \l(\begin{matrix} \textbf{U}_{11}(\calq, \overline{\sigma}_0^2, \tilde{\eta}_0, \widetilde{\calk})  &0&\textbf{U}_{13}(\calq, \gamma_0^2,\phi_0,\widetilde{\calk})&0&\textbf{U}_{15}(\calq, \gamma_0^2,\phi_0,\widetilde{\calk})\\ 0&\l(\int_0^T{\sigma_s^2ds}+3\tilde{\eta}_0T\r)U_{\theta_0}^{-1}&0&\textbf{U}_{24}(\theta_0, \gamma_0^2, \phi_0,\sum_{0<s\leq T}\Delta J_s^2)&0\\\textbf{U}_{13}(\calq, \gamma_0^2,\phi_0,\widetilde{\calk})&0&\textbf{U}_{33}(\calq, \gamma_0^2,\phi_0,\widetilde{\calk})&0&\textbf{U}_{35}(\calq, \gamma_0^2,\phi_0,\widetilde{\calk})\\0&\textbf{U}_{24}(\theta_0, \gamma_0^2, \phi_0,\sum_{0<s\leq T}\Delta J_s^2)&0&\textbf{U}_{24}(\theta_0, \gamma_0^2, \phi_0,\sum_{0<s\leq T}\Delta J_s^2)&0\\
\textbf{U}_{15}(\calq, \gamma_0^2,\phi_0,\widetilde{\calk})&0&\textbf{U}_{35}(\calq, \gamma_0^2,\phi_0,\widetilde{\calk})&0&\textbf{U}_{55}(\calq, \gamma_0^2,\phi_0,\widetilde{\calk})\end{matrix}\r),$$
\normalsize
with
\beas 
\textbf{U}_{11}(\calq, \overline{\sigma}_0^2, \tilde{\eta}_0, \widetilde{\calk}) &=&\frac{2\calq}{T} + 4\widetilde{\calk} + 12\tilde{\eta}_0^2 + 8 \tilde{\eta}_0 \overline{\sigma}_0^2,\\
\textbf{U}_{13}(\calq, \gamma_0^2,\phi_0,\widetilde{\calk}) &=& \frac{2\calq}{(1-\phi_0^2)T} + \frac{4\widetilde{\calk}}{1+\phi_0} + \frac{4 \gamma_0^4 \phi_0 (\phi_0^2 - \phi_0 + 2)}{1+\phi_0}, \\
\textbf{U}_{15}(\calq, \gamma_0^2,\phi_0,\widetilde{\calk}) &=& -\frac{2 \phi_0 \calq}{\gamma_0^2(1-\phi_0^2)T} +\frac{2(1-\phi_0)^2\widetilde{\calk}}{1-\phi_0^2} + \frac{4\gamma_0^2\phi_0(1-\phi_0)^2(\phi_0^2+1)}{1-\phi_0^2} \\
\textbf{U}_{33}(\calq, \gamma_0^2,\phi_0,\widetilde{\calk}) &=& \frac{2(1+\phi_0^2)\calq}{(1-\phi_0^2)^3T} + \frac{4\widetilde{\calk}}{(1+\phi_0)^2} + \frac{4 \gamma_0^4 \phi_0 (\phi_0^2 + \phi_0 + 2)}{(1+\phi_0)^3},\\
\textbf{U}_{35}(\calq, \gamma_0^2,\phi_0,\widetilde{\calk}) &=& -\frac{2(2+\phi_0^2)\calq}{(1-\phi_0^2)^3\gamma_0^2T} + \frac{2(1-\phi_0)\widetilde{\calk}}{(1+\phi_0)^2\gamma_0^2} + \frac{2 \gamma_0^2 \phi_0 (1-\phi_0)(\phi_0^2 + 2)}{(1+\phi_0)^3}\\
\textbf{U}_{55}(\calq, \gamma_0^2,\phi_0,\widetilde{\calk}) &=& \frac{(2\phi_0^4 + 7 \phi_0^2 +1)\calq}{(1-\phi_0^2)^3 \gamma_0^4T} + \frac{(1-\phi_0)^2\widetilde{\calk}}{(1+\phi_0)^2 \gamma_0^4} + \frac{ \phi_0 (4-\phi_0)(1-\phi_0)(\phi_0^2 + 1)}{(1+\phi_0)^3}
\eeas 
and for the terms involving $\theta_0$,
\beas 
\textbf{U}_{24}(\theta_0, \gamma_0^2, \phi_0,\sum_{0<s\leq T}\Delta J_s^2) &=& 2 \gamma_0^2T(1-\phi_0^2)^{-1} U_{\theta_0}^{-1}\big((1-\phi_0)\big\{(\phi_0^4-4\phi_0^3+5\phi_0^2-\phi_0+1)\tilde{\rho}_0\\ 
& & + (\phi_0^3 -\phi_0^2+3\phi_0-1)\tilde{\rho}_1 \big\} \\
&+& 2\phi_0(1-\phi_0)^2 \tilde{\rho}_2 + (2-\phi_0)(1-\phi_0)^4\sum_{k=2}^{+\infty}\phi_0^{k}\tilde{\rho}_k \big)P_{\theta_0}^{-1}\\
&-& 2(1-\phi_0^2)^{-1} U_{\theta_0}^{-1}\l\{(1-\phi_0)^3\tilde{\rho}_0 - (1-\phi_0)^2\tilde{\rho}_1 +(1-\phi_0)^3\sum_{k=2}^{+\infty}\phi_0^{k}\tilde{\rho}_k\r\} \\ & & \times P_{\theta_0}^{-1}\sum_{0<s\leq T}\Delta J_s^2,\\
\textbf{U}_{44}(\theta_0, \gamma_0^2, \phi_0,\sum_{0<s\leq T}\Delta J_s^2) &=& \gamma_0^2T P_{\theta_0}^{-1} - \frac{2(1-\phi_0)^2}{(1-\phi_0^2)^2}P_{\theta_0}^{-1}\l(\frac{\tilde{\rho}_0}{1-\phi_0^2}+ \sum_{k=1}^{+\infty}\l\{\frac{2\phi_0^k}{1-\phi_0^2} -k\phi_0^{k-1} \r\}\tilde{\rho}_k\r) \\ & & \times P_{\theta_0}^{-1}\sum_{0<s\leq T}\Delta J_s^2.\\
\eeas 
\medskip
Finally, recalling that $\l(\widehat{\sigma}_{n,exp}^2,\widehat{\theta}_{n,exp},\widehat{\sigma}_{n,err}^2,\widehat{\theta}_{n,exp},\widehat{a}_{n,err}^2\r)$ is equal to

\beas 
 \l(\widehat{\sigma}_{n,exp}^2, \theta_0 + N_n^{-1/2}\widehat{u}_{n,exp},\widehat{\gamma}_{n,err}^2(1-\widehat{\phi}_{n,err})^2, \theta_0 + N_n^{-1/2}\widehat{u}_{n,err}, n^{-1}T\widehat{\gamma}_{n,err}^2 \widehat{\phi}_{n,err}\r),
\eeas 
a straightforward application of the delta method yields Theorem \ref{theoremnoerror}.
\end{proof}

\subsection{Proofs related to the test}

\begin{proof}[Proof of Proposition \ref{propAVAR}]
We start by showing our claim under $\calh_0$. Note that the case $\widehat{V}_3$ is a consequence of Theorem \ref{theoremnoerror}. For $k \neq 3$, we conduct the proof for $\widehat{V}_{k}$ in several steps.

\textbf{Step 1.} We show that in the expressions of the variance estimators, we can replace the estimate returns $\Delta \widehat{X}_{i}^n$ by the efficient ones $\Delta X_{i}^n$. For the sake of brevity we prove it in the case $k = 1$. Note that the cases $k = 2,4,5$ can be proved following the same line of reasoning. Introducing $\overline{V}_1 := \frac{4N_n}{T^2} \sum_{i=2}^{N_n-1} \l(\Delta X_{i}^n\r)^2\l(\Delta X_{i-1}^n\r)^2$, we have to show that $\widehat{V}_1 - \overline{V}_1 \to^\proba 0$. Since by definition $\widehat{X}_{\ti{i}} = X_{\ti{i}} + \phi(Q_i^n,\theta_0) - \phi(Q_i^n,\widehat{\theta}_{n,exp})$, if we introduce $b_i(\theta) = \mu_i(\theta_0)-\mu_i(\theta)$, we have the representation
\bea 
\Delta \widehat{X}_{i}^n =\Delta X_{i}^n + b_i(\widehat{\theta}_{n,exp}).
\label{eqReturnsEspEff}
\eea 
Developing $\widehat{V}_1$ and using (\ref{eqReturnsEspEff}), we get 
\beas 
\widehat{V}_1 - \overline{V}_1 = A_{i,i-1}^n + A_{i-1,i}^n,
\eeas 
where
\beas 
A_{i,i-1}^n &=& \frac{4N_n}{T^2} \sum_{i=2}^{N_n-1} \big\{2b_{i-1}(\widehat{\theta}_{n,exp})b_i(\widehat{\theta}_{n,exp})\Delta X_{i-1}^n\Delta X_i^n + 2b_{i}(\widehat{\theta}_{n,exp})^2b_{i-1}(\widehat{\theta}_{n,exp})\Delta X_{i-1}^n\\
&+&2b_{i}(\widehat{\theta}_{n,exp})\l(\Delta X_{i-1}^n\r)^2\Delta X_{i}^n +b_{i}(\widehat{\theta}_{n,exp})^2\l(\Delta X_{i-1}^n\r)^2 + \half b_{i-1}(\widehat{\theta}_{n,exp})^2b_i(\widehat{\theta}_{n,exp})^2 \big\},\\
&=& \sum_{j=1}^5 A_{i,i-1}^n[j],
\eeas 
and $A_{i-1,i}^n$ has the same expression as above inverting the role of $i$ and $i-1$. Now, using the expansion
\bea 
b_i(\widehat{\theta}_{n,exp}) = \frac{\partial \mu_i(\theta_0)}{\partial \theta}\l(\widehat{\theta}_{n,exp} - \theta_0\r) + \half \frac{\partial^2 \mu_i(\tilde{\theta}_n)}{\partial \theta^2}\l(\widehat{\theta}_{n,exp} - \theta_0\r)^{\otimes 2},
\eea 
for some $\tilde{\theta}_n \in [\theta_0, \widehat{\theta}_{n,exp}]$, along with the fact that $\widehat{\theta}_{n,exp} - \theta_0 = O_\proba(N_n^{-1})$ by Theorem \ref{theoremnoerror} and that $\esp \l[\l.\sup_{\theta \in \Theta}\l|\frac{\partial^j \mu_i(\theta)}{\partial \theta^j}\r| \r|^p X\r] < \infty $ independent of $n$ and for any $p \geq 1$, $j \leq 2$, we easily deduce by direct calculation that each $A_{i,i-1}^n[j] = o_\proba(1)$. \\

\textbf{Step 2.} Now we have to show that for any $k \in \{1,2,4,5\}$, we have the convergence $\overline{V}_k \to^\proba AVAR(\widehat{\sigma}_{exp}^2 - \widehat{\sigma}_{err}^2)$, where $\overline{V}_k$ has the same expression as $\widehat{V}_k$ except that each $\Delta \widehat{X}_i^n$ is replaced by the efficient return $\Delta X_i^n$. The cases $k=4,5$ have already been tackled in \cite{ait2016hausman}, so that it remains to show it for $k=1,2$. In this step we show the case $k=1$, that is when there is no jump in the price process ($J=0$), and when $AVAR(\widehat{\sigma}_{exp}^2 - \widehat{\sigma}_{err}^2) = 4 T^{-2}  \int_0^T\alpha_s^{-1}ds \int_0^T\sigma_s^4 \alpha_s ds$. Let us introduce $\tilde{V}_1 = \frac{4N_n}{T^2}\sum_{i=2}^{N_n}\sigma_{\ti{i-2}}^4\Delta \ti{i-1}\Delta \ti{i}$. We first show that $\overline{V}_1 - \tilde{V}_1 = o_\proba(1)$. Note that
\beas 
\overline{V}_1 - \tilde{V}_1 = P_n^{(1)} + P_n^{(2)}, 
\eeas 
with
\beas 
P_n^{(1)} = \frac{4N_n}{T^2}\sum_{i=2}^{N_n}\l(\Delta X_{i-1}^n\r)^2 \l\{\l(\Delta X_{i}^n\r)^2-\sigma_{\ti{i-2}}^2\Delta \ti{i}\r\},
\eeas
and
\beas 
P_n^{(2)} = \frac{4N_n}{T^2}\sum_{i=2}^{N_n} \l\{\l(\Delta X_{i-1}^n\r)^2-\sigma_{\ti{i-2}}^2\Delta \ti{i-1}\r\} \sigma_{\ti{i-2}}^2\Delta \ti{i}.
\eeas 
We show that $P_n^{(1)} \to^\proba 0$. Assume first that the volatility process has no jumps ($\tilde{J} = 0$). Note that $P_n^{(1)} = \sum_{i=2}^{N_n-1} \chi_i^n$ with $\chi_i^n = \frac{4N_n}{T^2}\l(\Delta X_{i-1}^n\r)^2 \l\{\l(\Delta X_{i}^n\r)^2-\sigma_{\ti{i-2}}^2\Delta \ti{i}\r\}$, so that by Lemma 2.2.11 in \cite{JacodLimit2003}, we only need to show $\sum_{i=2}^{N_n-1} \esp[\chi_i^n|\calg_{i-1}^n] \to^\proba 0$ on the one hand, and $\sum_{i=2}^{N_n-1} \esp[\l(\chi_i^n\r)^2|\calg_{i-1}^n] \to^\proba 0$ on the other hand. We have
\beas 
\l|\sum_{i=2}^{N_n-1} \esp[\chi_i^n|\calg_{i-1}^n]\r| &\leq& \frac{4N_n}{T^2}\sum_{i=2}^{N_n-1}\l(\Delta X_{i-1}^n\r)^2 \esp\l[ \l.\int_{\ti{i-1}}^{\ti{i}}\l|\sigma_s^2 -\sigma_{t_{i-2}}^2\r|ds \r| \calg_{i-1}^n \r]\\
&\leq& KN_nn^{-3/2+3/2\gamma} \underbrace{{\sum_{i=2}^{N_n-1}\l(\Delta X_{i-1}^n\r)^2}}_{O_\proba(1)} \\
&\to^\proba& 0,
\eeas
since by Assumption \textbf{(H)} and (\ref{eqDeltaX}) for the continuous It\^{o} semi-martingale $\sigma^2$ we have 
$$\esp\l[ \l.\int_{\ti{i-1}}^{\ti{i}}\l|\sigma_s^2 -\sigma_{\ti{i-2}}^2\r|ds \r| \calg_{i-1}^n \r] \leq \Delta t_{i}^n \esp\l[ \l.\sup_{s \in [\ti{i-1},\ti{i}]} \l|\sigma_s^2 -\sigma_{\ti{i-2}}^2\r| \r| \calg_{i-1}^n \r] \leq Kn^{-3/2+3/2\gamma}.$$
Moreover,
\beas 
\esp_\calu \Big[ \sum_{i=2}^{N_n-1} \esp[\l(\chi_i^n\r)^2|\calg_{i-1}^n] \Big] &=& \frac{16N_n^2}{T^4}\esp_\calu \bigg[ \sum_{i=2}^{N_n-1}\l(\Delta X_{i-1}^n\r)^4 \esp\l[\l.\l\{\l(\Delta X_{i}^n\r)^2-\sigma_{\ti{i-2}}^2\Delta \ti{i}\r\}^2\r|\calg_{i-1}^n\r] \bigg]\\
&\leq& K N_n^2 n^{-2+2\gamma} \sum_{i=2}^{N_n-1}\esp_\calu\l(\Delta X_{i-1}^n\r)^4 \\
&\leq& K N_n^3 n^{-4+4\gamma} \to^\proba 0, 
\eeas
where again we have used \textbf{(H)} and (\ref{eqZeta2}). Finally, when $\tilde{J} \neq 0$, by the finite activity property, only a finite number of terms are affected in the above sums, and it is easy to see that the convergence still holds in that case. Thus we have proved $P_n^{(1)} \to^\proba 0$, and $P_n^{(2)} \to^\proba 0$ is proved similarly. Now, recalling (\ref{specIrregularGrid}), we decompose 
$$\tilde{V}_1 - 4 T^{-2}  \int_0^T\alpha_s^{-1}ds \int_0^T\sigma_s^4 \alpha_s ds = Q_n^{(1)} + Q_n^{(2)},$$
with
\beas 
Q_n^{(1)} = \frac{4N_n}{T^2}\Delta_n \sum_{i=2}^{N_n-1}\sigma_{\ti{i-2}}^4\alpha_{\ti{i-2}}\l\{U_{i-1}^n - 1\r\} \Delta \ti{i},
\eeas 
and 
\beas 
Q_n^{(2)} = \frac{4N_n}{T^2}\Delta_n \sum_{i=2}^{N_n-1}\sigma_{\ti{i-2}}^4\alpha_{\ti{i-2}} \Delta \ti{i} - 4 T^{-2}  \int_0^T\alpha_s^{-1}ds \int_0^T\sigma_s^4 \alpha_s ds.
\eeas
Using assumption \textbf{(H)} and the fact that the $U_i^n$'s are i.i.d, independent of the other quantities such that $\esp[U_i^n] = 1$, we easily deduce that $\esp_\calu \l(Q_n^{(1)}\r)^2 \leq KN_n n^{-2+2\gamma} \to^\proba 0$. Moreover, $Q_n^{(2)} \to^\proba 0$ is a direct consequence of $N_n \Delta_n\to^\proba \int_0^T\alpha_s^{-1}ds$, and the convergence of the Riemann sum $\sum_{i=2}^{N_n-1}\sigma_{\ti{i-2}}^4\alpha_{\ti{i-2}} \Delta \ti{i} \to^\proba \int_0^T{\sigma_s^4\alpha_sds}$. \\

\textbf{Step 3.} Finally we show the case $k=2$. We are going to show both convergences 
\beas  
A_n^{(1)} := \frac{4N_n}{T^2}\sum_{i=2}^{N_n-1}\l(\Delta X_i^n\r)^2 \l(\Delta X_{i-1}^n\r)^2 \mathbf{1}_{\{\mid \Delta X_{i}^n \mid \leq \widetilde{u}_i \}} \mathbf{1}_{\{\mid \Delta X_{i-1}^n \mid \leq \widetilde{u}_{i-1} \}} \to^\proba 4 T^{-2}  \int_0^T\alpha_s^{-1}ds \int_0^T\sigma_s^4 \alpha_s ds
\label{eqContinuousPart}
\eeas
and
\beas  
A_n^{(2)} & := & \frac{4}{T}\sum_{i=\tilde{k}_n}^{N_n-1-\tilde{k}_n}\l(\Delta X_i^n\r)^2\mathbf{1}_{\{\mid \Delta X_{i}^n \mid > \widetilde{u}_{i}\}}\l(\widehat{\sigma^{2}_{t_{i}} \alpha_{t_i}} +\widehat{\sigma^{2}_{t_{i-}} \alpha_{t_i-}}\r) \\
& \to^\proba & 4 T^{-2}  \int_0^T\alpha_s^{-1}ds\l\{\sum_{0 < s \leq T} \Delta J_s^2 (\sigma_s^2\alpha_s + \sigma_{s-}^2\alpha_{s-})\r\}.
\label{eqJumpPart}
\eeas
For $A_n^{(1)}$, we first show that we can replace $\Delta X_i^n$ by its continuous part $\Delta \tilde{X}_i^n$ in the square increments of the formula. To do so, define 
\beas 
B_n^{(1)} =  \frac{4N_n}{T^2}\sum_{i=2}^{N_n-1}\l(\Delta \tilde{X}_i^n\r)^2 \l(\Delta X_{i-1}^n\r)^2 \mathbf{1}_{\{\mid \Delta X_{i}^n \mid \leq \widetilde{u}_i \}} \mathbf{1}_{\{\mid \Delta X_{i-1}^n \mid \leq \widetilde{u}_{i-1} \}},
\eeas 
\beas 
B_n^{(2)} =  \frac{4N_n}{T^2}\sum_{i=2}^{N_n-1}\l(\Delta \tilde{X}_i^n\r)^2 \l(\Delta \tilde{X}_{i-1}^n\r)^2 \mathbf{1}_{\{\mid \Delta X_{i}^n \mid \leq \widetilde{u}_i \}} \mathbf{1}_{\{\mid \Delta X_{i-1}^n \mid \leq \widetilde{u}_{i-1} \}}.
\eeas 
Let us define $(\tau_q)_{1 \leq q \leq N_J <+ \infty \textnormal{ a.s}}$ be the successive jump times of $J$ (i.e $\Delta J_{\tau_q} \neq 0$ a.s), and for $1 \leq q \leq N_J$, $i_q$ is such that $\ti{i_q-1} < \tau_q \leq \ti{i_q}$. By the finite activity property of $J$, for $n$ sufficiently large we have 
\beas 
|A_n^{(1)} - B_n^{(1)}| &=& \frac{4N_n}{T^2}\sum_{q=1}^{N_J}\l(\l(\Delta J_{\tau_q}\r)^2 + 2\l|\Delta J_{\tau_q}\Delta \tilde{X}_{i_q}^n\r| \r) \l(\Delta X_{i-1}^n\r)^2 \mathbf{1}_{\{\mid \Delta X_{i}^n \mid \leq \widetilde{u}_i \}} \mathbf{1}_{\{\mid \Delta X_{i-1}^n \mid \leq \widetilde{u}_{i-1} \}},\\
&\leq& \frac{4N_n}{T^2}\sum_{q=1}^{N_J}\underbrace{\l(\l(\Delta J_{\tau_q}\r)^2 + 2\l|\Delta J_{\tau_q}\Delta \tilde{X}_{i_q}^n\r| \r)}_{O_\proba(1)} \underbrace{\l(\Delta X_{i-1}^n\r)^2}_{O_\proba(n^{-1+\gamma})} \underbrace{\tilde{\alpha} \l(\Delta t_{i_q}^n\r)^\omega}_{O_\proba(n^{-\omega+\gamma\omega})}\underbrace{|\Delta \tilde{X}_{i_q}^n + \Delta J_{\tau_q}|^{-1}}_{O_\proba(1)},\\
\eeas 
where we have used assumption \textbf{(H)}, and the fact that $\Delta J_{\tau_q} \neq 0$ whereas $\Delta \tilde{X}_{i_q}^n =o_\proba(1)$ for the estimate $|\Delta \tilde{X}_{i_q}^n + \Delta J_{\tau_q}|^{-1} = O_\proba(1)$. Since the sum is finite almost surely, overall we deduce $|A_n^{(1)} - B_n^{(1)}| = O_\proba(n^{\gamma -\omega +\gamma\omega})$ and since $\gamma$ can be considered arbitrary close to $0$ we deduce that $A_n^{(1)} - B_n^{(1)} \to^\proba 0$. Similarly we show that $B_n^{(1)}-B_n^{(2)} \to^\proba 0$. Now we get rid of the indicator functions in $B_n^{(2)}$. Define 

$$B_n^{(3)} = \frac{4N_n}{T^2}\sum_{i=2}^{N_n-1}\l(\Delta \tilde{X}_i^n\r)^2 \l(\Delta \tilde{X}_{i-1}^n\r)^2.$$
Then easy calculation gives
\beas 
\esp_\calu|B_n^{(2)} - B_n^{(3)}| &=& \frac{4N_n}{T^2}\esp_\calu\sum_{i=2}^{N_n-1}\l(\Delta \tilde{X}_i^n\r)^2 \l(\Delta \tilde{X}_{i-1}^n\r)^2 \mathbf{1}_{\{\mid \Delta X_{i}^n \mid > \widetilde{u}_i \} \cup \{ \mid \Delta X_{i-1}^n \mid > \widetilde{u}_{i-1} \}}\\
&\leq& \frac{4N_n}{T^2}\esp_\calu\sum_{i=2}^{N_n-1}\l(\Delta \tilde{X}_i^n\r)^2 \l(\Delta \tilde{X}_{i-1}^n\r)^2 \l(\frac{\mid \Delta X_{i}^n \mid}{\tilde{\alpha} \l(\Delta t_i^n\r)^\omega} +\frac{\mid \Delta X_{i-1}^n \mid}{\tilde{\alpha} \l(\Delta t_i^n\r)^\omega}  \r) \\
&\leq& KN_n^2 \times n^{-5/2+5/2\gamma + \omega - \omega\gamma} \to^\proba 0 
\eeas 
since $\omega < 1/2$. To conclude, note that by Step 2 of this proof $B_n^{(3)} \to^\proba 4 T^{-2}  \int_0^T\alpha_s^{-1}ds \int_0^T\sigma_s^4 \alpha_s ds$ so that combined with $A_n^{(1)} - B_n^{(3)} \to^\proba 0$, this yields the desired convergence for $A_n^{(1)}$ as well. For $A_n^{(2)}$, by similar techniques as above, defining 
\bea
C_n^{(1)} = \frac{4}{T}\sum_{q=1}^{N_J}\Delta J_{\tau_q}^2\l(\widehat{\sigma^{2}_{t_{i_q}} \alpha_{t_{i_q}}} +\widehat{\sigma^{2}_{t_{{i_q}-}} \alpha_{t_{i_q}-}}\r),
\label{jumpCaseSpotEst}
\eea 
we easily deduce $A_n^{(2)} = C_n^{(1)} + o_\proba(1)$. Moreover, we also easily deduce by assumption \textbf{(H)} along with the fact that $\tilde{k}_n\Delta_n \to 0$ and $\tilde{k}_n \to \infty$,
\beas 
\widehat{\sigma^{2}_{t_{i_q}} \alpha_{t_{i_q}}} &=& \frac{N_n}{\tilde{k}_n T}\sum_{j = i_q}^{i_q + \tilde{k}_n -1} \l(\Delta X_j^n\r)^2 \mathbf{1}_{\{\mid \Delta X_{j}^n \mid \leq \widetilde{u}_j\}}\\
&=& \frac{N_n}{\tilde{k}_n T}\sum_{j = i_q}^{i_q + \tilde{k}_n -1} \l(\Delta \tilde{X}_j^n\r)^2 +o_\proba(1)\\
&=& \frac{N_n}{\tilde{k}_n T}\sum_{j = i_q}^{i_q + \tilde{k}_n -1} \sigma_{\tau_q}^2\Delta t_j^n +o_\proba(1)\\
&=& \frac{N_n \Delta_n }{ T}\sigma_{\tau_q}^2\alpha_{\tau_q} \underbrace{\inv{\tilde{k}_n}\sum_{j = i_q}^{i_q + \tilde{k}_n -1} U_j^n}_{\to^\proba \esp[U_1^1] = 1} +o_\proba(1)\\
 &\to^\proba& T^{-1}\int_0^T{\alpha_s^{-1}ds}\sigma_{\tau_q}^2\alpha_{\tau_q},
\eeas 
where we have used that $N_n\Delta_n \to^\proba\int_0^T{\alpha_s^{-1}ds} $, and the law of large numbers for the i.i.d sequence $(U_i^n)_{i,n}$. Similarly, we also have $\widehat{\sigma^{2}_{t_{i_q}-} \alpha_{t_{i_q}-}} \to^\proba T^{-1}\int_0^T{\alpha_s^{-1}ds}\sigma_{\tau_q-}^2\alpha_{\tau_q-}.$ Finally, combined with (\ref{jumpCaseSpotEst}), we deduce the desired convergence for $A_n^{(2)}$. 

Under the alternative, similar techniques yield that when $\eta_0 > 0$ the variance estimators $\widehat{V}_k$, $k \in \{1,\cdots,5\}$ remain of order $O_\proba(1)$ (although they become inconsistent).
\end{proof}

We show in what follows the consistency of the test.

\begin{proof}[Proof of Corollary \ref{testConsistency}.]
Under the null hypothesis, the corollary is a direct consequence of Theorem \ref{theoremnoerror} (with $\eta_0=  \calk = 0$) and Proposition \ref{propAVAR}. Under the alternative $\eta_0 > 0$, by Theorem \ref{theoremnoerror}, we have 
$$\widehat{\sigma}_{n,err}^2 - \widehat{\sigma}_{n,exp}^2 \to^\proba 2\tilde{\eta}_0 > 0\textnormal{ } \proba-\textnormal{a.s.}$$
so that this yields $\frac{N_n\l(\widehat{\sigma}_{n,err}^2 - \widehat{\sigma}_{n,exp}^2\r)^2}{\widehat{V}_k} \to^\proba +\infty$ (by Proposition \ref{propAVAR} we have $\widehat{V}_k = O_\proba(1)$), which completes the proof. 
\end{proof}
Finally, we show Corollary \ref{corollaryRV}.

\begin{proof}[Proof of Corollary \ref{corollaryRV}]
We start by showing Formula (\ref{formulaSigmaExp}) along with the fact that $\widehat{\sigma}_{n,exp}^2$ is equal to the least square estimator of \cite{li2016efficient}, introduced in (9). Indeed, (\ref{formulaSigmaExp}) is obtained solving directly the first order condition $\frac{\partial l_{n,exp}}{\partial \sigma^2}(\widehat{\sigma}_{n,exp}^2,\widehat{\theta}_{n,exp}) = 0$, using Definition (\ref{loglikexp}) for $l_{n,exp}$. Moreover, the first order condition for $\widehat{\theta}_{n,exp}$ reads $\partial_\theta \mu(\widehat{\theta}_{n,exp})^T(Y-\mu(\widehat{\theta}_{n,exp})) = 0$, which is also the first order condition related to the quadratic loss introduced in equation (9) in \cite{li2016efficient}. This proves that $\widehat{\theta}_{n,exp}$ and the estimator (9) from \cite{li2016efficient} coincide. Now, the convergence (\ref{efficientpriceestimator2}) is a straightforward consequence of the consistency of $\widehat{\theta}_{n,exp}$ under $\calh_0$ along with the right continuity of the efficient price $X$. Therefore, the convergences stated in (\ref{cltRVestimated}) and (\ref{cltRVestimated2}) are particular cases of Theorem \ref{theoremnoerror}.
\end{proof}

\begin{proof}[Proof of Proposition \ref{propAVARlargenoise}]
When the noise is large (i.e. $a_0^2 >0$ is fixed), similar calculation to that of the proof of Proposition \ref{propAVAR} yields
\beas 
\widehat{V}_1 &=& \frac{4N_n}{T^2}\sum_{i=2}^{N_n} \underbrace{\Delta \epsilon_{t_i^n}^2 \Delta \epsilon_{t_{i-1}^n}^2}_{O_\proba(1)} + o_\proba(N_n^2) \\
&=& O_\proba(N_n^2).
\eeas 
Similarly, we have $\widehat{V}_i = O_\proba(N_n^2)$ for $i = 2,\cdots,5$.
\end{proof}

\begin{proof}[Proof of Corollary \ref{testConsistencylargenoise}]
Under the large noise alternative, adapting the proof of Theorem \ref{thmConsistencyH0} for a fixed $a_0^2 >0$ easily yields $\widehat{\sigma}_{n,err}^2 - \widehat{\sigma}_{n,exp}^2 = 2a_0^2T^{-1}N_n + o_\proba(N_n)$, that is $N_n(\widehat{\sigma}_{n,err}^2 - \widehat{\sigma}_{n,exp}^2)^2 = 4a_0^4T^{-2}N_n^3 + o_\proba(N_n^3)$ so that, by Proposition \ref{propAVARlargenoise}, for any $i = 1,\cdots,5$ we have $S_i = N_n(\widehat{\sigma}_{n,err}^2 - \widehat{\sigma}_{n,exp}^2)^2/\widehat{V}_i \to^\proba +\infty $ since $\widehat{V}_i = O_\proba(N_n^2)$.
\end{proof}

\begin{proof}[Proof of Proposition \ref{propAVARphi0}]
When $\phi=0$, note that we are in the situation where the model (\ref{genmodel}) remains true with $\theta_0 = \widetilde{\theta}$. In particular, under the small noise assumption, we still have $\widehat{\theta}_{exp}  - \widetilde{\theta} = O_\proba(1/N_n)$ and thus all the calculations in the proof of Proposition \ref{propAVAR} remain true. In particular for $i=1,2$, $V_i = O_\proba(1)$. Under the the large noise alternative, a similar argument yields $V_i = O_\proba(N_n^2)$ as in the proof of Proposition \ref{propAVARlargenoise}.

\end{proof}
\begin{proof}[Proof of Corollary \ref{testConsistencyphi0}]
Again, when $\phi=0$, this is equivalent to assume that the model of log-returns (\ref{genmodel}) remains true with $\theta_0 = \widetilde{\theta}$. In particular, under the small noise assumption, Theorem \ref{thmConsistencyH0} remains valid replacing $\theta_0$ by $\widetilde{\theta}$, and in particular we still have $\widehat{\sigma}_{n,err}^2 - \widehat{\sigma}_{n,exp}^2 \to^\proba 2\widetilde{\eta}_0$. Combined with Proposition \ref{propAVARphi0}, this yields for $i=1,2$, $S_i \to^\proba +\infty$. For the large noise case, a similar argument yields $\widehat{\sigma}_{n,err}^2 - \widehat{\sigma}_{n,exp}^2 = 2T^{-1}a_0^2N_n + o_\proba(N_n)$ as in the proof of Proposition \ref{propAVARlargenoise} and thus, by Proposition \ref{propAVARphi0} we can conclude for $i=1,2$ that $S_i \to^\proba +\infty$.
\end{proof}

\begin{proof}[Proof of Lemma \ref{lemmaConsistencyPi}]
The fact that $\widehat{\pi}_V = \pi_V + o_\proba(1)$ in the fixed noise $a_0^2>0$ is a direct consequence of the consistency of $\widehat{\theta}_{n,err}$ and $\widehat{a}_{n,err}^2$ by Theorem \ref{theoremCLT} along with the mixing condition (\ref{assCorr}). Now we prove the second claim. For $i \in \{1,...,N_n\}$, we use the notation $\phi_i(\theta) = \phi(Q_i,\theta)$. By the mean value theorem, there exist $A_n \in [\widehat{a}_{n,err}^2, \eta_0T/n]$ and $B_n \in [(N_n+1)^{-1}\sum_{i=0}^{N_n} \phi_i( \widehat{\theta}_{n,err})^2,\esp[\phi_0(\theta_0)^2]]$ such that 
\beas 
\widehat{\pi}_V - \pi_V = \l( A_n +  B_n\r)^{-2}\l\{ -B_n \l(\widehat{a}_{n,err}^2 - \frac{\eta_0T}{n}\r) + A_n \l( \frac{\sum_{i=0}^{N_n} \phi_i(\widehat{\theta}_{n,err})^2}{N_n+1} - \esp[\phi_0(\theta_0)^2]  \r)   \r\}, 
\eeas 
and since by Theorem \ref{theoremnoerror} we have the estimate $\widehat{a}_{err}^2 = \eta_0T/n + O_\proba(N^{-3/2})$ and $\widehat{\theta}_{n,err} = \theta_0 + O_\proba(N_n^{-1})$, we easily deduce, again using the mixing condition (\ref{assCorr}), $(N_n+1)^{-1}\sum_{i=0}^{N_n} \phi_i( \widehat{\theta}_{n,err})^2 = \esp[\phi(\theta_0)^2] + O_\proba(N_n^{-1/2})$, and thus 
$ \l( A_n+B_n\r)^{-2} =  \esp[\phi_0(\theta_0)^2]^{-2} + o_\proba(1)$, $B_n \l(\widehat{a}_{n,err}^2 - \eta_0T/n\r) = O_\proba(N_n^{-3/2})$, and finally $A_n \l( (N+1)^{-1}\sum_{i=0}^{N_n} \phi_i( \widehat{\theta}_{n,err})^2 - \esp[\phi_0(\theta_0)^2]  \r) = O_\proba(N_n^{-3/2})$ so that 
$$\widehat{\pi}_V - \pi_V = O_\proba(N_n^{-3/2})$$ 
and we are done.

\end{proof}
\bibliography{biblio}
\bibliographystyle{apalike} 

\newpage

\begin{table}
\centering
\caption{Overview of models and limit order book variables$^\dag$}
\label{tableoverview}
\begin{tabular}{lll}
\toprule
\toprule
\multicolumn{3}{l}{\emph{Limit order book variables}}\\
\bottomrule
Name & Symbol & Definition \\
\toprule
trade type  & $I_i$ & 1 if the trade at time
$t_i$
is buyer-initiated and -1 if seller-initiated \\
trading volume & $V_i$ & number of shares traded at $t_i$  \\
duration time  & $D_i$ & $D_i = t_i - t_{i-1}$   \\
quoted depth & $QD_i$ & The ask/bid depth specifies the volume available at the best ask/bid  \\
bid-ask spread & $S_i$   & $S_i = A_i - B_i$, with $A_i$ best ask price and $B_i$ best bid price  \\
order flow imbalance & $OFI_i$ & Imbalance between supply and demand at the best bid and ask prices\\ & & (including both quotes and cancellations) between $t_{i-1}$ and $t_i$    \\
\toprule
\emph{Models} & &\\
\bottomrule
Name & $\phi(Q_i,\theta_0)$ & Related literature\\
\bottomrule
null &  0 &   \\
Roll & $I_i \theta_0$ & \cite{roll1984simple}  \\
Glosten-Harris & $I_i (\theta_0^{(1)} + V_i \theta_0^{(2)})$ & \cite{glosten1988estimating} \\
signed timestamp & $I_i D_i^{-1} \theta_0$ & \cite{almgren2001optimal}  \\
signed spread & $\frac{1}{2} I_i S_i \theta_0$ &  \\
signed quoted depth & $I_i QD_i \theta_0$ &  \cite{kavajecz1999specialist} \\
order flow imbalance & $OFI_i \theta_0$ & \cite{cont2014price}  \\
NL signed spread & $I_i \frac{S_i \theta_0}{1 + S_i \theta_0}$ &  \\
general & $I_i (\theta_0^{(1)} + V_i \theta_0^{(2)}$ & \\
 & $+ D_i^{-1} \theta_0^{(3)} + S_i \theta_0^{(4)}$ &\\
  & $+ QD_i \theta_0^{(5)})$ &\\
  & $+ OFI_i \theta_0^{(6)}$ &\\
\bottomrule
\end{tabular}

\end{table}

\begin{table}
\centering
\caption{Simulation study results: fraction of rejections of the null hypothesis when the true scenario is the null hypothesis $\mathcal{H}_0$ or the alternative with i.i.d residual noise $\mathcal{H}_1$ at the 0.05 level. Note that $\xi^2 = a_0^2/\sqrt{T\int_0^T \sigma_u^4 du} $ corresponds to the noise-to-signal ratio.}
\label{tableSimu}
\begin{tabular}{lcccccccccc}
\toprule
\toprule
\multicolumn{3}{l}{\emph{sampling frequency}} & \multicolumn{2}{l}{\emph{tick by tick}} & \multicolumn{3}{l}{\emph{15 seconds}} & \multicolumn{3}{l}{\emph{30 seconds}}\\
& $a_0^2$ & $\xi^2$ & $S_1$ & $S_2$ & $S_3$ & $S_4$ & $S_5$ & $S_3$ & $S_4$ & $S_5$\\
\toprule
& & & & \multicolumn{3}{l}{Roll model}\\
\toprule
& & & \multicolumn{4}{c}{constant volatility}\\
\toprule
$\mathcal{H}_0$ & 0 & 0 & 0.05 & 0.07 &  0.05& 0.05& 0.05& 0.04& 0.04& 0.04\\
 & $10^{-9}$ & $\approx 10^{-6}$ & 1.00 & 1.00 &  0.05& 0.05& 0.05& 0.03& 0.03& 0.03\\
$\mathcal{H}_1$ & $10^{-8}$  & $\approx 10^{-5}$ & 1.00 & 1.00 & 0.33& 0.33& 0.34& 0.08& 0.08& 0.08\\
 & $10^{-7}$  &$\approx 10^{-4}$ & 1.00 & 1.00 & 1.00& 1.00& 1.00& 0.98& 0.99& 0.99\\
 \toprule
 & \multicolumn{9}{c}{time-varying volatility and no price jump}\\
\toprule
$\mathcal{H}_0$ & 0 & 0&0.05 & 0.07& 0.18& 0.05& 0.09& 0.16& 0.06& 0.09\\
 & $10^{-9}$ & $\approx 10^{-6}$ &1.00 & 1.00 & 0.18& 0.07& 0.11& 0.17& 0.05& 0.10\\
$\mathcal{H}_1$ & $10^{-8}$  & $\approx 10^{-5}$ & 1.00 & 1.00 & 0.63& 0.52& 0.55& 0.25& 0.15& 0.19\\
 & $10^{-7}$  & $\approx 10^{-4}$ & 1.00 & 1.00 & 1.00& 1.00& 1.00& 0.98& 0.97& 0.97\\
 \toprule
  & \multicolumn{9}{c}{time-varying volatility and price jump}\\
\toprule
$\mathcal{H}_0$ & 0 & 0& 0.03 & 0.05 & 0.09& 0.01& 0.05& 0.08& 0.03& 0.06\\
 & $10^{-9}$ & $\approx 10^{-6}$ & 1.00 & 1.00 & 0.09& 0.03& 0.06& 0.08& 0.03& 0.07\\
$\mathcal{H}_1$ & $10^{-8}$  & $\approx 10^{-5}$ & 1.00 & 1.00 & 0.30& 0.22& 0.29& 0.10& 0.05& 0.09\\
 & $10^{-7}$  & $\approx 10^{-4}$ & 1.00 & 1.00 & 0.99 & 0.41 & 0.88 & 0.75& 0.35& 0.57\\
 \toprule
& & &  \multicolumn{4}{c}{signed spread model}\\
\toprule
& & & \multicolumn{4}{c}{constant volatility}\\
\toprule
$\mathcal{H}_0$ & 0 & 0& 0.05 & 0.08 & 0.05& 0.05& 0.05& 0.04& 0.04& 0.04\\
 & $10^{-9}$ & $\approx 10^{-6}$ & 1.00 & 1.00 & 0.05 & 0.06& 0.06& 0.03& 0.04& 0.04\\
$\mathcal{H}_1$ & $10^{-8}$ & $\approx 10^{-5}$ & 1.00 & 1.00 & 0.30& 0.30& 0.30& 0.07& 0.07& 0.07\\
 & $10^{-7}$  & $\approx 10^{-4}$ & 1.00 & 1.00 & 1.00& 1.00& 1.00& 0.98& 0.98& 0.98\\
 \toprule
 & \multicolumn{9}{c}{time-varying volatility and no price jump}\\
\toprule
$\mathcal{H}_0$ & 0 & 0& 0.05 & 0.08 & 0.14& 0.04& 0.07& 0.16& 0.05& 0.08\\
 & $10^{-9}$ & $\approx 10^{-6}$ & 1.00 & 1.00 & 0.17& 0.06& 0.09& 0.16& 0.06& 0.09\\
$\mathcal{H}_1$ & $10^{-8}$  & $\approx 10^{-5}$ & 1.00 & 1.00 & 0.62& 0.52& 0.53& 0.26& 0.15& 0.18\\
 & $10^{-7}$  & $\approx 10^{-4}$ & 1.00 & 1.00 & 1.00& 1.00& 1.00& 0.98& 0.97& 0.97\\
 \toprule
  & \multicolumn{9}{c}{time-varying volatility and price jump}\\
\toprule
$\mathcal{H}_0$ & 0 & 0& 0.04 & 0.05 & 0.06& 0.01& 0.04& 0.08& 0.02& 0.06\\
 & $10^{-9}$ & $\approx 10^{-6}$ & 1.00 & 1.00 & 0.07& 0.03& 0.05& 0.08& 0.02& 0.05\\
$\mathcal{H}_1$ & $10^{-8}$  & $\approx 10^{-5}$ & 1.00 & 1.00 & 0.28& 0.21& 0.26& 0.13& 0.07& 0.11\\
 & $10^{-7}$  & $\approx 10^{-4}$ & 1.00 & 1.00 & 1.00& 0.38& 0.87& 0.77& 0.37& 0.59\\
\bottomrule
\end{tabular}
\end{table}

\begin{table}
\centering
\caption{Simulation study results: fraction of rejections of the null hypothesis when the true scenario is the alternative with  serially dependent, endogenous and heteroskedastic residual noise $\mathcal{H}_2$ at the 0.05 level. Note that $\xi^2 = a_0^2/\sqrt{T\int_0^T \sigma_u^4 du} $ corresponds to the noise-to-signal ratio.}
\label{tableSimuGen}
\begin{tabular}{lcccccccccc}
\toprule
\toprule
\multicolumn{3}{l}{\emph{sampling frequency}} & \multicolumn{2}{l}{\emph{tick by tick}} & \multicolumn{3}{l}{\emph{15 seconds}} & \multicolumn{3}{l}{\emph{30 seconds}}\\
& $a_0^2$ & $\xi^2$ & $S_1$ & $S_2$ & $S_3$ & $S_4$ & $S_5$ & $S_3$ & $S_4$ & $S_5$\\
\toprule
& & & & \multicolumn{3}{l}{Roll model}\\
\toprule
& & & \multicolumn{4}{c}{constant volatility}\\
\toprule
 & $10^{-9}$ & $\approx 10^{-6}$ & 1.00 & 1.00 &  0.07& 0.07& 0.07& 0.06& 0.07& 0.07\\
$\mathcal{H}_2$ & $10^{-8}$  & $\approx 10^{-5}$ & 1.00 & 1.00 & 0.43& 0.43& 0.43& 0.14& 0.14& 0.14\\
 & $10^{-7}$  &$\approx 10^{-4}$ & 1.00 & 1.00 & 1.00& 1.00& 1.00& 0.99& 0.99& 0.99\\
 \toprule
 & \multicolumn{9}{c}{time-varying volatility and no price jump}\\
\toprule
 & $10^{-9}$ & $\approx 10^{-6}$ &1.00 & 1.00 & 0.22& 0.09& 0.13& 0.21& 0.08& 0.14\\
$\mathcal{H}_2$ & $10^{-8}$  & $\approx 10^{-5}$ & 1.00 & 1.00 & 0.75& 0.63& 0.56& 0.30& 0.19& 0.22\\
 & $10^{-7}$  & $\approx 10^{-4}$ & 1.00 & 1.00 & 1.00& 1.00& 1.00& 0.99& 0.99& 0.99\\
 \toprule
  & \multicolumn{9}{c}{time-varying volatility and price jump}\\
\toprule
 & $10^{-9}$ & $\approx 10^{-6}$ & 1.00 & 1.00 & 0.12& 0.05& 0.08& 0.10& 0.05& 0.08\\
$\mathcal{H}_2$ & $10^{-8}$  & $\approx 10^{-5}$ & 1.00 & 1.00 & 0.34& 0.25& 0.33& 0.12& 0.07& 0.10\\
 & $10^{-7}$  & $\approx 10^{-4}$ & 1.00 & 1.00 & 1.00 & 0.50 & 0.99 & 0.83& 0.41& 0.63\\
 \toprule
& & &  \multicolumn{4}{c}{signed spread model}\\
\toprule
& & & \multicolumn{4}{c}{constant volatility}\\
\toprule
 & $10^{-9}$ & $\approx 10^{-6}$ & 1.00 & 1.00 & 0.07 & 0.08 & 0.08 & 0.07& 0.07& 0.07\\
$\mathcal{H}_2$ & $10^{-8}$ & $\approx 10^{-5}$ & 1.00 & 1.00 & 0.40 & 0.40 & 0.40 & 0.14& 0.15& 0.15\\
 & $10^{-7}$  & $\approx 10^{-4}$ & 1.00 & 1.00 & 1.00 & 1.00 & 1.00 & 0.99& 0.99& 0.99\\
 \toprule
 & \multicolumn{9}{c}{time-varying volatility and no price jump}\\
\toprule
 & $10^{-9}$ & $\approx 10^{-6}$ & 1.00 & 1.00 & 0.23 & 0.11 & 0.15 & 0.21& 0.10& 0.14\\
$\mathcal{H}_2$ & $10^{-8}$ & $\approx 10^{-5}$ & 1.00 & 1.00 & 0.73 & 0.62 & 0.63 & 0.32& 0.20& 0.24\\
 & $10^{-7}$  & $\approx 10^{-4}$ & 1.00 & 1.00 & 1.00 & 1.00 & 1.00 & 0.99& 0.99& 0.99\\
 \toprule
  & \multicolumn{9}{c}{time-varying volatility and price jump}\\
\toprule
 & $10^{-9}$ & $\approx 10^{-6}$ & 1.00 & 1.00 & 0.10 & 0.05 & 0.07 & 0.11& 0.04& 0.07\\
$\mathcal{H}_2$ & $10^{-8}$  & $\approx 10^{-5}$ & 1.00 & 1.00 & 0.34 & 0.26 & 0.31 & 0.16& 0.10& 0.15\\
 & $10^{-7}$  & $\approx 10^{-4}$ & 1.00 & 1.00 & 1.00 & 0.46 & 0.99 & 0.95& 0.50& 0.72\\
\bottomrule
\end{tabular}
\end{table}

\begin{table}
\centering
\caption{Simulation study results: comparison of eight volatility estimators. Note that $\xi^2 = a_0^2/\sqrt{T\int_0^T \sigma_u^4 du} $ corresponds to the noise-to-signal ratio, and that $e^x$ stands for $10^x$.}
\label{tableSimuCompa}
\begin{tabular}{lccccccccc}
\toprule
\toprule
$a_0^2$ & \multicolumn{3}{c}{0} & \multicolumn{3}{c}{$e^{-9}$} & \multicolumn{3}{c}{mix}\\
 \toprule
 $\xi^2$ & \multicolumn{3}{c}{0} & \multicolumn{3}{c}{$ \approx e^{-6}$} & \multicolumn{3}{c}{mix}\\
 \toprule
Estimator & Bias & Stdv. & RMSE & Bias & Stdv. & RMSE & Bias & Stdv. & RMSE \\
\toprule
\multicolumn{10}{c}{time-varying volatility and no price jump, \emph{tick by tick}}\\
\toprule
\multicolumn{10}{c}{Roll model}\\
\toprule
S & $1.71 e^{-7}$ & $3.89 e^{-6}$ & $3.89 e^{-6}$ & $1.33 e^{-7}$ & $6.30 e^{-6}$ & $6.30 e^{-6}$ & $1.52 e^{-7}$ & $5.24 e^{-6}$ & $5.24 e^{-6}$\\
QMLEexp & $- 5.65 e^{-8}$ & $3.10 e^{-6}$ & $3.10 e^{-6}$ & $4.67 e^{-5}$ & $3.48 e^{-6}$ & $4.68 e^{-5}$ & $2.38 e^{-5}$ & $2.31 e^{-5}$ & $3.32 e^{-5}$\\
QMLEerr & $7.60 e^{-8}$ & $5.93 e^{-6}$ & $5.93 e^{-6}$ & $1.02 e^{-7}$ & $6.12 e^{-6}$ &  $6.12 e^{-6}$ & $8.90 e^{-8}$ & $6.03 e^{-6}$ & $6.03 e^{-6}$\\
EQMLE & $7.61 e^{-8}$ & $5.93 e^{-6}$ & $5.93 e^{-6}$ &  $1.04 e^{-7}$ &  $6.12 e^{-6}$ &  $6.12 e^{-6}$ & $8.91 e^{-8}$ & $6.03 e^{-6}$ & $6.03 e^{-6}$\\
QMLE & $6.13 e^{-5}$ & $1.50 e^{-5}$ & $6.31 e^{-5}$ &  $5.76 e^{-5}$ & $1.53 e^{-5}$ &  $5.96 e^{-5}$ & $5.94 e^{-5}$ & $1.55 e^{-5}$ & $6.14 e^{-5}$\\
PAE & $- 2.88 e^{-6}$ & $ 2.99 e^{-5}$ & $ 3.00 e^{-5}$ & $- 2.88 e^{-6}$ & $ 2.99 e^{-5}$ & $ 3.00 e^{-5}$ & $- 2.88 e^{-6}$ & $ 2.99 e^{-5}$ & $ 3.00 e^{-5}$\\
RK & $6.10 e^{-5}$ & $1.40 e^{-5}$& $6.25 e^{-5}$ & $ 5.86 e^{-5}$ & $ 1.46 e^{-5}$ &  $ 6.04 e^{-5}$ & $ 5.98 e^{-5}$ & $ 1.39 e^{-5}$ & $ 6.14 e^{-5}$\\
RV & $ 3.27 e^{-4}$ &  $5.87 e^{-6}$& $ 3.27 e^{-4}$ & $ 3.74 e^{-4}$ & $ 6.48 e^{-6}$ & $ 3.74 e^{-4}$ & $ 3.51 e^{-4}$ & $ 6.33 e^{-6}$ & $ 3.51 e^{-4}$\\
\toprule
\multicolumn{10}{c}{spread model}\\
\toprule
S &$1.71 e^{-7}$ & $3.89 e^{-6}$ & $3.89 e^{-6}$ &  $1.33 e^{-7}$ & $6.30 e^{-6}$ & $6.30 e^{-6}$ & $1.52 e^{-7}$ & $5.24 e^{-6}$ & $5.24 e^{-6}$\\
QMLEexp & $- 5.66 e^{-8}$ & $3.11 e^{-6}$ & $3.11 e^{-6}$ & $ 4.67 e^{-5}$ & $ 3.48 e^{-6}$ & $ 4.68 e^{-5}$ & $2.38 e^{-5}$ & $2.31 e^{-5}$ & $3.32 e^{-5}$\\
QMLEerr & $7.62 e^{-8}$ & $5.93 e^{-6}$ & $5.93 e^{-6}$ & $ 1.02 e^{-7}$ & $ 6.12 e^{-6}$ & $ 6.12 e^{-6}$ & $8.91 e^{-8}$ & $6.03 e^{-6}$ & $6.03 e^{-6}$\\
EQMLE & $7.62 e^{-8}$ & $5.93 e^{-6}$ & $5.93 e^{-6}$ & $ 1.04 e^{-7}$ & $ 6.11 e^{-6}$ & $6.11 e^{-6}$ & $8.92 e^{-8}$ & $6.03 e^{-6}$ & $6.03 e^{-6}$\\
QMLE & $6.15 e^{-5}$ & $1.51 e^{-5}$ & $6.33 e^{-5}$ & $5.78 e^{-5}$ & $1.53 e^{-5}$ & $5.98 e^{-5}$ & $5.96 e^{-5}$ & $1.56 e^{-5}$ & $6.16 e^{-5}$\\
PAE &$- 2.87 e^{-6}$ & $ 2.99 e^{-5}$ & $ 3.00 e^{-5}$ & $-2.87 e^{-6}$ & $2.99 e^{-5}$ & $3.00 e^{-5}$ & $- 2.87 e^{-6}$ & $ 2.99 e^{-5}$ & $ 3.00 e^{-5}$\\
RK & $6.11 e^{-5}$ & $1.41 e^{-5}$& $6.27 e^{-5}$ & $5.88 e^{-5}$ & $1.47 e^{-5}$ & $6.06 e^{-5}$ & $ 6.00 e^{-5}$ & $ 1.44 e^{-5}$ & $ 6.17 e^{-5}$\\
RV & $ 3.31 e^{-4}$ &  $5.95 e^{-6}$& $ 3.31 e^{-4}$ & $ 3.78 e^{-4}$ & $ 6.51 e^{-6}$ & $ 3.78 e^{-4}$ & $ 3.55 e^{-4}$ & $ 6.38 e^{-6}$ & $ 3.55 e^{-4}$ \\
\bottomrule
\end{tabular}
\end{table}

\begin{table}
\centering
\caption{Descriptive statistics}
\label{descriptivestatistics}
\begin{tabular}{llccccc}
\toprule
\toprule
Name & Ticker &  Nb. Trades & Price & Tick & Price/Tick & Spread \\
 & & (daily) & & & ($\times 10^{-3}$)& (in ticks) \\
\toprule 
Accor &  ACCP & 1,900 & 31.87 & 0.005 & 6.37 & 3.04\\
Air Liquide &  AIRP &  4,116 & 94.64 & 0.01 & 9.46 & 2.53\\
Alstom &  ALSO & 2,818 & 42.55 & 0.005 & 8.51 & 2.75\\
Alcatel &  ALUA &  3,667 & 4.13 & 0.001 & 4.13& 1.52\\
AXA &  AXAF &  3,156 & 15.03 & 0.005 & 3.01& 1.25\\
Bouygues & BOUY & 2,275 & 34.13 & 0.005 & 6.83& 1.80\\
Credit Agricole &  CAGR & 3,014 & 11.74 & 0.005 & 2.35 & 1.17\\
Cap Gemini &  CAPP & 2,002 & 40.83 & 0.005 & 8.17 &2.04\\
Danone &  DANO & 3,525 & 46.30 & 0.005 & 9.26 & 2.07\\
Airbus &  EAD & 1,589 & 20.95 & 0.005 & 4.19 & 2.01\\
Essilor International &  ESSI & 1,090 & 52.80 & 0.01 & 5.28 & 2.43\\
France Telecom & FTE & 2,327 & 15.76 & 0.005 & 3.16 & 1.16\\
Engie & GSZ & 3,513 & 28.26 & 0.005 & 5.65& 1.68\\
Lafarge & LAFP & 3,344 & 44.36 & 0.005 & 8.87 & 3.58\\
Louis Vuitton & LVMH & 2,073 & 112.51 & 0.05 & 2.25 & 1.15\\
L'Oreal & OREP & 2,582 & 82.99 & 0.01 & 8.30 &2.55\\
Peugeot & PEUP & 2,129 & 28.17 & 0.005 & 5.63 & 2.77\\
Kering & PRTP & 961 & 109.15 & 0.05 & 2.18 & 1.22\\
Publicis Groupe & PUBP & 2,412 & 40.14 & 0.005 & 8.03 & 1.98\\
Renault & RENA & 3,728 & 39.43 & 0.005 & 7.89 & 2.78\\
Sanofi & SASY & 3,622 & 50.14 & 0.01 & 5.01 & 1.35\\
Schneider Electric & SCHN & 2,638 & 121.66 & 0.05 & 2.43 &1.15\\
Suez & SEVI & 1,115 & 14.52 & 0.005 & 2.90 &1.57\\
Vinci & SGEF & 3,133 & 44.44 & 0.005 & 8.89 & 2.76\\
Societe Generale & SOGN & 6,710 & 46.40 & 0.005 & 9.28 & 3.14 \\
STMicroelectronics & STM & 2,921 & 8.72 & 0.001 & 8.72 &3.15\\
Technip & TECF & 3,192 & 75.78 & 0.01 & 7.58 & 2.13\\
Total & TOTF & 5,909 & 43.09 & 0.005 & 8.62 & 1.80\\
Unibail Rodamco & UNBP & 1,440 & 153.95 & 0.05 & 3.08 & 1.28\\
Veolia Environnement & VIE & 1,724 & 21.97 & 0.005 & 4.39 & 1.51\\
Vivendi & VIV & 2,710 & 20.35 & 0.005 & 4.07 & 1.26\\
\bottomrule
\end{tabular}

\end{table}

\begin{table}
\centering
\caption{Measure of goodness of fit of several leading models$^\dag$}
\label{tableerror}
\begin{tabular}{lccc}
\toprule
\toprule
Model for the & MMN & residual noise & Proportion of  \\
explicative part &variance & variance & variance explained\\
\toprule
\multicolumn{4}{l}{\emph{sampling frequency: tick by tick}}\\
\bottomrule
null & $1.63 \times 10^{-8}$ & $1.63 \times 10^{-8}$ & 0.00 \%\\
Roll & $2.27 \times 10^{-8}$ & $3.89 \times 10^{-9}$  & 83.84 \%  \\
Glosten-Harris & $2.28 \times 10^{-8}$ & $3.40 \times 10^{-9}$  & 85.97 \%  \\
signed timestamp & $1.82 \times 10^{-8}$ & $1.45 \times 10^{-8}$  & 20.98 \%  \\
signed spread & $2.78 \times 10^{-8}$ & $2.45 \times 10^{-10}$  & 99.19 \%  \\
signed quote depth & $2.02 \times 10^{-8}$ & $7.95 \times 10^{-9}$ & 59.71 \%  \\
order flow imbalance & $1.61 \times 10^{-8}$ & $1.55 \times 10^{-8}$ & 3.86 \%  \\
NL signed spread & $2.77 \times 10^{-8}$ & $2.44 \times 10^{-10}$  & 99.17 \%  \\
general & $2.79 \times 10^{-8}$ & $1.88 \times 10^{-10}$  & 99.36 \%  \\
\toprule
\multicolumn{4}{l}{\emph{sampling frequency $\approx$ 15 seconds}}\\
\bottomrule
null & $1.51 \times 10^{-8}$ & $1.51 \times 10^{-8}$ & 0.00 \%\\
Roll & $2.08 \times 10^{-8}$ & $2.67 \times 10^{-9}$  & 88.67 \%  \\
Glosten-Harris & $2.11 \times 10^{-8}$ & $2.35 \times 10^{-9}$  & 90.09 \%  \\
signed timestamp & $1.58 \times 10^{-8}$ & $1.39 \times 10^{-8}$  & 15.27 \%  \\
signed spread & $2.61 \times 10^{-8}$ & $2.84 \times 10^{-10}$  & 99.04 \%  \\
signed quote depth & $1.87 \times 10^{-8}$ & $7.02 \times 10^{-9}$  & 64.68 \%  \\
order flow imbalance & $1.50 \times 10^{-8}$ & $1.44 \times 10^{-8}$ & 4.11 \%  \\
NL signed spread & $2.61 \times 10^{-8}$ & $2.84 \times 10^{-10}$  & 99.04 \%  \\
general & $2.63 \times 10^{-8}$ & $2.13 \times 10^{-10}$  & 99.22 \%  \\
\toprule
\multicolumn{4}{l}{\emph{sampling frequency $\approx$ 30 seconds}}\\
\bottomrule
null & $1.24 \times 10^{-8}$ & $1.24 \times 10^{-8}$ & 0.00 \% \\
Roll & $1.98 \times 10^{-8}$ & $3.87 \times 10^{-9}$  & 92.17 \%  \\
Glosten-Harris & $2.03 \times 10^{-8}$ & $1.63 \times 10^{-9}$  & 93.34 \%  \\
signed timestamp & $1.33 \times 10^{-8}$ & $1.20 \times 10^{-8}$  & 21.80 \%  \\
signed spread & $2.46 \times 10^{-8}$ & $3.99 \times 10^{-10}$  & 98.01 \%  \\
signed quote depth & $1.71 \times 10^{-8}$ & $5.66 \times 10^{-9}$ & 71.77 \%  \\
order flow imbalance & $1.19 \times 10^{-8}$ & $1.12 \times 10^{-8}$ & 5.63 \%  \\
NL signed spread & $2.45 \times 10^{-8}$ & $3.98 \times 10^{-10}$  & 98.00 \%  \\
general & $2.47 \times 10^{-8}$ & $3.75 \times 10^{-10}$  & 98.41 \%  \\
\bottomrule
\end{tabular}

\scriptsize $^\dag$The MMN variance, the residual noise variance and the proportion of variance explained are estimated using $\widehat{\pi}_V$ and averaged across the thirty one constituents of the CAC 40 and the 19 days in April 2011. 

\end{table}

\begin{table}
\centering
\caption{Fraction of rejections of the null hypothesis at the 0.05 level for the signed spread model$^\dag$}
\label{tablenewtests}
\begin{tabular}{lccccc}
\toprule
\toprule
Ticker & Est. Par. & Var. Ex. & Aver. & $S_1$ & $S_2$ \\
\multicolumn{6}{l}{\emph{Sampling frequency: tick by tick}}\\
\toprule
ACCP & 0.84& 99.78 \%& 0.00 & 0.00 & 0.00\\
AIRP & 0.79& 100.00 \%& 0.00 & 0.00 & 0.00\\
ALSO & 0.84& 99.69 \%& 0.00 & 0.00 & 0.00\\
ALUA & 0.82 & 100.00 \%& 0.05 & 0.05 & 0.05\\
AXAF & 0.73& 99.63 \%& 0.00 & 0.00 & 0.00\\
BOUY & 0.79 &100.00 \%& 0.05 & 0.05 & 0.05 \\
CAGR & 0.74  & 99.70 \%& 0.05 & 0.05 & 0.05\\
CAPP & 0.85  & 100.00 \%& 0.11 & 0.11 & 0.11\\
DANO & 0.82& 100.00 \%& 0.11 & 0.11 & 0.11\\
EAD & 0.81 & 99.96 \%& 0.03 & 0.00 & 0.05\\
ESSI & 0.82 & 99.89 \%& 0.00 & 0.00 & 0.00\\
FTE & 0.74& 93.92 \%& 0.16 & 0.16 & 0.16 \\
GSZ & 0.80 & 99.56 \%& 0.00 & 0.00 & 0.00 \\
LAFP & 0.82 & 100.00 \% & 0.00 & 0.00 & 0.00\\
LVMH & 0.72 & 94.45 \%& 0.11 & 0.11 & 0.11\\
OREP & 0.80 & 99.90 \%& 0.03 & 0.00 & 0.05\\
PEUP & 0.84 & 99.48 \%& 0.00 & 0.00 & 0.00\\
PRTP & 0.72& 98.61 \%& 0.00 & 0.00 & 0.00\\
PUBP & 0.82&  100.00 \%& 0.00 & 0.00 & 0.00 \\
RENA & 0.85& 99.89 \%&0.00 & 0.00 & 0.00 \\
SASY & 0.75 & 99.72 \% & 0.00 & 0.00 & 0.00\\
SCHN & 0.68 & 92.47 \% & 0.42 & 0.42 & 0.42\\
SEVI & 0.75 & 99.97 \% & 0.00 & 0.00 & 0.00 \\
SGEF & 0.85 & 99.75 \% & 0.00 & 0.00 & 0.00 \\
SOGN & 0.87 & 99.87 \% & 0.00 & 0.00 & 0.00 \\
STM & 0.87 & 100.00 \% & 0.00 & 0.00 &0.00 \\
TECF & 0.75& 99.99 \% & 0.05 & 0.05 & 0.05 \\
TOTF & 0.78 & 100.00 \% & 0.00 & 0.00 & 0.00 \\
UNBP & 0.69 & 99.23 \% & 0.00 & 0.00 & 0.00\\
VIE & 0.77 & 99.95 \% & 0.00 &0.00 & 0.00\\
VIV & 0.72 & 99.63 \%& 0.00 & 0.00 & 0.00\\
\toprule
Average & 0.79 & 99.19 \% & 0.06 & 0.06 & 0.06\\
\bottomrule
\end{tabular}

\scriptsize $^\dag$The values are averaged across the 19 days in April 2011. The column "Est. Par." stands for the estimated parameter and "Var. Ex." for "Proportion of variance explained". The column "Aver." corresponds to the fraction of rejections averaged across the two tests computed using tick by tick data.

\end{table}

\begin{table}
\centering
\caption{Fraction of rejections for the estimated efficient price related to the signed spread model at the 0.05 level$^\dag$}
\label{tableaittests}
\resizebox{0.55\textwidth}{!}{
\begin{tabular}{lccccccc}
\toprule
\toprule
Price & Aver. & $H_1$ & $H_2$ & $H_3$ & $T$ & $AC$ & $H_4$   \\
\toprule
\multicolumn{7}{l}{Main group}\\
\toprule
Est. price & 0.06 & 0.08 & 0.05 & 0.05 & 0.05 & 0.09 & 0.04 \\
\toprule
\multicolumn{7}{l}{Minor group}\\
\toprule
Est. price & 0.38 & 0.42 & 0.35 & 0.39 & 0.39 & 0.39 & 0.33\\
\toprule
\multicolumn{7}{l}{Overall}\\
\toprule
Est. price & 0.09 & 0.11 & 0.07 & 0.08 & 0.08 & 0.12 & 0.07\\
Mid price & 0.25 & 0.31 & 0.21 & 0.25 & 0.25 & 0.29 & 0.19\\
Obs. price & 1.00 & 1.00 & 1.00 & 1.00 & 1.00 & 1.00 & 0.99\\
\bottomrule
\end{tabular}}/

\scriptsize $^\dag$The values are averaged across the 19 days in April 2011. The column "Aver." corresponds to the fraction of rejections averaged across the six tests. The six tests from \cite{ait2016hausman} are implemented at the tick by tick frequency considering the estimated efficient price as the given observed price to be tested. We report the proportion of rejections for each constituent when considering the estimated price. We also report in "Obs. price" (respectively "Mid price") the proportion of rejections when considering the observed price (resp. the mid price). Finally, note that the main group corresponds to stocks a priori free of residual noise given our implemented tests, whereas the minor group stands for stocks with residual noise. 

\end{table}

\newpage

\begin{figure}
\includegraphics[width=1.2\linewidth]{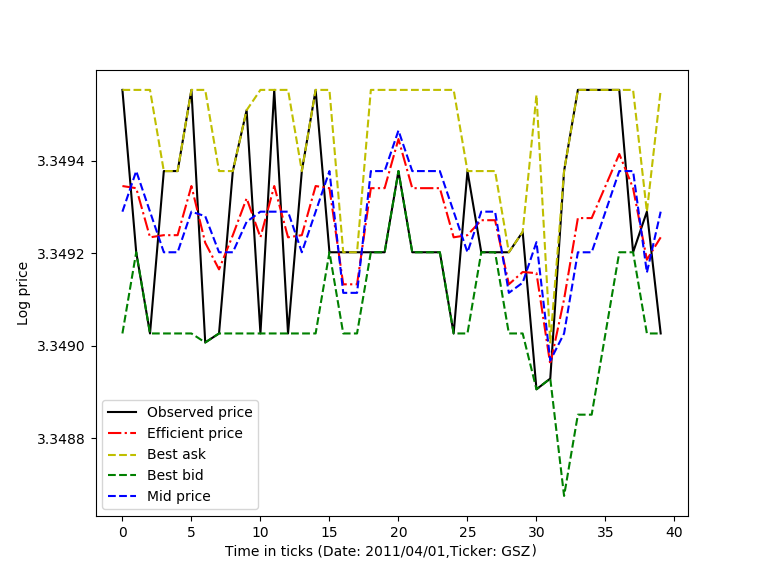}
\centering
\caption{An example of estimated efficient price related to the signed spread model.}
\label{figure_estimatedefficient}
\end{figure}

\begin{figure}
\includegraphics[width=1.2\linewidth]{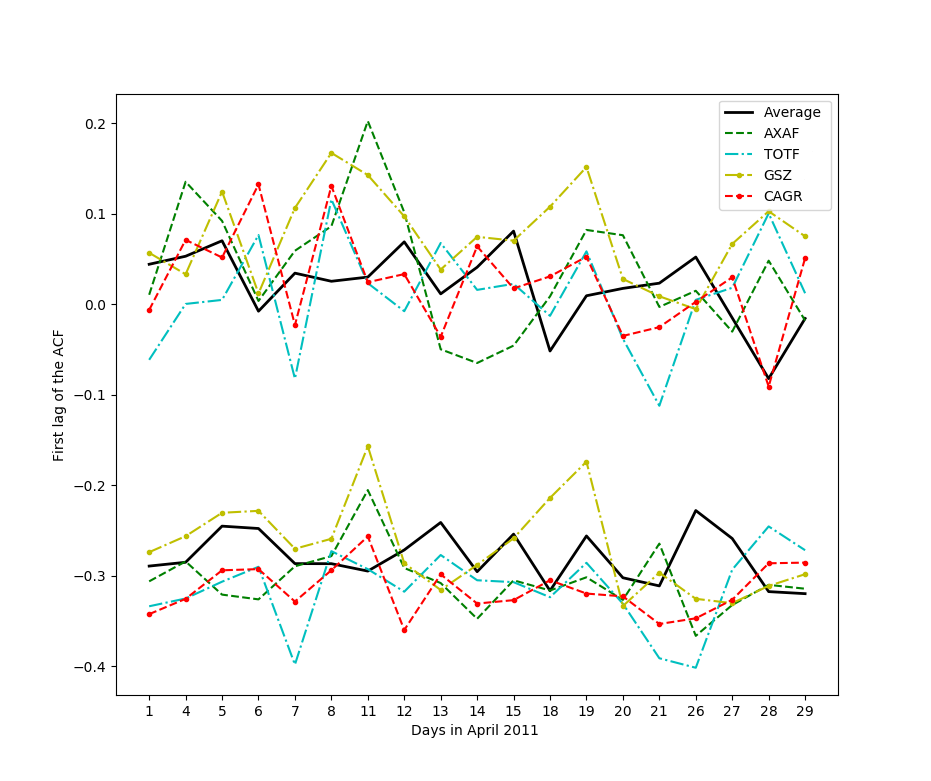}
\centering
\caption{First lag of the autocorrelation function over time related to the estimated efficient price with the signed spread model (above) and the observed price (below) for the across-the-assets average and 4 random assets.}
\label{figure_ACF}
\end{figure}

\begin{figure}
\includegraphics[width=1.2\linewidth]{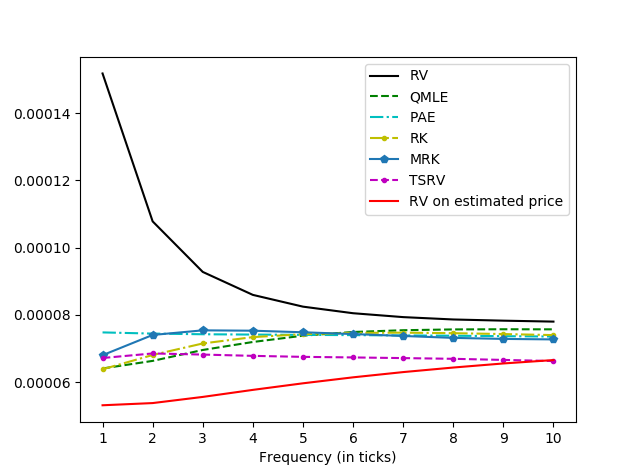}
\centering
\caption{Average daily signature plot for the CAC 40 constituents. Seven volatility estimators are implemented: RV, QMLE, PAE, RK, MRK (i.e. a modified realized kernel-based estimator robust to autocorrelated-noise considered in \cite{varneskov2017estimating}), TSRV and RV on the estimated efficient price related to the signed spread model.}
\label{figure_signatureplot}
\end{figure}

\begin{figure}
\includegraphics[width=1.2\linewidth]{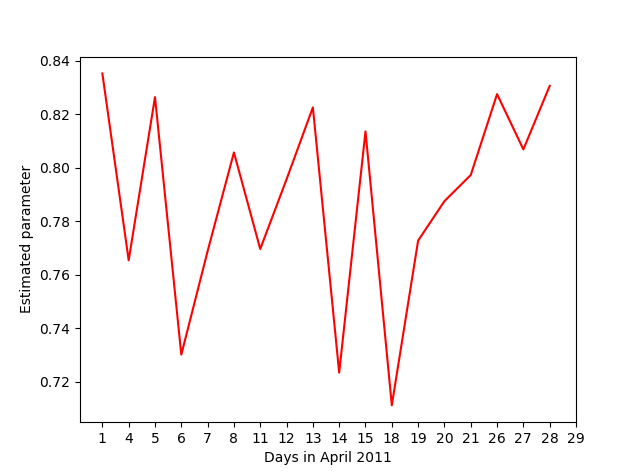}
\centering
\caption{Estimated parameter related to the signed spread model over time average across the CAC 40 constituents.}
\label{figure_estimatedparameter}
\end{figure}

\end{document}